\documentclass[11pt]{article}
\usepackage{amsmath}
\usepackage{amsthm}
\usepackage{amssymb}
\usepackage{amsfonts}
\usepackage{forest}
\usepackage{mathrsfs}
\usepackage{bbm}
\usepackage{bbold}
\usepackage{setspace}
\usepackage{thmtools}
\usepackage{thm-restate}
\usepackage{fullpage}
\usepackage{tcolorbox}
\usepackage[all]{xy}
\usepackage[title, titletoc]{appendix}
\usepackage[ruled,linesnumbered]{algorithm2e}
\usepackage{mathalfa}
\usepackage{amsbsy}
\usepackage{comment}
\usepackage{tikz}
\usetikzlibrary{positioning}
\pgfdeclarelayer{bg}    
\pgfsetlayers{bg,main}

\newtheorem{theorem}{Theorem}[section]
\newtheorem{definition}[theorem]{Definition}
\newtheorem{lemma}[theorem]{Lemma}
\newtheorem{fact}[theorem]{Fact}
\newtheorem{example}[theorem]{Example}
\newtheorem{corollary}[theorem]{Corollary}
\newtheorem{notation}[theorem]{Notation}
\newtheorem{remark}[theorem]{Remark}
\newtheorem*{notation*}{Notation}
\newtheorem*{lemma*}{Lemma}
\newtheorem*{note*}{Note}

\DeclareMathOperator*{\argmin}{arg\,min}

\newcommand{\alignset}[2]{\left\{ #1 \vphantom{#2} \suchthat \right. & \left. \vphantom{#1} #2 \right\}}
\newcommand{\multilinenorm}[2]{\left\| #1 \vphantom{#2} \right. \\ & \quad \left. \vphantom{#1} #2 \right\|}

\newcommand{\parens}[1]{\left( #1 \right)}
\newcommand{\sparens}[1]{\left[ #1 \right]}

\newcommand{\abs}[1]{\left| #1 \right|}
\newcommand{\norm}[1]{\left\| #1 \right\|}

\DeclareMathOperator*{\im}{Im}
\DeclareMathOperator*{\vectorspan}{span}

\newcommand{\set}[1]{\left\{ #1 \right\}}
\newcommand{\suchthat}[0]{\middle|}

\newcommand{\pr}[1]{\Pr\left[ #1 \right]}
\newcommand{\prex}[2]{\Pr_{#1}\left[ #2 \right]}

\DeclareMathOperator{\weightletter}{w}
\newcommand{\weight}[1]{\weightletter \parens{ #1 }}
\newcommand{\weightcplx}[2]{\weightletter_{#1}\parens{#2}}
\newcommand{\complex}[1]{\MakeUppercase{#1}}
\newcommand{\stdcomplex}{X}
\newcommand{\dimension}[1]{\MakeLowercase{#1}}

\DeclareMathOperator{\containment}{\Gamma}

\newcommand{\assignment}[1]{{\mathcal{\MakeUppercase{#1}}}}
\newcommand{\assignmentset}[1]{\pmb{\assignment{#1}}}
\newcommand{\lassignment}[1]{{\mathscr{\MakeUppercase{#1}}}}
\newcommand{\lassignmentset}[1]{{\pmb{\lassignment{#1}}}}

\newcommand{\face}[1]{\MakeLowercase{#1}}
\newcommand{\vertex}[1]{\MakeLowercase{#1}}
\newcommand{\stdvertex}{\vertex{v}}
\newcommand{\genvertex}{\vertex{u}}
\newcommand{\stdface}{\face{\sigma}}
\newcommand{\genface}{\face{\tau}}

\newcommand{\boundaryoperator}{\partial}
\newcommand{\cochain}[1]{\MakeUppercase{#1}}
\newcommand{\cochainset}[2]{C^{#1}\parens{#2}}
\newcommand{\stdcochain}{\cochain{f}}
\newcommand{\stdcocycle}{\cochain{Z}}
\newcommand{\stdcycle}{\gamma}
\newcommand{\cocycleset}[2]{Z^{#1}\parens{#2}}
\newcommand{\cycleset}[2]{Z_{#1}\parens{#2}}
\newcommand{\stdcoboundary}{\cochain{\phi}}
\newcommand{\coboundaryset}[2]{B^{#1}\parens{#2}}
\newcommand{\boundaryset}[2]{B_{#1}\parens{#2}}
\newcommand{\homologyset}[2]{H_{#1}\parens{#2}}
\usepackage[super]{nth}
\usepackage{algpseudocode}
\usepackage[T1]{fontenc}

\newcommand{\agreementexpansionconst}{\alpha}

\newcommand{\core}[1]{core\parens{#1}}
\DeclareMathOperator{\dist}{dist}
\newcommand{\distribution}[1]{\mathcal{#1}}

\newcommand{\rep}[2]{r^{#1}_{#2}}
\newcommand{\repex}[3]{r^{#1, #2}_{#3}}
\newcommand{\repcplx}[2][k]{\hat{R}^{#1}(#2)}
\newcommand{\repcplxcore}[3][k]{\hat{R}^{#1}_{#2}(#3)}
\newcommand{\cobdrtestconst}{\eta}

\newcommand{\trianglesnorm}[2][k]{\varepsilon_{\blacktriangle}^{#1}\parens{\cochain{#2}}}
\newcommand{\emptytrianglesnorm}[2][k]{\varepsilon_{\triangle}^{#1}\parens{\cochain{#2}}}

\newcommand{\sphericalbuilding}[2]{\mathbb{S}_{#1}^{#2}}

\newcommand{\intersectionsize}{\xi}

\SetKwProg{Pick}{pick}{}{}

\author{
Roy Gotlib
\footnote{Department of Computer Science, Bar-Ilan University, roy.gotlib@gmail.com, research supported by ERC. }
\and
Tali Kaufman
\footnote{Department of Computer Science, Bar-Ilan University, kaufmant@mit.edu, research supported by ERC and BSF.}
}

\title{List Agreement Expansion from Coboundary Expansion}

\begin{document}
    \maketitle
    \begin{abstract}
        One of the key components in PCP constructions are agreement tests.
        In agreement test the tester is given access to subsets of fixed size of some set, each equipped with an assignment.
        The tester is then tasked with testing whether these local assignments agree with some global assignment over the entire set.
        One natural generalization of this concept is the case where, instead of a single assignment to each local view, the tester is given access to $l$ different assignments for every subset.
        The tester is then tasked with testing whether there exist $l$ global functions that agree with all of the assignments of all of the local views.
        In this work we present sufficient condition for a set system to exhibit this generalized definition of list agreement expansion.
        This is, to our knowledge, the first work to consider this natural generalization of agreement testing.

        Despite initially appearing very similar to agreement expansion in definition, proving that a set system exhibits list agreement expansion seem to require a different set of techniques.
        This is due to the fact that the natural extension of agreement testing (i.e.\ that there exists a pairing of the lists such that each pair agrees with each other) does not suffice when testing for list agreement as list agreement crucially relies on a global structure.
        It follows that if a local assignments satisfy list agreement they must not only agree locally but also exhibit some additional structure.
        In order to test for the existence of this additional structure we use the connection between covering spaces of a high dimensional complex and its coboundaries.
        Specifically, we use this connection as a form of ``decoupling''.

        Moreover, we show that any set system that exhibits list agreement expansion also supports direct sum testing.
        This is the first scheme for direct sum testing  that works regardless of the parity of the sizes of the local sets.
        Prior to our work the schemes for direct sum testing were based on the parity of the sizes of the local tests.
    \end{abstract}

\section{Introduction}\label{sec:introduction}

\paragraph{Agreement testing} is an important tool that is central to many PCP constructions.
In agreement testing one is given access to a set of functions $\set{f_s}_{s \in S}$ that are thought of as ``local views'' of some global function $F: \bigcup_{s \in S}{s} \rightarrow \set{0,1}$.
These are local views in the sense that every $f_s$ is a function $f_s:s \rightarrow \set{0,1}$  and for every $s$: $F|_s = f_s$.
If a set of local functions $\set{f_s}_{s \in S}$ meets the above criterion we call that set an agreeing set of functions.
An agreement test is a probabilistic algorithm that picks two sets $s_1, s_2$ of size $k$ that intersect each other on $\intersectionsize$ elements and checks whether $f_{s_1}|_{s_1 \cap s_2} = f_{s_2}|_{s_1 \cap s_2}$.
A structure is said to support agreement testing (alternatively, a structure is an agreement expander) if the following two properties hold:
\begin{enumerate}
    \item The agreement test always accepts agreeing sets of functions.
    \item If the agreement test rejects a set with probability $\varepsilon$ then there is a global function that disagrees with at most $O\parens{\frac{1}{1-\frac{\intersectionsize}{k}}\varepsilon}$ of the local functions.
\end{enumerate}
Most works that pertain to agreement testing are interested in the case where the intersection between sets is small, specifically $\intersectionsize = \frac{k}{2}$.
This is because in those cases the proportion between $k$ and $\intersectionsize$ is constant and any local assignment that is rejected with probability $\epsilon$ has a global assignment which disagrees with $O(\epsilon)$ of the local assignments.
In this work, however, we are interested in a slightly different definition of agreement expansion called \emph{$1$-up agreement expansion}.
In $1$-up agreement the intersection between pairs of sets contain all but one element from each set (i.e.\ $\intersectionsize=k-1$).
Also note that many complexes (and even some sparse complexes~\cite{dinur2017high}) exhibit $1$-up agreement expansion (as well complexes such as the complete complex).

In~\cite{dinur2017high} Dinur and Kaufman proved that some high dimensional expanders support agreement testing and term these \emph{agreement expanders}.
Since then agreement expanders have proven to be extremely useful in various contexts:
From derandomization for direct product testing~\cite{dinur2017high} to conversion of local tests to robust tests~\cite{dinur2018local} and others (for more examples see~\cite{dikstein2020locally, dikstein2018boolean}).

\paragraph{List agreement expansion} In this paper we present a new, natural generalization of agreement testing in which, instead of being given a single function for each set $s$, we are given a list of $l$ functions and we want to test whether these functions are local views of $l$ global functions (where $l$ is some constant).
By that we mean that there exist $l$ global functions $F_1, \cdots , F_l: \bigcup_{s \in S}{s} \rightarrow \set{0,1}$ and a permutation $\pi_s$ for every $s$ such that every local view $f_s^i:s \rightarrow \set{0,1}$ agrees with the global function
$F_{\pi(i)}$ (i.e. $F_{\pi(i)}|_s = f_s^i$).
In this paper we ask whether there are structures that support list agreement tests.
We term such structures \emph{list agreement expanders}.

\begin{definition}[List Agreement expansion, informal. For formal see~\ref{def:list-agreement-expansion}]
    We say that a set system exhibits list agreement expansion if there exists a tester that, given access to a set of assignments $\lassignment{f} = \set{\lassignment{f}^s_i}_{s \in S, i \in \sparens{l}}$ such that $\lassignment{f}^s_i:s \rightarrow \set{0,1}$, queries $Q$ of the local assignments\footnote{The algorithm queries one of the local assignments in one of the lists.} $\lassignment{f}^{s_1}_{i_1},\cdots,\lassignment{f}^{s_Q}_{i_Q}$ and accepts or rejects such that the following holds:
    \begin{enumerate}
        \item Always accepts if there exists $l$ functions $F_1, \cdots , F_l: \bigcup_{s \in S}{s} \rightarrow \set{0,1}$ and a permutation $\pi_s$ for every $s$ such that every local view $f_s^i:s \rightarrow \set{0,1}$ agrees with the global function $F_{\pi_s(i)}$ (i.e. $F_{\pi_s(i)}|_s = f_s^i$).
        \item If the tester rejects with probability $\epsilon$ then $O(\epsilon)$ of the assignment in $\lassignment{f}$ can be changed such that the property stated above will hold.
    \end{enumerate}
\end{definition}

\paragraph{List agreement testing compared to agreement testing} At first list agreement testing might seem fairly reminiscent of agreement testing and, while the definitions are similar, list agreement testing seem to require different tools altogether.
In the list agreement testing paradigm one is not only concerned with having agreement between local assignments, but also that these agreements are structured in the right way.
One key example of agreements that are not well structured is the following:
Consider a cycle of odd length with the following $2$-assignments on every edge: $\lassignment{F}^{\set{u,v}} = \set{\sparens{u=1, v=0},\sparens{u=0, v=1}}$.
Note that any two edges that share a vertex agree on the intersection\footnote{In the sense that there is a permutation $\pi$ such that the $i$-th assignment of one of the faces agrees with the $\pi(i)$-th assignment of the other.}.
That being said, $\lassignment{F}$ is not agreeing (since if it was an agreeing assignment the graph would be bi-partite\footnote{One can interpret the local lists as ``$u$ and $v$ are on different sides of the graph''.}).

\paragraph{Conditions under which list agreement testing is possible}
As we discussed, list agreement not only requires agreement but also some additional structure.
We will show that this additional structure comes in the form of \emph{coboundary expansion} - A topological notion of expansion.
Our construction of a list agreement tester will rely on coboundary expansion as a way to decouple the $l$ instances of agreement testing and then use agreement expansion in order to achieve local agreement.
We found that the global structure required in order to have $l$ global cochains is, in a sense, equivalent to a coboundary of a different complex.
That complex is induced by the original complex and the agreement test.
We show that that the coboundaries on the induced complex are testable using the fact that the original complex is \emph{both} a coboundary expander and a $1$-up agreement expander.
Effectively, we use the coboundary expansion in order to derive the \emph{global} structure (i.e.\ that the local agreements can indeed be ``glued together'' into $l$ global functions) while using the agreement expansion in order to derive the \emph{local} structure (i.e.\ that the local agreements agree with each other).

As we previously hinted at, list agreement offers a very descriptive language which can, at times, be considerably richer than regular agreement.
This richer structure allows us to, for example, describe the question of whether a subgraph\footnote{Here a subgraph is determined by picking a subset of the vertices and all the edges that connect them.} of the complex's underlying graph is two sided or not.
One such subgraph of particular interest is a cycle as list agreement allows us to describe the question of whether a cycle is of odd length merely by knowing which vertices are a part of the cycle.
This, in effect, yields a non-constant lower bound on the number of queries required in order to test list agreement in the general case as testing whether a cycle is odd cannot be done locally (More details on this can be found in Section~\ref{sec:on-the-2-differing-assumption}).

In order to overcome this limitation we introduce a restriction on list agreement.
Namely, we require that the local assignments given to each set have some small distance separating them.
Using this small distance we show that testing list agreement is possible with a constant number of queries.

Now that we have an understanding of list agreement expansion we can present our main theorem:
\begin{theorem}[Main Theorem, informal. For formal see~\ref{main-theorem}]\label{thm:main-informal}
Any simplicial complex that has sufficient expansion properties (namely coboundary expansion and a $1$-up agreement expansion) supports list agreement testing using $3l$ queries\footnote{The test queries all the local assignments of three faces.} (where $l$ is the length of the list) under small distance assumptions on the local assignments\footnote{More specifically we assume that the local assignments differ on at least two vertices. This assumption cannot be removed in the domain of complexes that we examine, see Section~\ref{sec:on-the-2-differing-assumption}.}.
\end{theorem}

It is important to note that there are simplicial complexes who meet the Theorem's criteria.
For example, the spherical buildings and the complete complex have sufficient expanding conditions for Theorem~\ref{thm:main-informal} (See~\cite{dinur2019near} and~\cite{kaufman2020local}).

\paragraph{Direct sum testing} Another natural (and extremely useful) construction in hardness amplification is the direct sum.
\begin{definition}[Direct sum]
Given a function $f:S \rightarrow \set{0,1}$ (where $S$ is an arbitrary set) its $k$-fold direct sum is a function $F:\binom{S}{k} \rightarrow \set{0,1}$ such that: $F(A) = \sum_{a \in A}{f(a)}$.
\end{definition}
Direct sums are useful in a variety of contexts, from Yao's XOR Lemma~\cite{4568378} which states that if a function is hard to approximate then its direct sum is exponentially harder, to the hardness of approximation of problems in $P^{NP}\parallel$\footnote{$P^{NP}\parallel$ is the set of problems that can be solved in polynomial time with oracle access to a problem in $NP$ such that all the queries to the oracle are performed in parallel.}~\cite{impagliazzo2009approximate}.

\paragraph{A unified framework for direct sum testing} It is natural to ask how, given a function, can one make sure that it is indeed a direct sum in a derandomized fashion.
There have been several works on derandomizing direct sum testing~\cite{david2017direct, gotlib2019testing} but the tests presented in them for constant values of $k$ were heavily dependent on the parity of $k$.
In this work, we show how to use list agreement expanders in order to provide a new natural test for whether a function $F$ is a direct-sum (while having stronger assumptions on the expansion of the complex).
Our testing framework is the first that can be applied to any value of $k$ regardless of its parity.
In addition, our framework shaves off an $O(k)$ factor in the query complexity for the case when $k$ is odd (compared to~\cite{gotlib2019testing}).
Specifically we show that:
\begin{theorem}[Testability of Direct Sums, informal. For formal see Theorem~\ref{thm:direct-sum}]
    Any simplicial complex that supports list agreement testing supports direct sum testing, regardless of the parity of $k$.
\end{theorem}
\subsection{High Dimensional Expansion Toolset}\label{subsec:high-dimensional-expansion-toolset}

We will now move on to presenting the main toolset we use in the proof:
\paragraph{Simplicial complexes} are generalizations of graphs to higher dimensions.
A simplicial complex is a hyper-graph with closure property, i.e.\ if $\stdface$ is a hyper-edge then any subset of $\stdface$ is also an hyper-edge in the hyper-graph.
We term the hyper-edges of a simplicial complex as its faces and define the dimension of a face $\stdface$ to be $\abs{\stdface}-1$.
We denote the set of $i$ dimensional faces by $\stdcomplex(i)$.
We also define the dimension of a simplicial complex as the dimension of its maximal face.
For example, any non-empty graph is a $1$-dimensional simplicial complex\footnote{We add the empty face to the graph.}, its vertices are the $0$-dimensional faces and edges are its $1$-dimensional faces.
Note that in connected graphs all of the maximal faces are of the same dimension as every vertex is part of an edge (otherwise there exists an isolated edge and the graph is not connected), when this holds we say that complex is \emph{pure}.
We limit our discussion to pure simplicial complexes.
It is often convenient to think of high dimensional faces as geometrical shapes, for example we think of a $2$-dimensional face as a triangle, a $3$-dimensional face as a pyramid etc.
In many cases we will be interested in weighted simplicial complexes.
In weighted simplicial complexes a weight function is given to the top dimensioanl faces.
This weight function is positive and sums up to $1$.
The weight of lower dimensional faces is determined by the weight of the top dimensional faces in which they are contained.
It is important to note that even if the top dimensioanl faces all have the same weight, the same does not necessarily hold for faces of a lower dimension.

\paragraph{High dimensional expanders} are generalizations of one-dimensional expanders (i.e.\ graph expanders) to higher dimensions.
Unlike the one dimensional case, there are multiple definitions for the notion of expansion in higher dimensions.
Moreover, the connections between the various generalizations of expansion is unknown.
In this work we will use two of these definitions: \emph{coboundary expansion} and \emph{agreement expansion}.
We are going to think of both definitions of expansion as measures of how well certain properties can be tested on the complex.

\paragraph{Coboundary expansion} is a natural generalization of the combinatorial expansion to higher dimensions.
Specifically, one dimensional expansion can be thought of how well the following test tests whether a cut in the graph is trivial: Pick a random edge and accept iff it does not cross the cut.
Note that this corresponds directly with the graph's cheeger constant as the cut can be thought of as a set and the number of edges that cross the cut are exactly the outgoing edges of the set.
Therefore the cheeger constant of the graph determines \emph{exactly} how well the test performs.
Coboundary expansion is a generalization of the cheeger constant that includes the second dimension as well.
In the second dimension the property being tested is whether a given subset of edges represents a cut (i.e.\ whether there is an underlying cut such that an edge is in the set iff it crosses the cut).
The test being measured is the test that picks a triangle and accepts if the number of edges that are both in the set and the triangle is even.
In a coboundary expander both tests (i.e.\ the test in the first and second dimensions) have a large distance (larger than some constant).

\paragraph{Agreement expanders} are (possibly sparse) simplicial complexes and a distribution $\mathcal{D}$ on the $k$-dimensional faces such that,
for any set of assignments to the $k$-dimensional faces of the simplicial complex,
if most pairs $\stdface_1,\stdface_2 \sim \distribution{D}$ agree on their intersection then the local assignments are close to agreeing with some global assignment (for any $k$).
In this work we are going to use a strengthening of agreement expansion, namely \emph{$1$-up agreement expanders},
$1$-up agreement expanders are agreement expanders where the distribution $\mathcal{D}$ matches the procedure of picking a $\parens{k+1}$-dimensional face and then randomly selecting two $k$-faces that are contained in it.
In~\cite{dinur2017high} Dinur and Kaufman show that there exists a family of bounded degree $1$-up agreement expanders.

\paragraph{Covers} of a simplicial complex $\stdcomplex$ are simplicial complexes $\complex{y}$ that, for every vertex $\stdvertex$ of $\stdcomplex$ contain $l$ vertices $[\stdvertex,1],\dots,[\stdvertex,l]$,
In addition, $\stdface=\parens{\stdvertex_1, \dots, \stdvertex_i}$ is a face in $\stdcomplex$ iff $\complex{y}$ has $l$ disjoint faces of the form $\parens{[\stdvertex_1, j_1], \dots, [\stdvertex_i, j_i]}$.
For example: a $2$-cover of a $1$-dimensional complex (i.e.\ a graph) has, for every vertex $\stdvertex$, two vertices $[\stdvertex, 0],[\stdvertex, 1]$.
In addition, for every edge $\parens{\stdvertex,\genvertex}$ in the graph the cover has two edges, either $\parens{[\stdvertex,0],[\genvertex,0]}, \parens{[\stdvertex, 1],[\genvertex, 1]}$ or $\parens{[\stdvertex, 0],[\genvertex, 1]},\parens{[\stdvertex, 1],[\genvertex, 0]}$.
In this work we are going to discuss a special subset of the covering spaces of a simplicial complex.
Specifically we are going to be interested in the trivial covering spaces that are comprised of $l$ disjoint copies of the original complex.
We end our discussion of covering spaces by noting that the definition of a cover applies to general topological spaces and not necessarily to simplicial complexes.

\paragraph{Near Covers} are spaces that are close to being covering spaces of a given simplicial complex.
Near covers are a relaxation of covers where only the first two dimensions are required to be covered properly.
More specifically, every vertex $\stdvertex$ is covered by $l$ vertices and every edge $\set{\stdvertex, \genvertex}$ is covered by $l$ edges.
The higher dimensional faces are only covered if the lower dimensioanl faces are structured in a way that allows them to be covered.
For example consider a simplicial complex comprised of a single two dimensional face $\set{\stdvertex_1, \stdvertex_2, \stdvertex_3}$.
Consider the near cover of said complex whose $1$-dimensioanl faces are the following cycle of length $6$:
\[
    \set{
        \set{\sparens{\stdvertex_1,0}, \sparens{\stdvertex_2,0}},
        \set{\sparens{\stdvertex_2,0}, \sparens{\stdvertex_3,0}},
        \set{\sparens{\stdvertex_3,0}, \sparens{\stdvertex_1,1}},
        \set{\sparens{\stdvertex_1,1}, \sparens{\stdvertex_2,1}},
        \set{\sparens{\stdvertex_2,1}, \sparens{\stdvertex_3,1}},
        \set{\sparens{\stdvertex_3,1}, \sparens{\stdvertex_1,0}}
    }
\]
Note that every vertex and every edge are covered and yet the single $2$-dimensional face is not and thus this is a near cover that is not a cover.
It is important to note that every cover is a near cover and that sometimes \emph{genuine cover} is used instead of cover to emphasize that a near cover is indeed a cover.

\subsection{Proof Strategy}\label{subsec:proof-stategy}

In this work we are interested in testing whether a set system $S$ exhibits list agreement expansion.
Agreement expansion stems directly from the rapid convergence of random walk that spectral expanders exhibit.
While we still test for agreement (and therefore require the same spectral expansion assumptions on our set system) we also recall that list agreement requires some additional structure over agreement.
Specifically, not only does list agreement require local agreement it also requires the local agreements to exist in such a way that $l$ different global functions are formed.

Let us now characterize this additional structure further.
Consider a set of assignments $\lassignment{F}$ that gives every set $s$ in the set system a set of $l$ assignments $\lassignment{F}^s_{1}, \cdots, \lassignment{F}^s_{l}$.
For ease of presentation, allow us also to assume that for every two sets in the set system $s, s'$ there exists a unique permutation $\pi_{s,s'}$ such that $\lassignment{F}^s_{i}|_{s \cap s'} = \lassignment{F}^{s'}_{\pi_{s,s'}(i)}|_{s \cap s'}$.
Consider the following definition:
\begin{definition}[Coboundary structure]\label{def:coboundary-struct}
    Let $A \subseteq \binom{S}{2}$ and let $G$ be a group.
    We say that a function $f:A \rightarrow G$ has coboundary structure over $A$ with coefficients in $G$ if there exists $g: S \rightarrow G$ such that:
    \[
        f(a,b) = g(b)\parens{g(a)}^{-1}
    \]
\end{definition}
We will show that, regardless of the set $A$, $\lassignment{f}$ exhibits list agreement iff $f(s,s')=\pi_{s,s'}$ has a coboundary structure with coefficients in $S_l$ (the symmetric group with $l$ elements).
Recall $\lassignment{f}$ exhibits list agreement if there exists $l$ global functions $\assignment{f}_1, \cdots, \assignment{f}_l$ and, for every set $s$ there exists a permutation $\pi_s$ such that $\lassignment{F}^s_{i} = \assignment{f}_{\pi_s(i)}|_s$.
It is easy to see that in this case $\pi_{s,s'} = \pi_{s'} \parens{\pi_{s}}^{-1}$.
In addition, if $f(s,s')$ has a coboundary structure then there exists $g$ such that $f(a,b) = g(b)\parens{g(a)}^{-1}$.
In that case, the global functions $\assignment{F}_i(\stdvertex)$ can be calculated by picking any set $s$ that contains $\stdvertex$ and setting $\assignment{F}_i(\stdvertex) = \lassignment{f}^s_{g(s)(i)}(\stdvertex)$.
We are therefore interested in finding out if the permutations $\pi_{s,s'}$ exhibit a coboundary structure.

Alas the fact that two local lists of assignments agree with each other under some permutation does not necessarily mean that the permutation remains the same when correcting $f$ so that it exhibits coboundary structure (even if there is only one permutation that causes agreement).
It is therefore natural to ask ``\emph{how far are $\pi_{s,s'}$ from exhibiting a coboundary structure?}''.
This is \emph{exactly} what coboundary expansion with coefficients in $S_l$ measures.

Then, after we have corrected the function $f$ so that it exhibits a coboundary structure, we might have ruined some of the agreement we started with.
We will note, however, that now we have $l$ \emph{independent} instances of an agreement problem (One for each $\set{\lassignment{f}^s_{g(s)(i)}}_{s \in S}$, where $g: S \rightarrow S_l$ such that $f(s,s')=g(s')\parens{g(s)}^{-1}$).
We use the agreement expansion to resolve those.

To conclude, we use spectral expansion in order to derive the agreement and topological expansion in order to derive this additional structure.
\subsection{Proof Layout}\label{subsec:proof-layout}
We start by assuming that the complex is a $1$-up agreement expander and therefore the $1$-up test is a good agreement test.
We then model the choices done by the $1$-up agreement test as a simplicial complex which we dub ``the representation complex''.
The modeling is done in the following way:
The vertices of the new complex will correspond to the $k$-dimensional faces of the original complex.
The edges of the new complex will correspond to the choices done by the $1$-up agreement test (i.e.\ if a pair of faces $\stdface_1, \stdface_2$ are chosen with some probability then there is an edge between $\stdface_1$ and $\stdface_2$).
We will construct the higher dimensional faces of the representation complex so that the weight of the an edge $(\stdface_1,\stdface_2)$ in the representation complex will correspond with the probability that the pair $\stdface_1,\stdface_2$ is chosen by the $1$-up agreement tester (see Definition~\ref{the-representation-complex} for a formal definition).
Note that the edges of the representation complex can be thought of as the set $A$ from Definition~\ref{def:coboundary-struct}.
We will therefore be interested in examining the expansion properties of the second dimension of the representation complex.
Alas, the representation complex is not a coboundary expander but it does have a structure that facilitates bounding the distance of a cochain from being a coboundary using a local property (see Appendix~\ref{sec:on-the-testability-of-co-boundaries-in-the-representation-complex} for more details).

Before we move on we note that the function $f$ from Definition~\ref{def:coboundary-struct} can be thought of as a description of a near cover of the representation complex in the following way:
Every vertex $\stdvertex$ is covered by $\abs{G}$ elements - $\set{\sparens{\stdvertex, 1}, \cdots, \sparens{\stdvertex, \abs{G}}}$.
The edges are described by the function $f$ in the following way: If $\stdvertex$ and $\genvertex$ are connected in the original complex then $\sparens{\stdvertex, i}$ and $\sparens{\genvertex,f\parens{\set{\stdvertex, \genvertex}}(i)}$ are connected in the near cover.
In order to simplify the presentation of the rest of the proof we are going to use the language of near covers.
It is important to note that it is equivalent to the presentation in the proof strategy.

We start by showing how to relate any $l$-assignment whose local assignments are sufficiently differing to an $l$-near-cover of the representation complex.
Specifically, every one of the local assignments to $\stdface$ in the $l$-assignment will correspond to a vertex that covers $\stdface$ and two vertices in the cover are connected if the assignments they correspond to agree on their intersection (and, of course, if the edge that they cover exists in the representation complex).
Unlike the assumption in the proof strategy, we cannot assume that every pair of lists agree with each other.
We therefore cover every edge $(\stdvertex_1, \stdvertex_2)$ whose assignments do not agree with each other using some fixed matching between the vertices that cover $\stdface_1$ and $\stdface_2$.
Full details of this construction can be found in Section~\ref{sec:local-assignments-in-the-original-complex-imply-a-near-cover-in-the-representation-complex}.

Now that we have constructed the near cover we want it to have the additional structure that we presented in the proof strategy.
We note that a near cover exhibits the additional structure we are interested in iff it corresponds to a coboundary.
We also note that the near cover corresponds to a coboundary iff it is a genuine cover that is comprised of $l$ distinct copies of the original complex.
We use the connection between cochains and near covering spaces as well as the result from Appendix~\ref{sec:on-the-testability-of-co-boundaries-in-the-representation-complex} to show that it is possible to bound the distance of the near cover we constructed from being comprised of $l$ distinct copies of the original complex.
Consider now a corrected version in the near cover.
While now we have the additional structure we are interested in, we still have not gotten the agreement we are interested in.

Assume for now that we can query this genuine cover directly (which we cannot as we only have access to a near cover that is close to it).
In the genuine cover, each of the copies of the original cover has an associated assignment to its $k$-faces (since every vertex in the cover is associated with a single local assignment to the face it represents).
Consider running the $1$-up agreement test on each of the copies of the representation complex with its associated assignments.
The test picks a $(k+1)$-face $\stdface$ and two of its $k$-sub-faces $\stdface_1, \stdface_2$, it then queries the local assignment of these $k$-faces and accepts iff they agree on their intersection.
Another way of looking at the test is that it samples an edge $e$ in the representation complex and accepts iff the faces represented by the vertices in the edge agree on their intersection.
We can therefore bound the distance of the set of assignments from agreeing using the agreement expansion of the original complex.

Alas we cannot query the genuine cover directly.
We can, however, bound from above the probability that the agreement test would reject had we ran it on that genuine cover.
Recall the construction of the near cover and note that if $e$ was covered by edges that represent agreement in the original near cover and, in addition, the cover of $e$ was not changed then the assignments to both sides of $e$ agree on their intersection.
Therefore the agreement test would only reject if it picked an edge that had no agreement in the first place or an edge that was changed when the near cover was corrected.
Note that the norm of the edges that satisfy these properties can be derived locally by querying the lists of assignments - the former by observing whether there is an agreement and the latter due to the testability of coboundaries in the representation complex.
We use this in order to bound the distance of any $l$-assignment from being an agreeing $l$-assignment.

Before we move on we wish to emphasise how the covering spaces allow us to decouple dependencies.
Consider the near covering space of the representation complex.
In an ideal case, the near cover induced by the $l$-assignment is a genuine cover that corresponds to a coboundary.
In such cases the cover is comprised of $l$ independent copies of the original complex and we can run the agreement test $l$ times independently and get the distance of the $l$-assignment from agreeing.
In most cases, however, the near cover induced by the $l$-assignment is not a genuine cover that corresponds to a coboundary.
In these cases the copies of the complex are dependent on each other in the sense that what would have been copies of the original complex are now connected via various edges.
A crucial step in the test we propose involves ``decoupling'' these dependencies.
We will show that, despite the fact that the representation complex's cohomology does not vanish, the coboundaries of the representation complex are still locally testable.
Using this test we can bound the distance of the $l$-assignment from inducing independent instances of agreement testing.
We believe that the technique of modeling objects of interest as near covers of a coboundary expander could be a useful measure of dependency as not only does it bound the distance of the object from being independent but it is also locally testable.
\subsection{Related Work}\label{subsec:related-work}
Agreement testing is an inherit part of most PCP constructions as well as various other applications and were extensively studied in the past few years (for examples, see~\cite{dinur2017high,dikstein2019agreement,dinur2019analyzing,dinur2018towards,barak2018small,khot2017independent, dinur2017agreement}).
In recent years a connection between high dimensional expanders and agreement testing was established:
In~\cite{dinur2017high} Dinur and Kaufman defined the notion of agreement expanders described above.
Later Kaufman and Mass presented in~\cite{kaufman2020local} a new method for constructing agreement expanders.
In this work we generalize the notion of agreement expanders to the case where each face has $l$ local assignments.
The goal in this new setting is to check whether there exists $l$ global functions that match all the local assignments (for formal definition see~\ref{def:l-assignment}).

In order to provide said generalization we use the connection between cocycles and cover spaces of simplicial complexes.
The connection between cocycles and cover spaces of a simplicial complex was first introduced in~\cite{surowski1984covers}.
Then in~\cite{dinur2019near} Dinur and Meshulam defined the notion of cover stable complexes which are complexes where if a near cover satisfies most local conditions then it is close to a genuine cover.
They then show that a simplicial complex is cover stable if and only if it is an expanding with respect to non-abelian cohomology.
We use the connection between covers of high dimensional expanders and their cocycles in a new way:
Specifically, we use the structure of covers that correspond with coboundaries in order to untangle the near cover into disjoint copies of the original complex.

Direct sums are natural construction in hardness of approximation.
A function $F:\binom{[n]}{k} \rightarrow \set{0,1}$ is a $k$-direct-sum of $f:[n]\rightarrow\set{0,1}$ if for every $\stdface \in \binom{[n]}{k}$ it holds that $F(\stdface) = \sum_{\stdvertex \in \stdface}{\stdvertex}$.
Testing direct sums is a problem that was extensively studied in recent years.
The first connections between testing direct sums and high dimensional expanders was presented in~\cite{kaufman2014high} in which Kaufman and Lubotzky presented a test for the case of $k=2$.
Later in~\cite{david2017direct} (and generalized to high dimensional expanders in~\cite{dinur2017high}) a test for constant even values of $k$ was found.
This test, however, could not deal with the case where $k$ is odd due to inherit limitations.
After that, another test was proposed by Gotlib and Kaufman in~\cite{gotlib2019testing}.
Their test could deal with the case where $k$ is an odd constant, however there seem to be no simple way of extending their setting to the case where $k$ is even.
In this work we use list agreement expanders in order to present the first test for direct sums for constant values of $k$ regardless of its parity.

\subsection{Paper Organisation}\label{subsec:paper-organisation}
We will start by defining the notions of high dimensional expansion we use throughout the paper in Section~\ref{sec:preliminaries}.
In addition to that Section~\ref{sec:preliminaries} includes a formal definition of list agreement expansion and the distance measure we are going to use.
Then, in section~\ref{sec:the-representation-complex}, we present the representation complex in depth - we characterize its structure and expansion properties as well as its cover stability.
In section~\ref{sec:on-the-covers-of-complexes-that-represent-co-boundaries} we present properties of covers that correspond to coboundaries\footnote{This section is completed by Appendix~\ref{sec:on-the-testability-of-co-boundaries-in-the-representation-complex} in which we show that the coboundaries are testable in the representation complex.}.
We then move on to show how assignments to the original complex imply a near cover for the representation complex in section~\ref{sec:local-assignments-in-the-original-complex-imply-a-near-cover-in-the-representation-complex}.
After that we show how to use the cover stability of the representation complex in order to show that the complex is a list agreement expander in section~\ref{sec:presenting-a-test-for-list-agreement}.
We then show how to use list agreement expanders in order to provide a test for $k$-direct-sum that is independent of the oddity of $k$ (in section~\ref{sec:testing-direct-sums-using-list-agreement-expansion}).
In Section~\ref{sec:on-the-2-differing-assumption} we discuss the $2$-differing assumption further and show that without the $2$-differing assumption there are some complexes of interest (for example, the spherical buildings) in which there are no tests for list agreement that perform a constant number to queries.
Then we show a lower bound to the number of queries required for list agreement in Appendix~\ref{sec:lower-bound-on-the-number-of-queries}.

\section{Preliminaries}\label{sec:preliminaries}

\subsection{Simplicial Complexes}\label{subsec:simplicial-complexes}
As we stated before, the generalization of graphs we will be using are simplicial complexes,
which we will now present more thoroughly:
\begin{definition}[Simplicial complex]
    A simplicial complex $\stdcomplex$ is set of sets such that if $\stdface \in \stdcomplex$ then every $\stdface' \subseteq \stdface$ is also in $\stdcomplex$.
    Each of the sets in $\stdcomplex$ are termed the faces of $\stdcomplex$.
\end{definition}
\begin{note*}
    The faces of a simplicial complex have an orientation.
    The orientation must be consistent in the sense that there is an ordering of the $0$-dimensional faces that determines the orientation of all the faces in higher dimensions.
    In the vast majority of this paper the orientation is irrelevant to the arguments and therefore we will ignore it for the most part.
    When we are forced to consider the orientation of a face $\stdface = \set{\stdvertex_0,\dots,\stdvertex_i}$ we will denote the face with round brackets (i.e.~$\stdface = \parens{\stdvertex_0,\dots,\stdvertex_i}$)
\end{note*}
Let us also define the dimension of a face and the dimension of the complex
\begin{definition}[Dimensions]
    Let $\stdcomplex$ be a simplicial complex.
    Define the dimension of a face $\stdface \in \stdcomplex$ to be $\dim(\stdface) = \abs{\stdface}-1$.
    Also, denote the set of faces of dimension $i$ by $\stdcomplex(i)$
    In addition, define the dimension of the simplicial complex $\stdcomplex$ to be:
    \[
        \dim(\stdcomplex) = \max_{\stdface \in \stdcomplex}\parens{\dim\parens{\stdface}}
    \]
\end{definition}
\begin{definition}[Pure simplicial complex]
    A simplicial complex is called pure if all of its maximal faces are of the same dimension.
\end{definition}
From this point onwards we limit our discussion to pure simplicial complexes (and whenever we refer to a simplicial complex we will actually refer to a pure simplical complex).
A standard weight function is defined over the various faces of a simplicial complex which we will present below:
\begin{definition}[Weight function]
    Let $\stdcomplex$ be a $d$ dimensional simplicial complex.
    Define the weight of a face $\stdface \in \stdcomplex$ as the fraction of maximal faces in $\stdcomplex$ that contain $\stdface$.
    Formally:
    \[
        \weightcplx{\stdcomplex}{\stdface}=\frac{\abs{\set{\genface \in \stdcomplex(d) \suchthat \stdface \subseteq \genface}}}{\binom{d+1}{\abs{\stdface}} \abs{\stdcomplex(d)}}
    \]
    When the complex is clear from context we will omit it from the notation.
\end{definition}
Note that this weight function can be thought of as a probability distribution over the faces of each dimension (and indeed we will think of it as such).
Moreover there is a way to sample a face with the probability distribution defined by the weight.
We can also use this weight function to define the following norm over sets of faces:
\begin{definition}[Norm on faces]\label{def:norm-on-faces}
    Let $\stdcomplex$ be a $d$-dimensional simplicial complex, let $-1 \le i \le d$ and also let $S \subseteq \stdcomplex(i)$.
    Define the norm of $S$ in $\stdcomplex$ to be:
    \[
        \norm{S}_{\stdcomplex} = \sum_{\stdface \in S}{\weightcplx{\stdcomplex}{\stdface}}
    \]
    When the complex is clear from context we will omit it from the notation.
\end{definition}
In many cases it is helpful to consider all the faces of higher dimension that contain a face from a set of faces $S$.
In order to do that we are going to use the containment operator defined below:
\begin{definition}[Containment operator]
    Let $\stdcomplex$ be a $d$-dimensional simplicial complex, let $-1 \le i \le j \le d$ and also let $S \subseteq \stdcomplex(i)$.
    Define the following:
    \[
        \containment^j(S) = \set{\genface \in \stdcomplex(j) \suchthat \exists \stdface \in S: \stdface \subseteq \genface}
    \]
\end{definition}
There is also a connection between $\norm{S}$ and $\norm{\containment^k(S)}$ depicted in the following fact:
\begin{fact}
    Let $\stdcomplex$ be a $d$-dimensional simplicial complex, let $-1 \le i \le j \le d$ and also let $S \subseteq \stdcomplex(i)$ then:
    \[
        \norm{S} \le \norm{\containment^j(S)} \le \binom{j+1}{i+1}\norm{S}
    \]
\end{fact}
We will now introduce the notions of cochains, cocycles and coboundaries which are natural spaces of functions over any given simplicial complex:
\begin{definition}[Cochains]
    Let $\stdcomplex$ be a simplicial complex and let $G$ be a group.
    Define the set of $0$-cochains with coefficients in $G$ to be the set of functions from the vertices of $\stdcomplex$ to $G$.
    In addition, denote the set of $0$-cochains with coefficients in $G$ by $\cochainset{0}{\stdcomplex;G}$.
    Also, define the set of $1$-cochains with coefficients in $G$ as the following set:
    \[
        \cochainset{1}{\stdcomplex;G} = \set{\stdcochain:\stdcomplex(1) \rightarrow G \suchthat \stdcochain(\stdvertex,\genvertex) = \parens{\stdcochain(\genvertex, \stdvertex)}^{-1}}
    \]
\end{definition}
We also define the following operators:
\begin{definition}[Coboundary operators]
    Define the operator $d_{-1}:\cochainset{-1}{\stdcomplex;G} \rightarrow \cochainset{0}{\stdcomplex;G}$ to be:
    \[
        d_{-1}\cochain{f}(v) = \cochain{f}(\emptyset)
    \]
    Define the operator $d_0:\cochainset{0}{\stdcomplex;G} \rightarrow \cochainset{1}{\stdcomplex;G}$ to be:
    \[
        d_0\cochain{f}(u,v) = \cochain{f}(u)\parens{\cochain{f}(v)}^{-1}
    \]
    Also define the operator $d_1$ such that for every cochain $\stdcochain \in \cochainset{1}{\stdcomplex;G}$ and every $(u,v,w)\in \stdcomplex(2)$:
    \[
        d_1\cochain{f}(u,v,w) = \cochain{f}(u,v)\cochain{f}(v,w)\cochain{f}(w,u)
    \]
\end{definition}
These operators define the following spaces over the first three dimensions of the simplicial complex:
\begin{definition}[Cocycles and coboundaries]
    Let $\stdcomplex$ be a simplicial complex and $G$ a group, define the following spaces:
    \begin{itemize}
        \item For $i \in \set{0,1,2}$ define the $i$ dimensional coboundaries:
        \[
            \coboundaryset{i}{\stdcomplex;G} = \set{d_{i-1}\cochain{f}  \suchthat \cochain{f} \in \cochainset{i-1}{\stdcomplex;G}}
        \]
        \item For $i \in \set{0,1}$ define the $i$ dimensional cocycles:
        \[
            \cocycleset{i}{\stdcomplex;G} = \set{\cochain{f} \in \cochainset{i}{\stdcomplex;G} \suchthat d_i\cochain{f} = \mathbbm{1}}
        \]
    \end{itemize}
    And, as with the cochains, when $G=\mathbbm{F}_2$ we omit $G$ from the notation.
\end{definition}
\begin{note*}
    We only define these spaces in the first two dimensions because we are working with a general group $G$ rather then an abelian group.
    If we assume that $G$ is an abelian group one can generalize the definition of cochains, cocycles and coboundaries as well as the coboundary operators to higher dimensions.
\end{note*}
Consider also the following:
\begin{fact}
    For every simplicial complex $\stdcomplex$ and group $G$: $\coboundaryset{i}{\stdcomplex; G} \subseteq \cocycleset{i}{\stdcomplex; G} \subseteq \cochainset{i}{\stdcomplex; G}$
\end{fact}
Of particular interest to our case are coboundaries with coefficients in the symmetric group with $l$ elements that we denote $S_l$.

We also extend the norm defined in Definition~\ref{def:norm-on-faces} to a norm over the cochains:
\begin{definition}[Norm of a cochain]
    Let $\stdcomplex$ be a simplicial complex.
    For any $\stdcochain \in \cochainset{i}{\stdcomplex;G}$, define the following norm:
    \[
        \norm{\stdcochain}_{\stdcomplex} = \norm{\set{\stdface \in \stdcomplex(i) \suchthat \cochain{f}(\stdface) \ne 1}}
    \]
    When the complex is clear from context we will omit it from the notation.
\end{definition}
This norm also defines a natural distance function between any two cochains:
\begin{definition}[Distance between cochains]
    Let $\stdcochain_1, \stdcochain_2$ be two cochains in some simplicial complex $\stdcomplex$.
    Define the distance between $\stdcochain_1$ and $\stdcochain_2$ to be:
    \[
        \dist\parens{\stdcochain_1, \stdcochain_2} = \norm{\stdcochain_1\parens{\stdcochain_2}^{-1}}
    \]
\end{definition}
It is also natural, given a simplicial complex, to describe the simplicial complex that is constructed using only faces whose dimension is at most some $i$.
This is called the \emph{skeleton} of the simplicial complex and is formally defined below:
\begin{definition}[Skeleton]
    Let $\stdcomplex$ be a simplicial complex.
    Define its $i$-th skeleton to be the following simplicial complex:
    \begin{equation*}
        \set{\stdface \in \stdcomplex \suchthat \dimension{\stdface} \le i}
    \end{equation*}

\end{definition}
It would also be useful to define local views of faces in a simplicial complex.
We call these local views \emph{links} and think of them as the faces seen by a certain face.
\begin{definition}[Link]
    Let $\stdcomplex$ be a $d$-dimensional simplicial complex and let $\stdface \in \stdcomplex(i)$.
    Define the link of $\stdface$ in $\stdcomplex$ as the following $\parens{d-i}$-dimensional simplical complex:
    \[
        \stdcomplex_{\stdface} = \set{\genface \setminus \stdface \suchthat \stdface \subseteq \genface \text{ and } \stdface \in \stdcomplex}
    \]
\end{definition}
Note that the weight function of the faces in the link of a face $\stdface$ is strongly connected to the weight of $\stdface$ and the weight function of the original complex:
\begin{lemma}
    Let $\stdcomplex$ be a $d$-dimensional simplicial complex and let $\stdface$ be an $i$-dimensional face.
    The weight of a $j$-dimensional face in the link of $\stdface$ is:
    \[
        \weightcplx{\stdcomplex_\stdface}{\genface} = \frac{\weightcplx{\stdcomplex}{\genface \cup \stdface}}{\binom{\abs{\stdface}+\abs{\genface}}{\abs{\genface}}\weightcplx{\stdcomplex}{\stdface}} = \frac{\weightcplx{\stdcomplex}{\genface \cup \stdface}}{\binom{i+j+2}{j+1}\weightcplx{\stdcomplex}{\stdface}}
    \]
\end{lemma}
We finish our presentation of simplicial complexes by defining isomorphic simplicial complexes:
\begin{definition}[Isomorphic simplicial complexes]
    Let $\stdcomplex, \complex{y}$ be two simplicial complexes.
    We say that $\stdcomplex$ is isomorphic to $\complex{y}$ and denote $\stdcomplex \cong \complex{y}$ if there exists an invertible function $f: \complex{y}\rightarrow\stdcomplex$ such that of every $\stdface, \genface \in \complex{y}$ it holds that:
    \[
        \stdface \subseteq \genface \Leftrightarrow f(\stdface) \subseteq f(\genface)
    \]
\end{definition}
\subsection{On Assignments and $l$-Assignments}\label{subsec:on-assignments-and-l-assignments}

We will now present the notions of assignments and $l$-assignments we will then move on to define list agreement expanders.

We will start by defining assignments and agreeing assignments.
Assignments are essentially a set of local functions from every $k$-dimensional face of the complex to $\set{0,1}$, while agreeing assignments can be thought of as ``snippets'' of some global function.
We will end the discussion of assignments by defining the distance between two assignments as the number of faces on which they differ.
\begin{definition}[Assignment]
    Define an assignment to the $k$-faces of a simplicial complex $\stdcomplex$ to be $\assignment{F} = \set{\assignment{F}^\stdface}_{\stdface \in \stdcomplex(k)}$ such that $\assignment{F}^\stdface: \stdface \rightarrow \set{0,1}$.
    We also denote the set of assignments by $\assignmentset{S}$.
\end{definition}
\begin{definition}[Agreeing assignment]
    Define the set of agreeing assignments to the $k$-faces to be:
    \begin{equation*}
        \assignmentset{A} = \set{\assignment{F} \suchthat \exists \cochain{f} \in \cochainset{0}{\stdcomplex} \forall \stdface \in \stdcomplex(k): \cochain{f}|_\stdface = \assignment{F}^\stdface}
    \end{equation*}
    We also say that $\assignment{F}$ agrees with $\cochain{f}$.
\end{definition}
\begin{definition}[Distance function for assignments]
    Define the distance between two \\ $k$-assignments as:
    \begin{equation*}
        \dist{(\assignment{F}, \assignment{G})} = \norm{\set{\stdface \in \stdcomplex(k) \suchthat \assignment{F}^\stdface \ne \assignment{G}^\stdface}}
    \end{equation*}
    In addition, given a set of assignments $\assignmentset{S}$ define the distance of an assignment $\assignment{F}$ from $\assignmentset{S}$ to be:
    \[
        \dist{\parens{\assignment{F}, \assignmentset{S}}} = \min_{\assignment{s} \in \assignmentset{s}}{\set{\dist{(\assignment{F}, \assignment{S})}}}
    \]
\end{definition}
We are now ready to define $l$-assignments and agreeing $l$-assignments.
$l$-assignments are $l$ parallel assignments, i.e.\ every face in the complex has $l$ local functions associated with it.
\begin{definition}[$l$-assignments]\label{def:l-assignment}
    Given a simplicial complex $\stdcomplex$ define a $k$-dimensional $l$-assignment to be:
    \begin{equation*}
        \lassignment{F} = \set{\lassignment{F}^\stdface_i}_{\stdface \in \stdcomplex(k), i \in [l]}
    \end{equation*}
    Such that:
    \begin{equation*}
        \lassignment{F}^\stdface_i:\stdface \rightarrow \set{0,1}
    \end{equation*}
    An assignment of $l$ local function to each $k$-face of $\stdcomplex$.
\end{definition}
We define the distance between two $l$-assignments as the number of local functions on which they differ (normalised by the weights of the faces and the length of the list).
\begin{definition}[Distance between $l$-assignments]
    Define the distance between two $k$-dimensional $l$-assignments as:
    \begin{equation*}
        \dist{(\lassignment{F}, \lassignment{G})} = \sum_{\stdface \in \stdcomplex(k)}{\weight{\stdface}\frac{\abs{\set{i \in [l] \suchthat \lassignment{F}_i^\stdface \ne \lassignment{G}_i^\stdface}}}{l}}
    \end{equation*}
    In addition, given a set of $l$-assignments $\lassignmentset{S}$ define the distance of an $l$-assignment $\lassignment{F}$ from $\lassignmentset{S}$ to be:
    \[
        \dist{\parens{\lassignment{F}, \lassignmentset{S}}} = \min_{\lassignment{s} \in \lassignmentset{s}}{\set{\dist{(\lassignment{F}, \lassignment{S})}}}
    \]
\end{definition}
The notion of agreement in this case is more complicated then in the non-list case.
In regular assignments an agreeing assignment is an assignment that is consistent with some global function.
In the $l$-assignment agreement case we are interested whether there are $l$ global functions such that the local assignments of each vertex are a list of ``snippets'' of the $l$ global functions.
\begin{definition}[Agreeing $l$-assignment]
    Define an agreeing $k$-dimensional $l$-assignments to be a $k$-dimensional $l$-assignment $\lassignment{F}$ such that there are $l$ cochains $\cochain{F}_1,\dots,\cochain{F}_l \in \cochainset{0}{\stdcomplex}$ and for every face $\stdface \in \stdcomplex(k)$ there exists a permutation $\pi_\stdface$ such that the assignment $\assignment{F}_i = \set{\lassignment{F}^\stdface_{\pi_\stdface(i)}}_{\stdface \in \stdcomplex(k)}$ agrees with $\cochain{F}_i$.
    Denote the set of agreeing assignments as $\lassignmentset{A}$.
\end{definition}
We are now ready to define list agreement expansion:
\begin{definition}[List agreement expander]\label{def:list-agreement-expansion}
    Let $\stdcomplex$ be a $d$-dimensional pure simplical complex.
    We say that $\stdcomplex$ is a $(\beta, l)$-agreement-expander if there is a probabilistic algorithm $A$ such that for every dimension $k$ and every $k$-dimensional $l$-assignment $\lassignment{F}$ it holds that:
    \begin{equation*}
        \frac{\pr{A^\lassignment{F} \text{ rejects}}}{\dist\parens{\lassignment{F}, \lassignmentset{A}}} \ge \beta
    \end{equation*}
\end{definition}
In this paper we will present a weaker notion of list agreement expander.
Specifically we are going to assume that the $l$-assignments are $2$-locally-different defined below:
\begin{definition}[Locally differing $l$-assignment]
    Let $\lassignment{f}$ be an $l$-assignment.
    We say that $\lassignment{f}$ is $2$-locally-differing if for every $i \ne j$, every $\stdface \in \stdcomplex(k)$ there exists $x^{\stdface,i,j}_{1} \ne x^{\stdface,i,j}_{2}$ such that:
    \[
        \lassignment{f}^{\stdface}_i(x^{\stdface,i,j}_{1}) \ne \lassignment{f}^{\stdface}_j(x^{\stdface,i,j}_{1}) \text{ and } \lassignment{f}^{\stdface}_i(x^{\stdface,i,j}_{2}) \ne \lassignment{f}^{\stdface}_j(x^{\stdface,i,j}_{2})
    \]
\end{definition}
Before we conclude this section it will be helpful to consider the following property of the distance between an $l$-assignment and the agreeing $l$-assignments:
\begin{lemma}\label{distance-decomposition}
    For every set of permutations $\set{\pi_\stdface}_{\stdface \in \stdcomplex(k)}$ it holds that:
    \begin{equation*}
        \dist{(\lassignment{F}, \lassignmentset{A})} \le \frac{1}{l}\sum_{i=1}^l{\dist\parens{\set{\lassignment{F}^\stdface_{\pi_\stdface(i)}}_{\stdface \in \stdcomplex(k)}, \assignmentset{A}}}
    \end{equation*}
\end{lemma}
\begin{proof}
    Let $\set{\pi_\stdface}_{\stdface \in \stdcomplex(k)}$ be a set of permutations.
    Consider the assignments $\assignment{F}_i = \set{\lassignment{F}^\stdface_{\pi_\stdface(i)}}_{\stdface \in \stdcomplex(k)}$.
    For every $\assignment{F}_i$ let $\assignment{G}_i$ be an agreeing $k$-assignment such that $\dist\parens{\assignment{F}_i, \assignment{G}_i} = \dist\parens{\assignment{F}_i, \assignmentset{A}}$.
    Consider the $l$-assignment $\lassignment{G}=\set{\lassignment{G}^\stdface_i}_{\stdface \in \stdcomplex(k)}$ such that $\lassignment{G}^\stdface_{i} = \assignment{G}_{\pi_\stdface^{-1}(i)}^\stdface$.
    Note that $\lassignment{G}$ is an agreeing $k$-dimensional $l$-assignment since $\set{\lassignment{G}^\stdface_{\pi_\stdface(i)}}$ is an agreeing $k$-dimensional assignment for all $i$.
    We finish the proof by noting that:
    \begin{align*}
        \dist{(\lassignment{F}, \lassignmentset{A})}
        &\le \dist{(\lassignment{F}, \lassignment{G})}
        = \sum_{\stdface \in \stdcomplex(k)}{w(\stdface)\frac{\abs{\set{i \in [l] \suchthat \lassignment{F}_i^\stdface \ne \lassignment{G}_i^\stdface}}}{l}} \\
        & = \frac{1}{l}\sum_{\stdface \in \stdcomplex(k)}{w(\stdface)\abs{\set{i \in [l] \suchthat \lassignment{F}_{\pi_\stdface(i)}^\stdface \ne \lassignment{G}_{\pi_\stdface(i)}^\stdface}}} \\
        & = \frac{1}{l}\sum_{\stdface \in \stdcomplex(k)}{w(\stdface)\sum_{i=1}^l{\mathbbm{1}_{\set{\lassignment{F}_{\pi_\stdface(i)}^\stdface \ne \lassignment{G}_{\pi_\stdface(i)}^\stdface}}(i,\stdface)}} \\
        & = \frac{1}{l}\sum_{i=1}^l{\sum_{\stdface \in \stdcomplex(k)}{w(\stdface)\mathbbm{1}_{\set{\lassignment{F}_{\pi_\stdface(i)}^\stdface \ne \lassignment{G}_{\pi_\stdface(i)}^\stdface}}(i,\stdface)}} \\
        & = \frac{1}{l}\sum_{i=1}^l{\sum_{\stdface \in \set{\genface \in \stdcomplex(k) \suchthat \lassignment{F}_{\pi_\genface(i)}^\genface \ne \lassignment{G}_{\pi_\genface(i)}^\genface}}{w(\stdface)}} \\
        & = \frac{1}{l}\sum_{i=1}^l{\norm{\set{\stdface \in \stdcomplex(k) \suchthat \lassignment{F}_{\pi_\stdface(i)}^\stdface \ne \lassignment{G}_{\pi_\stdface(i)}^\stdface}}} \\
        & = \frac{1}{l}\sum_{i=1}^l{\norm{\set{\stdface \in \stdcomplex(k) \suchthat \assignment{F}_{i}^\stdface \ne \assignment{G}_{i}^\stdface}}}
        = \frac{1}{l}\sum_{i=1}^l{\dist\parens{\assignment{F}_i, \assignment{G}_i}}
        = \frac{1}{l}\sum_{i=1}^l{\dist\parens{\assignment{F}_i, \assignmentset{A}}}
    \end{align*}
\end{proof}
\begin{corollary}
    Let $\lassignment{F}$ be an $l$-assignment and let $\tilde{\lassignment{F}}$ be an agreeing $l$-assignment such that $\dist{(\lassignment{F}, \lassignmentset{A})} = \dist{(\lassignment{F}, \tilde{\lassignment{F}})}$.
    In addition let $\set{\pi_\stdface}_{\stdface \in \stdcomplex(k)}$ be a permutation such that for every $i$: $\set{\tilde{\lassignment{F}}^\stdface_{\pi_\stdface(i)}}_{\stdface \in \stdcomplex}$ is an agreeing assignment then:
    \begin{equation*}
        \dist{(\lassignment{F}, \lassignmentset{A})} = \frac{1}{l}\sum_{i=1}^l{\dist\parens{\set{\lassignment{F}^\stdface_{\pi_\stdface(i)}}_{\stdface \in \stdcomplex(k)}, \assignmentset{A}}}
    \end{equation*}
\end{corollary}
\begin{proof}
    Define $\lassignment{G}$ as in Lemma~\ref{distance-decomposition} and note that $\lassignment{G}=\tilde{\lassignment{F}}$.
    Also note that \\
    $\dist{(\lassignment{F}, \lassignmentset{A})} = \dist{(\lassignment{F}, \tilde{\lassignment{F}})} = \dist{(\lassignment{F}, \lassignment{G})}$.
    From here the proof follows the proof of Lemma~\ref{distance-decomposition}.
\end{proof}
\subsection{Coboundary Expansion}\label{subsec:co-boundary-expansion}
The first form of high dimensional expansion we will present is \emph{coboundary expansion}.
Coboundary expansion was first defined by Linial and Meshulam~\cite{linial2006homological} and independently by Gomov~\cite{gromov2010singularities}.
This form of expansion generalizes the notion of $1$-dimensional expansion naturally.
Before we define the coboundary expansion to higher dimensions, we will re-define the $0$-dimensional expansion in the language that will allow for easier generalization.
Consider the standard definition of Cheeger's constant:
\begin{definition}[Cheeger's constant]
    Let $X=(\stdcomplex,E)$ be a graph, define the Cheeger constant of a graph to be:
    \[
        \min_{A \ne V,\emptyset}{\set{\frac{\abs{E(A,V \setminus A)}}{\min{\abs{A}, \abs{V \setminus A}}}}}
    \]
\end{definition}
In our redefinition, instead of subsets of the vertices we will consider $0$-dimensional cochains with coefficients in $\mathbbm{F}_2$.
We will represent a set of vertices $A \subseteq V$ be the following cochain:
\[
    \stdcochain_A(\stdvertex) = \begin{cases}
                                    1    &\stdvertex \in A\\
                                    0    &\stdvertex \notin A
    \end{cases}
\]
Consider how, in our new setting, we can describe an edge between $A$ and $\bar{A}$ and note that the set of edges $\set{\stdvertex, \genvertex}$ that are between $A$ and $\bar{A}$ are exactly the edges for which $d_0\stdcochain_A(\stdvertex,\genvertex) \ne 0$.
Lastly, consider the following reformulation of the Cheeger's constant:
\[
    \min_{\cochain{f} \in \cochainset{0}{\stdcomplex} \setminus \coboundaryset{0}{\stdcomplex}}{\set{\frac{\norm{d_0 \cochain{f}}}{\dist\parens{\cochain{f}, \coboundaryset{0}{\stdcomplex}}}}}
\]
Note that $\coboundaryset{0}{\stdcomplex}$ contains the cochains that correspond to $V$ and $\emptyset$.
Now, consider the following generalization of that reformulation of Cheeger's constant:
\begin{definition}[$i$-dimensional expansion]
    Define the $i$\textsuperscript{th} dimensional Cheeger constant to be:
    \[
        h_i(\stdcomplex;G)= \min_{\cochain{f} \in \cochainset{i}{\stdcomplex;G} \setminus \coboundaryset{i}{\stdcomplex;G}}{\set{\frac{\norm{d_i \cochain{f}}}{\dist\parens{\cochain{f}, \coboundaryset{i}{\stdcomplex; G}}}}}
    \]
    If $h_i(\stdcomplex; G) \ge \epsilon$ we say that the $i$\textsuperscript{th} dimension of $\stdcomplex$ $\epsilon$-expands with $G$-coefficients.
\end{definition}
And the following generalization of expansion:
\begin{definition}[Coboundary expansion]
    Let $\stdcomplex$ be a $d$-dimensional simplicial complex.
    We say that $\stdcomplex$ is an $\epsilon$-coboundary expander with coefficients in $G$ if for every face $\stdface \in \stdcomplex$ with dimension smaller than $d-2$ it holds that $h_0(\stdcomplex_\stdface; G) \ge \epsilon$ and $h_1(\stdcomplex_\stdface; G) \ge \epsilon$.
\end{definition}
In~\cite{kaufman2021unique} Kaufman and Mass showed how to construct coboundary expanders independently of the underlying group.
\subsection{Covers of Simplicial Complexes}\label{subsec:covers-of-simplicial-complexes}
We will now move on to formally present the concept of covers of simplicial complexes and their connection to the cocycles of the simplicial complex:
\begin{definition}[Cover space]
    Let $\stdcomplex, \complex{y}$ be two $d$-dimensional simplicial complexes.
    We say that $\parens{\complex{y}, f_\complex{y}}$ $l$-fold evenly covers $\stdcomplex$ if $f_\complex{y}$ is a surjective map from $\complex{y}$ to $\stdcomplex$ such that:
    \begin{itemize}
        \item $f_\complex{y}$ is a translation function: for every $\stdface \in \complex{y}$ it holds that $\abs{\stdface} = \abs{f_\complex{Y}{(\stdface)}}$ and \\ $\stdface \subseteq \genface \Leftrightarrow f_\complex{y}(\stdface) \subseteq f_\complex{y}(\genface)$
        \item Locally $\stdcomplex$ and $\complex{y}$ look the same: for every face $\stdface \in \complex{y}$, $f_\complex{y}$ is an isomorphism between $\set{\genface \in \complex{y} \suchthat \stdface \subseteq \genface}$ and $\set{\genface \in \stdcomplex \suchthat f_\complex{y}(\stdface) \subseteq \genface}$.
        \item Every non-empty face of $\stdcomplex$ is covered by exactly $l$-faces from $\complex{y}$: $\forall \stdface \in \stdcomplex \setminus \set{\emptyset}: \abs{f_\complex{y}^{-1}(\stdface)} = l$
    \end{itemize}
    We will refer to $f_\complex{y}$ as the covering map from $\complex{y}$ to $\stdcomplex$.
    In addition we are going to say that $\complex{y}$ is an $l$-cover of $\stdcomplex$ if there exists a covering map from $\complex{y}$ to $\stdcomplex$.
    Lastly we would sometimes refer to a cover as defined here as a \emph{genuine} cover (as opposed to a near cover, defined below).
\end{definition}
\begin{definition}[Lift]
    Let $\stdcomplex$ be a simplicial complex, $\complex{Y}$ be a cover of $\stdcomplex$ with the covering map $f_\complex{y}$.
    Also let $\gamma = \parens{\gamma_1, \cdots, \gamma_n}$ be a path in $\stdcomplex$.
    A lift of $\gamma$ to $\complex{y}$ is a path $\gamma' = \parens{\gamma_1', \cdots, \gamma_n'}$ in $\complex{y}$ such that for all $i$ it holds that $f_\complex{y}(\gamma_i') = \gamma_i$.
\end{definition}
We also adopt the definition of near covers defined by~\cite{dinur2019near}:
\begin{definition}[Near cover]
    Let $\stdcomplex$ be a simplicial complex and let $\complex{y}$ be another simplicial complex.
    We say that $\complex{y}$ is an $l$ near cover of $\stdcomplex$ if there exists a function $f_\complex{y}:\complex{y}\rightarrow\stdcomplex$ such that:
    \begin{itemize}
        \item For each vertex $\stdvertex \in \stdcomplex(0)$ it holds that $f_\complex{y}^{-1}$ can be identified with $\sparens{l}$. 
        We will therefore use the notation $\sparens{u,i}$ to denote the vertex in $f_\complex{y}^{-1}\sparens{\genvertex}$ that corresponds to $i$.
        \item For every edge $(\stdvertex, \genvertex) \in \stdcomplex(1)$ there exists $\pi \in S_l$ such that if $f_\complex{y}(\sparens{\stdvertex, i}, \sparens{\genvertex, j}) = (\stdvertex, \genvertex)$ then $i=\pi\parens{j}$.
    \end{itemize}
    Where $S_l$ is the symmetric group of order $l$.
\end{definition}
Note that the difference between near covers and genuine covers is that in near covers the faces of dimension larger or equal to $2$ might not be properly covered.
Consider, for example, a complex that contains a single triangle $T=\set{\set{\stdvertex_1, \stdvertex_2, \stdvertex_3}}$ and consider the near cover of $T$ whose $1$-dimensional faces are (where $G=\mathbbm{F}_2$):
\[
    \complex{y}(1) = \set{\substack{\set{[\stdvertex_1,0], [\stdvertex_2,0]},\set{[\stdvertex_2,0], [\stdvertex_3,1]},\set{[\stdvertex_3,1], [\stdvertex_1,1]},\\ \set{[\stdvertex_1,1], [\stdvertex_2,1]},\set{[\stdvertex_2,1], [\stdvertex_3,0]},\set{[\stdvertex_3,0], [\stdvertex_1,0]}}}
\]
Note that it is indeed a near cover since every vertex $\stdvertex_i$ is covered by $[\stdvertex_1,0], [\stdvertex_1,1]$ which can be identified with $0$ and $1$ respectively.
In addition, since $\mathbbm{F}_2$ is the symmetric group the second condition holds trivially.
Also note that $\complex{y}$ is not a cover of $\complex{t}$ since the second and third conditions of being a cover fail:
The first condition fails for any vertex in the complex (since every vertex is a member of the triangle therefore $\set{\genface \in \stdcomplex \suchthat f_\complex{y}(\stdface) \subseteq \genface}$ contains a set of size 3 while $\set{\genface \in \complex{y} \suchthat \stdface \subseteq \genface}$ contains no such sets\footnote{$\complex{y}$'s $1$-dimensional faces form a cycle of length 6 and therefore no triangle can be formed using the edges of $\complex{y}$.}).
Let us now note that every $1$-dimensional cochain in $\stdcomplex$ implies a near cover by the following:
\begin{definition}
    Let $G$ be a group acting on the left of a set $S$ and let $\stdcochain \in \cochainset{1}{\stdcomplex;G}$.
    Define the complex $\complex{y}_\stdcochain$ as the complex whose $0$-dimensional faces are $\complex{y}_\stdcochain(0) = \set{[\stdvertex, s] \suchthat \stdvertex \in \stdcomplex \text{ and } s \in S}$ and its higher dimensional faces are:
    \[
        \complex{y}_\stdcochain(i) = \set{\set{[\stdvertex_0, s_0], \dots [\stdvertex_i, s_i]} \suchthat \set{\stdvertex_0,\dots,\stdvertex_i} \in \stdcomplex \text{ and } \forall i,j: s_i=\stdcochain(\stdvertex_i, \stdvertex_j)s_j}
    \]
\end{definition}
Surowski showed in~\cite[Proposition 3.2]{surowski1984covers} a characterization of when $\complex{y}_\stdcochain$ is a genuine cover of $\stdcomplex$.
Specifically:
\begin{lemma}[Proposition 3.2 in~\cite{surowski1984covers}, rephrased]
    $\complex{y}_\stdcochain$ is a genuine cover of $\stdcomplex$ with the covering map $f_{\complex{y}_\stdcochain}(\set{[\stdvertex_0, s_0], \dots [\stdvertex_i, s_i]}) = \set{\stdvertex_0,\dots,\stdvertex_i}$ iff $\stdcochain$ is a cocycle.
\end{lemma}
\subsection{Agreement Expansion}\label{subsec:agreement-expansion}
Another type of high dimensional expansion is agreement expansion.
This type of expansion is concerned with the relation between agreement of local assignments and the distance of an assignment from agreeing.
It was first introduced by Dinur and Kaufman in~\cite{dinur2017high}.
In this subsection we will introduce the concept of agreement testing which is a crucial part of our construction:
\begin{definition}[Agreement expansion in the $i$-th dimension]
    Let $\stdcomplex$ be a $d$-dimensional simplicial complex and let $\distribution{D}$ be a distribution over pairs of $i$ dimensional faces that intersect each other on $\intersectionsize$ vertices\footnote{$p$ may be dependent on $i$.}.
    Define:
    \[
        a_{i,\intersectionsize}(\stdcomplex, \distribution{D}) = \min_{\assignment{a} \in \assignmentset{S}}{\set{\frac{\prex{(\stdface, \genface) \sim \distribution{D}}{\assignment{a}^\stdface|_{\stdface \cap \genface} = \assignment{a}^\genface|_{\stdface \cap \genface}}}{\dist{\parens{\assignment{a}, \assignmentset{a}}}}}}
    \]
\end{definition}
In addition let us define an agreement expander:
\begin{definition}[Agreement expander]
    Let $\stdcomplex$ be a $d$-dimensional simplicial complex and let $\distribution{D}$ be a distribution.
    We say that $\stdcomplex$ is an $\agreementexpansionconst$-agreement-expander if:
    \[
        \forall 0 \le i \le d: a_{i,\intersectionsize}(\stdcomplex, \distribution{D}) \ge \agreementexpansionconst \parens{1-\frac{\intersectionsize}{i}}
    \]
\end{definition}
Note that the dependency on the $\parens{1-\frac{i}{\intersectionsize}}$ stems from the second eigenvalue of the random walk that walks between two $i$ dimensional faces via a $i+\intersectionsize$ dimensional face.

In this paper we will be interested in a special case of agreement expander, specifically $1$-up agreement expander.
\begin{definition}[$1$-up agreement expander]
    Define the distribution $\distribution{D}_{\uparrow}$ as the distribution that, in order to sample two $i$-dimensional faces, samples an $(i+1)$-dimensional face and then sample two of its $i$-dimensional sub-faces.
    Note that this distribution guarantees an intersection of size $i-1$.
    We say that a simplicial complex is a $1$-up agreement expansion if it is an agreement expander with the distribution $\distribution{D}_{\uparrow}$.
\end{definition}

\section{The Representation Complex}\label{sec:the-representation-complex}

In this section we are going to present a new complex that represents the agreement test over the complex we are interested in.
We will do that by constructing an edge for each of the possible choices of the 1-up agreement test.
We will add higher dimensional faces so that the norm of the faces in the representation complex would correspond with the weight of the faces they represent.
Then we will analyze the structure of this new complex.
We finish the discussion of the representation complex by discussing its expansion properties.

\begin{definition}[Representation Function]
    Define the representation function of a set of $k$-faces to be $R:P(\stdcomplex(k)) \rightarrow P(\stdcomplex(0))$ to be $R(s) = \bigcup_{s' \in s}{s'}$.
\end{definition}
\begin{definition}[Representation Complex]\label{the-representation-complex}
    Given a simplicial complex $\stdcomplex$ define the representation complex of $\stdcomplex$ to be $\repcplx{\stdcomplex}$ such that:
    \begin{itemize}
        \item $\repcplx{\stdcomplex}(-1) = \set{\emptyset}$
        \item $\repcplx{\stdcomplex}(0) = \set{\set{\stdface} \suchthat \stdface \in \stdcomplex(k)}$
        \item $\forall 1 \le i \le d-k: \repcplx{\stdcomplex}(i) = \set{\stdface \in \binom{\stdcomplex(k)}{i+1} \suchthat R(\stdface) \in \stdcomplex(i+k) \text{ and } \abs{\bigcap_{\stdvertex \in \stdface}{\stdvertex}} = k}$
    \end{itemize}
\end{definition}
\begin{note*}
    For clarity of notation we will treat $\repcplx{\stdcomplex}(0)$ as if it equals $\stdcomplex(k)$ (i.e.\ we will consider $\set{\stdface}$ and $\stdface$ to be equivalent when discussing the representation complex's $0$-dimensional faces).
    Moreover, we will sometimes treat $R$ as a function whose origin is the faces in $\repcplx[k]{\stdcomplex}$ (which are subsets of $\stdcomplex(k)$) and range is the faces of $\stdcomplex$.
\end{note*}
\begin{definition}
    We say that a face $\stdface \in \stdcomplex$ is represented by $r_\stdface \in \repcplx{\stdcomplex}$ if $R(r_\stdface) = \stdface$.
    Conversely we say that $r_\stdface$ represents $\stdface$.
\end{definition}
Now that we have defined the notion of representation in the new structure it is time to start analyzing it.
We will show that every face whose dimension is larger then $1$ has sunflower-like structure i.e.\ there is a core that is a subset of each of the vertices in the face. And, in addition, the core is the intersection of any two vertices in the face.
\begin{definition}
    Define the core of a face $\stdface \in \repcplx{\stdcomplex}$ to be: $\core{\stdface} = \bigcap_{\stdvertex \in \stdface}{\stdvertex}$
\end{definition}
\begin{lemma}
    \label{representation-has-a-core}
    Let $d \ge 1$ and let $\stdface \in \hat{R}(\stdcomplex(d))$, then every $ \stdvertex_1,\stdvertex_2 \in \stdface$ such that $\stdvertex_1 \ne \stdvertex_2$ it holds that: $\stdvertex_1 \cap \stdvertex_2 = \core{\stdface}$
\end{lemma}
\begin{proof}
    Using the counter positive argument assume that there exists $\stdvertex_1,\stdvertex_2$ such that $\stdvertex_1 \cap \stdvertex_2 \ne \core{\stdface}$. Since $\abs{\bigcap_{\stdvertex \in \stdface}{\stdvertex}} = k$ we know that $\stdvertex_1 \cap \stdvertex_2 \supseteq \bigcap_{\stdvertex \in \stdface}{\stdvertex} = \core{\stdface}$.
    We also know that $\abs{\stdvertex_1} = \abs{\stdvertex_2} = k+1$ therefore if $\abs{\stdvertex_1 \cap \stdvertex_2} = k+1$ then $\stdvertex_1 = \stdvertex_2$ which contradicts our assumption.
    In the case where $\abs{\stdvertex_1 \cap \stdvertex_2} = k$ it holds that $\abs{\stdvertex_1 \cap \stdvertex_2} = \abs{\core{\stdface}}$ which proves the lemma.
\end{proof}
We can now prove that the representation complex is indeed a simplicial complex:
\begin{lemma}[The representation complex is a simplicial complex]
    The representation complex is a simplicial complex.
\end{lemma}
\begin{proof}
    Let $\stdface \in \repcplx{\stdcomplex}$ be a face and let $\tilde{\stdface}$ be a subset of $\stdface$.
    If $\abs{\tilde{\stdface}}=1$ then $\tilde{\stdface}$ is exactly one of the vertices and therefore in $\repcplx{\stdcomplex}$.
    In any other case $R(\tilde{\stdface}) \subseteq R(\stdface) \in \stdcomplex$ and therefore $R(\tilde{\stdface}) \in \stdcomplex$.
    In addition it is easy to see that due to Lemma~\ref{representation-has-a-core} it holds that $\abs{\bigcap_{\genface \in \stdface}{\genface}} = k$.
    Therefore $\bar{\stdface} \in \repcplx{\stdcomplex}$ and the representation complex is a simplicial complex.
\end{proof}

\subsection{On the Structure of the Representation Complex}\label{subsec:on-the-structre-of-the-representation-complex}
We will now discuss the structure of the representation complex including how it represents the original complex.
We will begin by considering how each face is represented in the representation complex.
Specifically we will show that every face is represented in the representation complex once for every one of its cores (which are the $\parens{k-1}$-faces contained in it).
\begin{lemma}
    \label{singular-representation-per-core}
    Let $\stdface$ be a face in $\stdcomplex(k+i)$ (such that $i \ge 1$).
    And let $c\in \binom{\stdface}{k}$ be a possible core of $\stdface$ (i.e.\ a $(k-1)$-dimensional sub-face of $\stdface$).
    Then there exists a single face $\rep{c}{\stdface}$ such that $\stdface$ is represented by $\rep{c}{\stdface}$ and $\core{\rep{c}{\stdface}} = c$.
\end{lemma}
\begin{proof}
    Consider $\rep{c}{\stdface} = \set{c \cup \set{\stdvertex} \suchthat \stdvertex \in \stdface \setminus c}$.
    We will start by proving that this is indeed a face in the representation complex:
    \begin{itemize}
        \item $c \cup \set{\stdvertex} \subseteq \stdface$ and, in addition, $\abs{c \cup \set{\stdvertex}} = k+i+1$ therefore $c \cup \set{\stdvertex} \in \stdcomplex(k)$.
        \item $\abs{\rep{c}{\stdface}} = \abs{\stdface \setminus c} = i+1$.
        \item It is easy to see that $R(\rep{c}{\stdface}) = \stdface \in \stdcomplex(k+i)$.
        \item $\abs{\bigcap_{\stdvertex \in \stdface}{\stdvertex}} = \abs{c} = k$.
    \end{itemize}
    Therefore $\rep{c}{\stdface}$ is a face in the representation complex which represents $\stdface$.
    We will use the counter-positive argument in order to prove that $\rep{c}{\stdface}$ is singular up to the core $c$.
    Assume that there are two different faces $\tilde{\rep{c}{\stdface}}, \hat{\rep{c}{\stdface}}$ that represent $\stdface$ and have the same core $c$.
    WLOG assume that $\tilde{\rep{c}{\stdface}} {\displaystyle \not \subseteq } \hat{\rep{c}{\stdface}}$  and let $\genface \in \hat{\rep{c}{\stdface}} \setminus \tilde{\rep{c}{\stdface}}$.
    Let $\stdvertex$ be the vertex such that $\stdvertex \in \genface \setminus c$.
    Because $\tilde{\rep{c}{\stdface}}$ and $\hat{\rep{c}{\stdface}}$ represent the same face there must be some $\genface' \in \tilde{\rep{c}{\stdface}}$ such that $\stdvertex \in \genface'$.
    Since $\stdvertex \notin c$ the only candidate for $\genface'$ is $c \cup \set{\stdvertex}$ (since $\abs{\genface'}$ must be $k+1$).
    Therefore $\genface=\genface'\in \tilde{\rep{c}{\stdface}}$ which contradicts our assumption on $\genface$.
\end{proof}
\begin{notation}
    Denote by $\repex{k}{c}{\stdface}$ the face in the representation complex of the $k$-dimensional faces whose core is $c$ and represents $\stdface$.
    In most cases the dimension will be clear from context we would omit the dimension from the notation and denote the face by $\rep{c}{\stdface}$.
\end{notation}
The following three corollaries follow directly from the characterization of the representation of a $k$-dimensional face:
\begin{corollary}
    \begin{equation*}
        R^{-1}(\stdface) = \set{\rep{c}{\stdface} \suchthat c \in \binom{\stdface}{k}}
    \end{equation*}
\end{corollary}
\begin{corollary}
    For every face $\stdface \in \stdcomplex(k+i)$ ($i > 0$) and every two different possible cores $c_1, c_2 \in \binom{\stdface}{k}$ such that $c_1 \ne c_2$ it holds that $\rep{c_1}{\stdface} \ne \rep{c_2}{\stdface}$
\end{corollary}
\begin{proof}
    Assuming that $\rep{c_1}{\stdface} = \rep{c_2}{\stdface}$, then for every face in $\genface \in \rep{c_1}{\stdface}$ it holds that $c_1 \cup c_2 \subseteq \genface$ which contradicts the fact that the cardinality of the core is exactly $k$.
\end{proof}
\begin{corollary}
    For every $\dimension{i}>0$ and every face $\stdface \in \stdcomplex(k+i)$:
    \begin{equation*}
        \abs{R^{-1}(\stdface)} = \begin{cases}
        \binom{k+i+1}{k} & i>0\\
        1                & i = 0
        \end{cases}
    \end{equation*}
\end{corollary}
Now that we understand the representation better we move on to discuss the norm of the representation complex.
Although we discuss the norm in every dimension of the representation complex, the $1$-dimensional faces of the representation complex are of particular interest.
Specifically we will show that sampling a $1$-dimensional face in the representation complex (with the appropriate norm) is equivalent to the choice that the 1-up agreement test preforms.
We will start the discussion by showing how the weight function in the representation complex relates to the weight function of the original complex:
\begin{lemma}
    \label{weight-relation-between-a-complex-and-its-representation}
    For every face $\stdface \in \stdcomplex(k+i)$ and every representation of the face $\rep{c}{\stdface}$ it holds that:
    \begin{equation*}
        \weightcplx{{\repcplx{\stdcomplex}}}{\rep{c}{\stdface}} = \frac{1}{\abs{R^{-1}(\stdface)}}\weightcplx{\stdcomplex}{\stdface}
    \end{equation*}
\end{lemma}
\begin{proof}
    For $i > 0$:
    \begin{align*}
        \weightcplx{\repcplx{\stdcomplex}}{\rep{c}{\stdface}} & = \frac{\abs{\set{\genface \in \repcplx{\stdcomplex}(d-k) \suchthat \rep{c}{\stdface} \subseteq \genface}}}{\binom{d-k+1}{i+1} \cdot \abs{\repcplx{\stdcomplex}(d-k)}}
        = \frac{\abs{\set{\genface \in \stdcomplex(d) \suchthat \stdface \subseteq \genface}}}{\binom{d-k+1}{i+1} \binom{d+1}{k} \abs{\stdcomplex(d)}} = \\
        & = \frac{\abs{\set{\genface \in \stdcomplex(d) \suchthat \stdface \subseteq \genface}}}{\binom{k+i+1}{k} \binom{d+1}{k+i+1} \abs{\stdcomplex(d)}}
        = \frac{1}{\binom{k+i+1}{k}} \cdot \frac{\abs{\set{\genface \in \stdcomplex(d) \suchthat \stdface \subseteq \genface}}}{\binom{d+1}{k+i+1} \abs{\stdcomplex(d)}} = \\
        & = \frac{1}{\binom{k+i+1}{k}} \cdot \weightcplx{\stdcomplex}{\stdface}
        = \frac{1}{\abs{R^{-1}(\stdface)}}\weightcplx{\stdcomplex}{\stdface}
    \end{align*}
    The second equality is due to the fact that every maximal face in the original complex is represented by $\binom{d+1}{k}$ faces in the representation complex.
    In addition, the third equality is due to Lemma~\ref{combinatorial-bounds-1}.
    For $i = 0$:
    \begin{align*}
        \weightcplx{\repcplx{\stdcomplex}}{\rep{c}{\stdface}} & = \frac{\abs{\set{\genface \in \repcplx{\stdcomplex}(d-k) \suchthat \rep{c}{\stdface} \subseteq \genface}}}{\binom{d-k+1}{1} \cdot \abs{\repcplx{\stdcomplex}(d-k)}}
        = \frac{\binom{k+1}{k}\abs{\set{\genface \in \stdcomplex(d) \suchthat \stdface \subseteq \genface}}}{\binom{d-k+1}{1} \binom{d+1}{k} \abs{\stdcomplex(d)}} = \\
        & = \frac{\abs{\set{\genface \in \stdcomplex(d) \suchthat \stdface \subseteq \genface}}}{\parens{\frac{d-k+1}{k+1} \binom{d+1}{k}} \abs{\stdcomplex(d)}} = \\
        & = \frac{\abs{\set{\genface \in \stdcomplex(d) \suchthat \stdface \subseteq \genface}}}{\binom{d+1}{k+1} \abs{\stdcomplex(d)}} = \weightcplx{\stdcomplex}{\stdface} = \frac{1}{\abs{R^{-1}(\stdface)}}\weightcplx{\stdcomplex}{\stdface}
    \end{align*}
    The second equality holds because the representation of a $k$-dimensional face $\stdface$ is contained in the representation of any face in $\containment\parens{\stdface}$.
    Therefore if $\stdface$ is contained in some maximal face $\genface$ then $\rep{c}{\stdface}$ is contained in the $\binom{k+1}{k}$ representations of $\genface$ whose cores lie in $\stdface$.
\end{proof}
\begin{corollary}
    $\norm{R^{-1}(\stdface)} = \norm{\set{\stdface}}$
\end{corollary}
Now let us consider how one can sample a face (of any dimension) in the representation complex with probability relative to its weight.
Consider the following algorithm:\\
\begin{algorithm}[H]
    \caption{sample from the representation complex}\label{alg:sample-from-the-representation-complex}
    \SetAlgoLined
    \DontPrintSemicolon
    Sample $\genface \in \stdcomplex(k+i)$ according to the norm of $\stdcomplex$.\\
    Pick a set  $c \in \binom{\genface}{k}$ uniformly at random.\label{alg:sample-from-the-representation-complex:sample-core}\\
    Return $r^c_\genface$
\end{algorithm}
\begin{lemma}
    \label{sampling-from-representation-complex}
    One can sample from the representation complex according to its norm.
\end{lemma}
\begin{proof}
    We will show that Algorithm~\ref{alg:sample-from-the-representation-complex} samples faces of dimension $i$ according to their weight.
    Let $\rep{c'}{\stdface}$ be a face of dimension $i \ge 1$ in the representation complex.
    Then:
    \begin{align*}
        \pr{\text{Algorithm~\ref{alg:sample-from-the-representation-complex} outputs }\rep{c}{\stdface}} &= \pr{\genface = \stdface \wedge c = c'}
        = \pr{\genface = \stdface \wedge c
        = c' \suchthat \genface = \stdface} \pr{\genface = \stdface}\\
        &= \pr{c = c' \suchthat \genface = \stdface} \pr{\genface = \stdface}
        = \frac{1}{\binom{k+i+1}{k}} \weightcplx{\stdcomplex}{\stdface}
        = \frac{1}{\abs{R^{-1}(\stdface)}} \weightcplx{\stdcomplex}{\stdface}\\
        &= \weightcplx{\repcplx{\stdcomplex}}{\rep{c}{\stdface}}
    \end{align*}
    The first equality is due to Lemma~\ref{singular-representation-per-core} and the last equality is due to Lemma~\ref{weight-relation-between-a-complex-and-its-representation}.\\
    If $\rep{c'}{\stdface}$ is a face of dimension $0$ then step~\ref{alg:sample-from-the-representation-complex:sample-core} can be ignored since the representation for these faces is independent of the core chosen and the algorithm outputs $\tau$.
\end{proof}
Consider now what happens when Algorithm~\ref{alg:sample-from-the-representation-complex} samples a $1$-dimensional face in the representation complex.
It samples a $(k+1)$-dimensional face $\genface$ in the original complex and then chooses a core.
Choosing a core is equivalent to choosing two $k$-dimensional sub-faces of $\genface$ (since any core of $\genface$ is shared by exactly two faces).
Note that this is the exact process in which the $1$-up agreement test samples its faces.
Therefore picking a $1$-dimensional face in the representation complex corresponds to a choosing faces in the $1$-up agreement test.

\subsection{On the Expansion Properties of the Representation Complex}\label{subsec:on-the-expansion-properties-of-the-representation-complex}
Now that we have defined the representation complex and presented how it represents the original complex we will move on to discuss the expansion properties of the representation complex.
We will show that the representation complex, despite not being an expander by itself, does have some expansion properties.
Specifically we are interested in the expansion of the faces of the first dimension since cocycles in the first dimension correspond with covers of the complex.
In order to discuss the expansion of the representation complex it will be useful to discuss how faces that share a core behave.
Lets begin by discussing the representation complex around a core:
\begin{definition}[Representation complex around a core]
    Define the representation complex of $\stdcomplex$ around a core $c$ to be:
    \begin{equation*}
        \repcplxcore[k]{c}{\stdcomplex} = \set{\stdface \in \repcplx{\stdcomplex}(i) \suchthat \core{\stdface} = c, i \ge 1} \cup \set{\stdface \in \repcplx{\stdcomplex}(0) \suchthat c \subseteq \stdface} \cup \set{\emptyset}
    \end{equation*}
\end{definition}
This is indeed a simplicial complex since it is constructed from picking a set of faces of maximal dimension from $\repcplx{\stdcomplex}$ and all their sub-faces.
We will now prove that around every core the representation complex is a coboundary expander.
We will do that by characterizing them as the $\parens{k-1}$-dimensional links of the original complex:
\begin{lemma}
    \label{the-representation-complex-around-a-core-is-isomorphic-to-links}
    For every core of every face $c \in \stdcomplex(k-1)$ it holds that $\stdcomplex_c \cong \repcplxcore[k]{c}{\stdcomplex}$.
\end{lemma}
\begin{proof}
    Consider the maps $f:\repcplxcore[k]{c}{\stdcomplex} \rightarrow \stdcomplex_c$ and $g: \stdcomplex_c \rightarrow \repcplxcore[k]{c}{\stdcomplex}$ defined as follows:
    \begin{equation*}
        f(\genface) = R(\genface) \setminus c
        \quad\text{and}\quad
        g(\stdface) = \rep{c}{\parens{\stdface \cup c}}
    \end{equation*}
    Note that:
    \begin{equation*}
        \forall \genface \in \repcplxcore[k]{c}{\stdcomplex}: g(f(\genface)) = g(R(\genface) \setminus c) = \rep{c}{\parens{\parens{R(\genface) \setminus c} \cup c}} = \rep{c}{R(\genface)} = \genface
    \end{equation*}
    \begin{equation*}
        \forall \stdface \in \stdcomplex_c: f(g(\stdface)) = f(\rep{c}{\parens{\stdface \cup c}}) = \parens{\stdface \cup c} \setminus c = \stdface
    \end{equation*}
    The last equality in the first line is due to the singular representation of $\genface$ per core.
    The last equality in the second line holds because the core of all the faces in $\repcplxcore[k]{c}{\stdcomplex}$ is $c$.
    In addition:
    \begin{itemize}
        \item If $\genface_1 \subseteq \genface_2$ then: $f(\genface_1) = R(\genface_1) \setminus c = \bigcup_{t\in \genface_1}{t \setminus c} \subseteq \bigcup_{t\in \genface_2}{t \setminus c} = R(\genface_2) \setminus c = f(\genface_2)$.
        \item If $\stdface_1 \subseteq \stdface_2$ then: $g(\stdface_1) = \rep{c}{\parens{\stdface_1 \cup c}} = \set{\set{i} \cup c \suchthat i \in \stdface_1} \subseteq \set{\set{i} \cup c \suchthat i \in \stdface_2} = \rep{c}{\parens{\stdface_2 \cup c}} = g(\genface_2)$.
    \end{itemize}
\end{proof}
\begin{note*}
    In the proof of the previous lemma we disregarded the orientation of the faces (which we are yet to define for the representation complex).
    We will define the orientation of the representation complex to be the orientation that stems from this isomorphism.
    It is important to notice that the orientation of the representation complex is completely determined by the orientation of the faces in the original complex.
\end{note*}
We will continue by showing that the norm of a face in the representation complex is directly proportional to its weight around each core:
\begin{lemma}
    For every core $c$ there exists a constant $t_c$ such that for every face $\stdface$ whose core is $c$ it holds that $\weight{\stdface}_{\repcplx{\stdcomplex}}=t_c \cdot \weight{\stdface}_{\hat{R}_c(\stdcomplex)}$.
\end{lemma}
\begin{proof}
    \begin{align*}
        \weight{\stdface}_{\repcplx{\stdcomplex}}&=\frac{\abs{\set{\stdface' \in \repcplx{\stdcomplex}(d-k) \suchthat \stdface \subseteq \stdface'}}}{\binom{{d-k+1}}{\abs{\stdface}}\abs{\repcplx{\stdcomplex}(d-k)}}\\
        &=\frac{\abs{\hat{R}_c(\stdcomplex)(d-k)}}{\abs{\repcplx{\stdcomplex}(d-k)}} \cdot \frac{\abs{\set{\stdface' \in \hat{R}_c(\stdcomplex)(d-k) \suchthat \stdface \subseteq \stdface'}}}{\binom{d-k+1}{\abs{\stdface}}\abs{\hat{R}_c(\stdcomplex)(d-k)}}\\
        &=\frac{\abs{\hat{R}_c(\stdcomplex)(d-k)}}{\abs{\repcplx{\stdcomplex}(d-k)}} \weight{\stdface}_{\hat{R}_c(\stdcomplex)}
    \end{align*}
    The second equality holds due to Lemma~\ref{representation-has-a-core}
\end{proof}
Now that we have characterized the representation complex around each of its cores it is time to move on to discuss its expansion properties.
We will show how to deconstruct each cochain $\stdcochain$ to a multiplication of other cochains $\stdcochain^c$ such that each of the cochains $\stdcochain^c$ is restricted, in a sense, around a core.
This deconstruction is going to play a key role in the proof that the representation complex expands in the first dimension.
Specifically, we know that around every core the complex is an expander and therefore each of the $\stdcochain^c$s expands.
This implies that the entire cochain expands.
Let us start be defining the deconstruction:
\begin{definition}[Multiplication of cochains]
    Let $\stdcochain_1, \stdcochain_2 \in \cochainset{i}{\stdcomplex;G}$ define the multiplication of them as:
    \begin{equation*}
        (\stdcochain_1 \cdot \stdcochain_2)(\stdface) = \stdcochain_1(\stdface) \cdot \stdcochain_2(\stdface)
    \end{equation*}
\end{definition}
\begin{definition}[Cochain around a core]
    \label{co-chain-around-core}
    For every dimension $i$ and every cochain \\
    $\stdcochain \in \cochainset{i}{\repcplx{\stdcomplex};G}$ define the cochain $\stdcochain$ around $c$ to be:
    \begin{equation*}
        \stdcochain^c(\stdface) = \begin{cases}
                                      \stdcochain(\stdface) & core\parens{\stdface}=c\\
                                      1 & \text{otherwise}
        \end{cases}
    \end{equation*}
\end{definition}
\begin{lemma}
    For every cochain $\stdcochain \in \cochainset{i}{\repcplx{\stdcomplex};G}$ ($i \ge 1$) it holds that:
    \begin{equation*}
        \stdcochain = \prod_{c \in \stdcomplex(k-1)}{\stdcochain^c}
    \end{equation*}
    In addition, for every $\stdface \in \repcplx{\stdcomplex}(i)$ there exists at most one core $c$ such that $\stdcochain^c(\stdface) \ne 1$ (therefore the multiplication order is irrelevant).
\end{lemma}
\begin{proof}
    Every face of dimension larger than $1$ has a single core therefore $\stdcochain(\stdface) \ne 1 \Rightarrow \stdcochain^{\core{\stdface}}(\stdface) \ne 1$.
    And for every other possible core it holds that $\stdcochain^{c}(\stdface) = 1$.
    In addition:
    \begin{equation*}
        \stdcochain(\stdface) = \stdcochain^{\core{\stdface}}(\stdface)=\prod_{c \in \stdcomplex(k-1)}{\stdcochain^c(\stdface)}
    \end{equation*}
\end{proof}
\begin{corollary}
    For every cochain of dimension at least $1$: $\norm{\stdcochain} = \sum_{c\in\stdcomplex(k-1)}{\norm{\stdcochain^c}}$
\end{corollary}
We will now show that the deconstruction of a cochain to a set of cochains is, in a sense, independent.
We will start by showing that the distance of a cochain from the coboundaries can be derived by measuring the distance of the cochain from the coboundaries around every core:
\begin{lemma}
    \label{distance-core-decomposition}
    Let $\stdcochain \in \cochainset{i}{\repcplx[k]{\stdcomplex};G}$.
    If $\cochain{A}=\argmin_{\cochain{B}\in \cocycleset{i}{\repcplx[k]{\stdcomplex};G}}{\norm{\stdcochain\cochain{B}^{-1}}}$ then for every core $c$ it holds that $\cochain{A}^c=\argmin_{\cochain{B}\in \cocycleset{i}{\repcplxcore[k]{c}{\stdcomplex};G}}{\norm{\stdcochain^c\cochain{B}^{-1}}}$
\end{lemma}
\begin{proof}
    It is easy to see that for every possible core $c$ and any two cochains $\stdcochain_1, \stdcochain_2$ it holds that $(\stdcochain_1 \cdot \stdcochain_2)^c = \stdcochain_1^c \cdot \stdcochain_2^c$ and $\parens{\stdcochain^{-1}}^c = \parens{\stdcochain^c}^{-1}$.
    In addition:
    \begin{equation*}
        \norm{\stdcochain\cochain{A}^{-1}}
        = \sum_{c \in \stdcomplex(k-1)}{\norm{\parens{\stdcochain\cochain{A}^{-1}}^c}}
        = \sum_{c \in \stdcomplex(k-1)}{\norm{\stdcochain^c\parens{\cochain{A}^{-1}}^c}}
        = \sum_{c \in \stdcomplex(k-1)}{\norm{\stdcochain^c\parens{\cochain{A}^{c}}^{-1}}}
    \end{equation*}
    We will complete the proof using the counter positive argument: Assume that there exists a core $c'$ in which $\cochain{L} = \argmin_{\cochain{B}\in \cocycleset{i}{\repcplxcore[k]{c'}{\stdcomplex};G}}{\norm{\stdcochain^{c'}\cochain{B}^{-1}}}$ and $\norm{\stdcochain^{c'}\cochain{L}^{-1}} < \norm{\stdcochain^{c'}\parens{\cochain{A}^{c'}}^{-1}}$.
    Therefore consider the following cochain:
    \begin{equation*}
        \cochain{A}'(\stdface) = \begin{cases}
                                     \cochain{A}^{\core{\stdface}}(\stdface)      &\core{\stdface} \ne c'\\
                                     \cochain{L}(\stdface)             &\text{otherwise}
        \end{cases}
    \end{equation*}
    Note that every face $\genface \in \repcplx{\stdcomplex}(i+1)$ is completely contained in $\repcplxcore[k]{\core{\genface}}{\stdcomplex}$ and therefore if $\core{\genface} \ne c'$ then $d_1\cochain{A}'(\genface) = d_1\cochain{A}^{\core{\genface}}(\genface)=1$ and otherwise $d_1\cochain{A}'(\genface) = d_1\cochain{L}(\genface)=1$.
    Therefore $\cochain{A}'$ is a cocycle.
    Also note that:
    \begin{align*}
        \norm{\stdcochain\cochain{A}'^{-1}}
        & = \sum_{c \in \stdcomplex(k-1)}{\norm{\parens{\stdcochain\cochain{A}'^{-1}}^c}}
        = \sum_{c \in \stdcomplex(k-1) \setminus \set{c'}}{\norm{\stdcochain \parens{\cochain{A}^c}^{-1}}} + \norm{\stdcochain\cochain{L}^{-1}}\\
        &< \sum_{c \in \stdcomplex(k-1)}{\norm{\parens{\stdcochain\cochain{A}^{-1}}^c}}
        = \norm{\stdcochain\cochain{A}^{-1}}
    \end{align*}
    Which contradicts our choice of $\cochain{A}$.
\end{proof}
\begin{lemma}[The representation around every core is an expander]
    \label{the-representation-around-every-core-is-an-expander}
    For every core $c$, $\repcplxcore[k]{c}{\stdcomplex}$ is a $\gamma$-coboundary-expander.
\end{lemma}
\begin{proof}
    The representation complex around the core $c$ is isomorphic to the link of $c$ in the original complex due to Lemma~\ref{the-representation-complex-around-a-core-is-isomorphic-to-links}.
    $\stdcomplex_c$ is a $\gamma$-coboundary-expander and therefore so is $\repcplxcore[k]{c}{\stdcomplex}$.
\end{proof}
\begin{lemma}\label{co-boundary-operator-around-a-core}
    For every core $c$ and every cochain $\stdcochain$ it holds that $\parens{d_1\stdcochain}^c = d_1\parens{\stdcochain^c}$
\end{lemma}
\begin{proof}
    Let $\stdcochain \in \cochainset{i}{\stdcomplex;G}$ for $i \ge 1$:
    \begin{align*}
        \forall \stdface \in \stdcomplex(i+1): \parens{d_1\stdcochain}^c(\stdface)
        & = \begin{cases}
                d_1\stdcochain(\stdface)   & \core{\stdface} = c\\
                1                             & \text{otherwise}
        \end{cases}
        = \begin{cases}
              d_1\stdcochain(\stdface)   & \core{\stdface} = c\\
              d_1\mathbbm{1}(\stdface)    & \text{otherwise}
        \end{cases}\\
        &= d_1 \begin{cases}
                      \stdcochain(\stdface)   & \core{\stdface} = c\\
                      \mathbbm{1}(\stdface)    & \text{otherwise}
        \end{cases} =
        d_1\parens{\stdcochain^c}
    \end{align*}
\end{proof}
\begin{lemma}[The representation complex expands in the first dimension]
    \label{representation-complex-expands-in-first-dimension}
    For every cochain $\stdcochain \in \cochainset{1}{\repcplx{\stdcomplex};G}$:
    \begin{equation*}
        \frac{\norm{d_1 \stdcochain}}{\dist\parens{\stdcochain, \cocycleset{1}{\repcplx{\stdcomplex}_\stdface;G}}} \ge \gamma
    \end{equation*}
\end{lemma}
\begin{proof}
    Due to Lemma~\ref{the-representation-around-every-core-is-an-expander} around every core the representation complex is a $\gamma$-coboundary-expander.
    Let $\cochain{A}=\argmin_{\cochain{B}\in \cocycleset{i}{\repcplxcore[k]{c}{\stdcomplex};G}}{\norm{\stdcochain\cochain{B}^{-1}}}$.
    \begin{align*}
        \norm{d_1 \stdcochain} &= \sum_{c \in \stdcomplex(k-1)}{\norm{d_1 \stdcochain^c}}
        = \sum_{c \in \set{cores}}{t_c\norm{d_1\stdcochain^c}_{\repcplxcore[k]{c}{\stdcomplex}}} \\
        &\ge \gamma \sum_{c \in \stdcomplex(k-1)}{t_c \dist_{\repcplxcore[k]{c}{\stdcomplex}}\parens{\stdcochain^c, \cocycleset{i}{\repcplxcore[k]{c}{\stdcomplex};G}}}
        = \gamma \sum_{c \in \stdcomplex(k-1)}{t_c \norm{\cochain{A}^c}_{\repcplxcore[k]{c}{\stdcomplex}}} \\
        &= \gamma \sum_{c \in \stdcomplex(k-1)}{\norm{\cochain{A}^c}} = \gamma\norm{\cochain{A}}
        = \gamma \dist\parens{\stdcochain, \cocycleset{i}{\repcplxcore[k]{c}{\stdcomplex}}}
    \end{align*}
\end{proof}
We finish this section by noting that, despite the fact that the representation complex's cohomology is not trivial, the distance of a cochain from the coboundaries can still be tested.
This test will be further explored in Appendix~\ref{sec:on-the-testability-of-co-boundaries-in-the-representation-complex} where we prove the following lemma:
\begin{lemma}[The Coboundaries of the Representation Complex are Testable]\label{co-boundaries-in-representation-complex-are-testable}
    If $\stdcomplex$ is a $\gamma$-coboundary-expander then there exists a tester $T$ that queries exactly $3$ edges\footnote{We can claim something even stronger. Specifically that $T$ picks three vertices and queries the edges between them.} and a constant $\cobdrtestconst = \cobdrtestconst(k, \gamma)$ such that:
    \[
        \dist\parens{\stdcochain, \coboundaryset{1}{\repcplx{\stdcomplex}}} \le \cobdrtestconst\cdot \pr{T \text{ rejects } \stdcochain}
    \]
\end{lemma}
\section{On the Covers That Correspond to Coboundaries}\label{sec:on-the-covers-of-complexes-that-represent-co-boundaries}
Recall the connection between cocycles and cover spaces of a simplicial complex.
Specifically, recall that $Y_\stdcocycle$ is a cover space of $\stdcomplex$ iff $\stdcocycle$ is a cocycle.
In this section we will be interested in the properties of cover spaces that corresponds with coboundaries.
We start by showing that, if $\stdcocycle$ is a coboundary then the lift of any cycle $C$ in $\stdcomplex$ is a cycle in $Y_\stdcocycle$:
\begin{lemma}[Lifts of cycles are cycles]
    Let $\stdcoboundary \in \coboundaryset{1}{\stdcomplex;G}$ and let $Y_\stdcoboundary$ be a cover of $\stdcomplex$ with the simplicial map $f:Y_\stdcoboundary \rightarrow X$.
    Also let $C=(\stdvertex_0,\dots,\stdvertex_i,\stdvertex_0)$ be a vertex-edge cycle in $\stdcomplex$.
    Then every lift of $C$ is a cycle in the vertex in $Y_\stdcoboundary$.
\end{lemma}
\begin{proof}
    Consider the lift of the cycle to the cover $Y_\stdcoboundary$ and note that it is comprised of the following edges (for some $s \in G$):
    \begin{align*}
        \set{[\stdvertex_0, s], [\stdvertex_1, \stdcoboundary(\stdvertex_1, \stdvertex_0)s]},
        \set{[\stdvertex_1, \stdcoboundary(\stdvertex_1, \stdvertex_0)s], [\stdvertex_2, \stdcoboundary(\stdvertex_0, \stdvertex_1)\stdcoboundary(\stdvertex_1, \stdvertex_0)s]}, \\
        \dots ,
        \set{[\stdvertex_i, \prod_{j=0}^{i-1}{\stdcoboundary(\stdvertex_{j+1}, \stdvertex_j)}s], [\stdvertex_0, \stdcoboundary(\stdvertex_0, \stdvertex_i)\prod_{j=0}^{i-1}{\stdcoboundary(\stdvertex_{j+1}, \stdvertex_j)}s]}
    \end{align*}
    Because $\stdcoboundary$ is a coboundary there exists a $\psi\in\cochainset{0}{\stdcomplex}$ such that $\stdcoboundary(\stdface_1, \stdface_2) = \psi(\stdface_1)\psi(\stdface_1)^{-1}$ therefore:
    \begin{align*}
        \stdcoboundary(\stdvertex_0, \stdvertex_i)\prod_{j=0}^{i-1}{\stdcoboundary(\stdvertex_{j+1}, \stdvertex_j)} &=
        \psi(\stdvertex_0)\psi(\stdvertex_i)^{-1}\prod_{j=0}^{i-1}{\psi(\stdvertex_{j+1})\psi(\stdvertex_j)^{-1}} \\
        & = \psi(\stdvertex_0)\psi(\stdvertex_i)^{-1}\psi(\stdvertex_i)\psi(\stdvertex_0)^{-1}
        = 1
    \end{align*}
    Therefore the path starts and ends on the same vertex.
\end{proof}
We can show something even stronger, namely that the cover is actually comprised of $l$ distinct instances of the original complex.
\begin{lemma}[Cover decomposition]\label{cover-decomposition}
Let $\stdcoboundary \in \coboundaryset{1}{\stdcomplex;G}$ and let $Y_\stdcoboundary$ be a cover of $\stdcomplex$ then there exists $Y_1,\dots,Y_l$ such that $Y_\stdcoboundary=\mathaccent\cdot{\bigcup}_{i=1}^l{Y_i}$ and $Y_1 \cong \cdots \cong Y_l \cong \stdcomplex$.
\end{lemma}
\begin{proof}
    $Y_\stdcoboundary$ is a cover of $\stdcomplex$.
    $\stdcoboundary$ is a coboundary therefore there exists $\psi \in \cochainset{0}{\stdcomplex;G}$ such that $\stdcoboundary = d_1\psi$.
    We will show that for any dimension $i$ and every face $\stdface \in Y_\stdcoboundary(i)$ is of the form \\
    $\stdface = \set{[\stdvertex_1, \psi(\stdvertex_1) s], \dots, [\stdvertex_{i}, \psi(\stdvertex_{i}) s]}$ such that $\set{\stdvertex_1,\dots,\stdvertex_i} \in \stdcomplex(i)$.
    Note that if $[\stdvertex_1, \psi(\stdvertex_1) s] \in \stdface$ then the rest of the vertices in the face are of the form:
    \begin{equation*}
        [\stdvertex_j, \stdcoboundary(\stdvertex_j, \stdvertex_1)\psi(\stdvertex_1) s]
        = [\stdvertex_j, \psi(\stdvertex_j)\psi(\stdvertex_1)^{-1}\psi(\stdvertex_1) s]
        = [\stdvertex_j, \psi(\stdvertex_j) s]
    \end{equation*}
    Consider the sub-complexes of $Y_\stdcoboundary$ defined by:
    \begin{equation*}
        Y_j(i) = \set{\set{[\stdvertex_1, \psi(\stdvertex_1) i], \dots, [\stdvertex_i, \psi(\stdvertex_i) i]} \suchthat \set{\stdvertex_1,\dots,\stdvertex_i} \in \stdcomplex(i)}
    \end{equation*}
    It is easy to see that these are all simplicial complexes.
    We will now show that they are distinct, that their union is $Y_\stdcoboundary$ and that they are indeed isomorphic to $\stdcomplex$.

    Assume that $Y_{j_1} \cap Y_{j_2} \ne \set{\emptyset}$.
    Let $[\stdvertex, k]$ be a vertex in $Y_{j_1} \cap Y_{j_2}$ (there must be one due to the closure property of both complexes).
    Therefore $\psi(\stdvertex) j_1 = k = \psi(\stdvertex) j_2$ and therefore $j_1=j_2$ (due to the fact that $\psi(\stdvertex) \in G$ and therefore it has an inverse).

    For every dimension $i$ let $\set{[\stdvertex_1, \psi(\stdvertex_1) s], \dots, [\stdvertex_i, \psi(\stdvertex_i) s]} \in Y_\stdcoboundary(i)$.
    Note that $s \in S$ and \\
    $\set{\stdvertex_1,\dots,\stdvertex_i} \in \stdcomplex(i)$ and therefore $\set{[\stdvertex_1, \psi(\stdvertex_1) s], \dots, [\stdvertex_i, \psi(\stdvertex_i) s]} \in Y_j(i)$.
    It is easy to see that $Y_j \subseteq Y_\stdcoboundary$ for every $j$.

    For every $Y_j$ consider $f:\stdcomplex\rightarrow Y_j$ to be $f(\set{\stdvertex_1,\dots,\stdvertex_i}) = \set{[\stdvertex_1,\psi(\stdface)j],\dots,[\stdvertex_i,\psi(\stdface)j]}$.
    Notice that this is a bijection (with the inverse function being $f^{-1}(\set{[\stdvertex_1,\psi(\stdface)j],\dots,[\stdvertex_i,\psi(\stdface)j]}) = \set{\stdvertex_1,\dots,\stdvertex_i}$).
    In addition notice that if $\stdface_1 \subseteq \stdface_2$ then $f(\stdface_1) \subseteq f(\stdface_2)$.
\end{proof}
\section{Local Assignments in the Original Complex Imply a Near Cover in the Representation Complex}\label{sec:local-assignments-in-the-original-complex-imply-a-near-cover-in-the-representation-complex}

In this section we will show how to derive a near cover for the representation complex from local assignments in the original complex.
More specifically, we will provide an algorithm that, given some edge, returns the edges that cover it in the near cover.
This cover will not only be helpful in separating the complex to somewhat-consistent components but will also allow us to bound how close each component is to being fully consistent.\\
Consider the following algorithm for querying the edges of the near cover:\\
\begin{algorithm}[H]\caption{query cover edge ($\set{\stdface_1, \stdface_2})$}\label{alg:query-cover-edge}
\SetAlgoLined
\DontPrintSemicolon
    Query the list functions on $\stdface_1$: $\lassignment{F}^{\stdface_1}_1,\dots, \lassignment{F}^{\stdface_1}_l$.\\
    Query the list functions on $\stdface_2$: $\lassignment{F}^{\stdface_2}_1,\dots, \lassignment{F}^{\stdface_2}_l$.\\
    \uIf{there exists $\pi$ such that for every $\stdvertex \in \stdface_1 \cap \stdface_2$ and every $i$ it holds that $\lassignment{F}^{\stdface_1}_i(\stdvertex) = \lassignment{F}^{\stdface_2}_{\pi(i)}(\stdvertex)$}{ \label{alg:query-edge:find-agreeing-edge}
    \Return $\set{\set{\lassignment{F}^{\stdface_1}_i, \lassignment{F}^{\stdface_2}_{\pi(i)}} \suchthat i \in \sparens{l}}$
    }\Else{
    \Return $\set{\set{\lassignment{F}^{\stdface_1}_i, \lassignment{F}^{\stdface_2}_{i}} \suchthat i \in \sparens{l}}$ \label{alg:query-cover-edge:arbitrary-choice}
    }
\end{algorithm}

Before we move on consider the following intuition to what Algorithm~\ref{alg:query-cover-edge} does:
Effectively Algorithm~\ref{alg:query-cover-edge} tries to find a matching between the lifts of $\stdface_1$ and $\stdface_2$ that could be extended into a coboundary.

We consider the near cover whose vertices are the assignments associated with each $k$-face, its edges are the result of Algorithm~\ref{alg:query-cover-edge} and we also complete the cover upwards (i.e.\ if a higher dimensional face can be covered it will be).\\
It will also be useful to have a way to distinguish between edges that satisfy the condition in step~\ref{alg:query-edge:find-agreeing-edge} and those that are not.
\begin{definition}
    An adequately covered edge is an edge for which there exists a permutation $\pi$ such that for every $\stdvertex \in \stdface_1 \cap \stdface_2$ and every $i$ it holds that $\lassignment{F}^{\stdface_1}_i(\stdvertex) = \lassignment{F}^{\stdface_2}_{\pi(i)}(\stdvertex)$.
\end{definition}
We can also check whether an edge is adequately covered:\\
\begin{algorithm}[H]\caption{is adequately covered ($\set{\stdface_1, \stdface_2})$}\label{alg:is-adequately-covered}
\SetAlgoLined
\DontPrintSemicolon
Query the list functions on $\stdface_1$: $\lassignment{F}^{\stdface_1}_1,\dots, \lassignment{F}^{\stdface_1}_l$.\\
Query the list functions on $\stdface_2$: $\lassignment{F}^{\stdface_2}_1,\dots, \lassignment{F}^{\stdface_2}_l$.\\
\Return{whether there exists $\pi$ such that for every $\stdvertex \in \stdface_1 \cap \stdface_2$ and every $i$ it holds that $\lassignment{F}^{\stdface_1}_i(\stdvertex) = \lassignment{F}^{\stdface_2}_{\pi(i)}(\stdvertex)$}
\end{algorithm}
Now that we have defined the near cover implied by the functions on the $k$-faces of the complex it is time to move on to present the test.
\section{Presenting a Test for List Agreement}\label{sec:presenting-a-test-for-list-agreement}

Now we are finally ready to present our test:\\
\begin{algorithm}[H]\caption{test for list agreement}\label{alg:test-for-list-agreement}
\SetAlgoLined
\DontPrintSemicolon
\Pick{with probability 0.5}{
    Run the test guaranteed by Lemma~\ref{co-boundaries-in-representation-complex-are-testable}. Whenever the test performs a query run Algorithm~\ref{alg:query-cover-edge} and provide the test with its answer. \label{alg:list-agreement:triangle-test}\\
    Sample an edge in the representation complex and check whether it is adequately covered (using Algorithm~\ref{alg:is-adequately-covered}). \label{alg:list-agreement:edges-are-adequatly-covered}
}
\end{algorithm}
\begin{note*}
    We have shown in Lemma~\ref{sampling-from-representation-complex} that one can indeed sample from the representation complex.
\end{note*}
We will spend the rest of this section proving that this algorithm is indeed a test for list agreement.
\begin{theorem}[Main Theorem]\label{main-theorem}
    Algorithm~\ref{alg:test-for-list-agreement} is a test for list agreement when the $l$-assignment is $2$-locally-differing over a complex that is a 1-up agreement expander and whose every link is a coboundary expander over the symmetric group with $l$ elements\footnote{This includes the entire complex as it is the link of $\emptyset$.}.
\end{theorem}
Before we prove the main theorem let us reiterate the proof's strategy: We are first going to use the connection between cochains and near covers in order to claim that the near cover implied by the functions in the lists of the $k$-faces is close to a genuine cover that correspond to a coboundary.
Consider the genuine cover that is close to that near cover - this cover can be thought of as $l$ independent copies of the representation complex.
Note that because our near cover was made using the functions in the lists of the different vertices each one of the copies in the cover can be thought of as assigning each vertex of the representation complex with a single local function.
Also note that the vertices of the representation complex corresponds to the $k$-faces of the original complex.
Using this fact one can think of each of the copies of the representation complex in the cover as an assignment of a singular function to each of the $k$-faces of the original complex.
We would then proceed to estimate the distance of each of these assignments of functions from being agreeing assignments.
In an ideal setting we would be able to run the 1-up agreement expander test described in~\cite{dinur2017high} on each of the copies and measure their rejection probability.
Alas we do not have access to query the genuine cover (and therefore the correct local assignments for each vertices).
What we can do, however, is bound from above the rejection probability of these tests without running them.
We do so by first noting that the random choice done in that test corresponds to picking an edge in the representation complex.
Then consider when the agreement test rejects an edge in the representation complex, this can happen in one of two cases - Either the edge is not adequately covered\footnote{Note that if the edge is not adequately covered in the near cover then it is not adequately covered in the cover.} or the edge was adequately covered in the near cover but the near cover and the cover differ on that edge.
Using this fact we can bound from above the probability that the 1-up agreement test rejects when given access to each one of the copies of the representation complex.
Therefore the distance of each copy from having a global function can be estimated.
This is also a bound on the distance of the lists from having $l$ global functions that agree with them.
\begin{notation*}
    Denote by $Y$ the near cover generated by Algorithm~\ref{alg:query-cover-edge}
    and by $\tilde{Y}$ a genuine cover that is close to $Y$ and represents a coboundary (Specifically the one guaranteed by Lemma~\ref{co-boundaries-in-representation-complex-are-testable}).
    Also denote by $f_Y:Y \rightarrow \hat{R}(\stdcomplex)$ the near covering map implied by Algorithm~\ref{alg:query-cover-edge} and by $f_{\tilde{Y}}: \tilde{Y} \rightarrow \hat{R}(\stdcomplex)$ the covering map between $\tilde{Y}$ and $\hat{R}(\stdcomplex)$.
    Lastly denote by $\tilde{Y}_1,\dots,\tilde{Y}_l$ the $l$ copies
    \footnote{Note that this is indeed the structure of a cover that corresponds to a coboundary due to Lemma~\ref{cover-decomposition}.}
    of $\hat{R}(\stdcomplex)$ that make up $\tilde{Y}$ and by $f_{\tilde{Y_i}}: \tilde{Y_i} \rightarrow \hat{R}(\stdcomplex)$ the restriction of the covering map to the copy $i$.
\end{notation*}
\begin{definition}
    Let $\lassignment{F}$ be a $k$-dimensional $l$-assignment of $\stdcomplex$ and let $\tilde{Y}$ be the cover of $\stdcomplex$ described above.
    Define $\assignment{F}_1,\dots,\assignment{F}_l$ the sub-assignments of $\lassignment{F}$ implied by $Y$ to be assignments such that:
    \begin{equation*}
        \forall i \in [l]: \tilde{Y}_i(0) = \assignment{F}_i
    \end{equation*}
\end{definition}
\begin{notation*}
    Denote the set of inadequately covered edges of $\hat{R}(\stdcomplex)$ as $I_{\repcplx{\stdcomplex}}$.
\end{notation*}
\begin{notation*}
    Denote by $A$ (stands for adjustment) the set of edges that are covered differently between $Y$ and $\tilde{Y}$, formally:
    \begin{equation*}
        A = \set{\stdface \in \repcplx{\stdcomplex}(1) \suchthat f_Y^{-1}(\stdface) \ne f_{\tilde{Y}}^{-1}(\stdface)}
    \end{equation*}

\end{notation*}
\begin{definition}\label{def:disagreeing-edges}
    Define $D_i$ the disagreeing edges of $\tilde{Y}_i$, formally:
    \begin{equation*}
        D_i = \set{\set{r_1, r_2} \in \repcplx{\stdcomplex}(1) \suchthat \assignment{F}_i^{r_1}|_{\core{\set{r_1, r_2}}} \ne \assignment{F}_i^{r_2}|_{\core{\set{r_1, r_2}}}}
    \end{equation*}
\end{definition}
Consider what happens when we run the 1-up agreement test on $\stdcomplex$ with the $k$-dimensional assignments of $\assignment{F}_i$.
The 1-up agreement test picks a $\parens{k+1}$-dimensional face and then two of its $k$-subsets and check whether the local assignment of the sub-faces from $\assignment{F}_i$ agree on their intersection.
Note that this is completely equivalent to picking an edge in the representation complex (since, as shown in Algorithm~\ref{alg:sample-from-the-representation-complex} this is exactly the process of sampling an edge in $\hat{R}(\stdcomplex)$) and check whether the functions associated with its vertices agree on the core of the edge.
Therefore, when viewing the 1-up agreement test through the lens of the representation complex the set of edges that the 1-up agreement test rejects is exactly $D_i$.
The following lemma is a formal proof of this notion:
\begin{lemma}\label{lm:characterising-defective-edges}
    For every $\tilde{Y}_i$ if the 1-up agreement test rejects the edge $\stdface \in \hat{R}(\stdcomplex)$ then either $\stdface$ is inadequately covered or $\stdface$ is covered differently by $\tilde{Y}$ then it is covered by $Y$, formally:
    \begin{equation*}
        D_i \subseteq A \cup I_{\repcplx{\stdcomplex}}
    \end{equation*}
\end{lemma}
\begin{proof}
    Let $\set{\stdface_1, \stdface_2} \in D_i$.
    Assuming that $\set{\stdface_1, \stdface_2} \notin A \cup I_{\repcplx{\stdcomplex}}$ then, because $\set{\stdface_1, \stdface_2} \notin A$ it holds that $f_Y^{-1}(\set{\stdface_1, \stdface_2}) = f_{\tilde{Y}}^{-1}(\set{\stdface_1, \stdface_2})$.
    In addition, because $\set{\stdface_1, \stdface_2} \notin I_{\repcplx{\stdcomplex}}$ then $\stdface$ is adequately covered i.e.
    \begin{equation*}
        \forall \set{\lassignment{F}^{\stdface_1}_i, \lassignment{F}^{\stdface_2}_j} \in f_Y^{-1}(\stdface): \lassignment{F}^{\stdface_1}_i|_{\core{\set{\stdface_1, \stdface_2}}} = \lassignment{F}^{\stdface_2}_j|_{\core{\set{\stdface_1, \stdface_2}}}
    \end{equation*}
    Therefore:
    \begin{equation*}
        \forall \set{\lassignment{F}^{\stdface_1}_i, \lassignment{F}^{\stdface_2}_j} \in f_{\tilde{Y}}^{-1}(\stdface): \lassignment{F}^{\stdface_1}_i|_{\core{\set{\stdface_1, \stdface_2}}} = \lassignment{F}^{\stdface_2}_j|_{\core{\set{\stdface_1, \stdface_2}}}
    \end{equation*}
    Which contradicts the fact that $\stdface \in D_i$.
\end{proof}
\begin{lemma}\label{distance-bound}
    Let $\stdcomplex$ be a $\gamma$-coboundary-expander and $\agreementexpansionconst$-agreement-expander.
    Also let $\cobdrtestconst = \cobdrtestconst(k,\gamma)$ be the constant from Lemma~\ref{co-boundaries-in-representation-complex-are-testable} then for every $l$-assignment $\lassignment{F}$ it holds that:
    \begin{equation*}
        \frac{\cobdrtestconst}{2\cobdrtestconst+2}\frac{\agreementexpansionconst}{k}\cdot \dist\parens{\lassignment{F}, \lassignmentset{A}} \le \pr{\text{Algorithm~\ref{alg:test-for-list-agreement} rejects}}
    \end{equation*}
\end{lemma}
\begin{proof}
    \begin{align*}
        \dist\parens{\lassignment{F}, \lassignmentset{A}}
        & \le \frac{1}{l}\sum_{i=1}^l{\dist\parens{\assignment{F}_i, \assignmentset{A}}}
        \le \frac{1}{l}\sum_{i=1}^l{\frac{k}{\agreementexpansionconst}\norm{D_i}}
        \le \frac{1}{l}\sum_{i=1}^l{\frac{k}{\agreementexpansionconst}\norm{A \cup I_{\repcplx{\stdcomplex}}}}
        \le \frac{k}{\agreementexpansionconst}\norm{A \cup I_{\repcplx{\stdcomplex}}}
    \end{align*}
    Note that step~\ref{alg:list-agreement:edges-are-adequatly-covered} of Algorithm~\ref{alg:test-for-list-agreement} picks an edge in $\repcplx{\stdcomplex}$ and checks whether it is in $A$ therefore:
    \begin{equation*}
        \norm{I_{\repcplx{\stdcomplex}}} = \pr{\text{Step~\ref{alg:list-agreement:edges-are-adequatly-covered} of Algorithm~\ref{alg:test-for-list-agreement} rejects}}
    \end{equation*}
    In addition due to Lemma~\ref{co-boundaries-in-representation-complex-are-testable} it holds that:
    \begin{equation*}
        \cobdrtestconst\norm{A} \le \pr{\text{Step~\ref{alg:list-agreement:triangle-test} of Algorithm~\ref{alg:test-for-list-agreement} rejects}}
    \end{equation*}
    If algorithm~\ref{alg:test-for-list-agreement} rejects with probability $\epsilon$ then both $\pr{\text{Step~\ref{alg:list-agreement:triangle-test} of algorithm~\ref{alg:test-for-list-agreement} rejects}} \le 2\epsilon$ and $\pr{\text{Step~\ref{alg:list-agreement:edges-are-adequatly-covered} of algorithm~\ref{alg:test-for-list-agreement} rejects}} \le 2\epsilon$ and therefore:
    \begin{align*}
        \norm{A \cup I_{\repcplx{\stdcomplex}}}
        & \le \norm{A} + \norm{I_{\repcplx{\stdcomplex}}}\\
        &\le \frac{1}{\cobdrtestconst} \pr{\text{Step~\ref{alg:list-agreement:triangle-test} of algorithm~\ref{alg:test-for-list-agreement} rejects}} + \pr{\text{Step~\ref{alg:list-agreement:edges-are-adequatly-covered} of algorithm~\ref{alg:test-for-list-agreement} rejects}}\\
        &\le 2 \parens{\frac{1}{\cobdrtestconst}+1}\epsilon
    \end{align*}
    Hence:
    \begin{equation*}
        \dist\parens{\lassignment{F}, \lassignmentset{A}} \le 2 \frac{k}{\agreementexpansionconst} \parens{\frac{1}{\cobdrtestconst}+1}\epsilon = \parens{\frac{2\cobdrtestconst+2}{\cobdrtestconst\agreementexpansionconst}}\epsilon
    \end{equation*}
    Therefore:
    \begin{equation*}
        \frac{\cobdrtestconst}{2\cobdrtestconst+2}\frac{\agreementexpansionconst}{k} \dist\parens{\lassignment{F}, \lassignmentset{A}} \le \epsilon
    \end{equation*}
\end{proof}

We will now move on to prove that Algorithm~\ref{alg:test-for-list-agreement} accepts with probability $1$ any agreeing $2$-locally-differing $l$-assignment.
Before we do that, however, let us start by examining the results of Algorithm~\ref{alg:query-cover-edge}:
\begin{lemma}[$\pi$ found in Algorithm~\ref{alg:query-cover-edge} is singular]\label{singular-pi}
Let $\lassignment{f}$ be a $2$-locally-differing $l$-assignment.
If Algorithm~\ref{alg:query-cover-edge} finds a permutation $\pi$ in step~\ref{alg:query-edge:find-agreeing-edge} when given access to $\lassignment{f}$ then it is singular.
\end{lemma}
\begin{proof}
    We will use the counter positive argument - assume that there exists two permutations $\pi_1$ and $\pi_2$ ($\pi_1\ne\pi_2$) such that for every $\stdvertex \in \stdface_1 \cap \stdface_2$ and every $i$ it holds that $\lassignment{F}^{\stdface_1}_i(\stdvertex) = \lassignment{F}^{\stdface_2}_{\pi_1(i)}(\stdvertex)$ and $\lassignment{F}^{\stdface_1}_i(\stdvertex) = \lassignment{F}^{\stdface_2}_{\pi_2(i)}(\stdvertex)$.
    Let $i$ be an integer such that $\pi_1(i) \ne \pi_2(i)$ and note that for every vertex $\stdvertex \in \stdface_1 \cap \stdface_2$ it holds that:
    \[
        \lassignment{F}^{\stdface_1}_i(\stdvertex) = \lassignment{F}^{\stdface_2}_{\pi_1(i)}(\stdvertex) = \lassignment{F}^{\stdface_1}_{\pi_2^{-1}(\pi_1(i))}(\stdvertex)
    \]
    Note that because $\abs{\stdface_1 \cap \stdface_2} = \abs{\stdface_1}-1$ it holds that there exists at most one vertex $\stdvertex'$ such that $\lassignment{F}^{\stdface_1}_i(\stdvertex') \ne \lassignment{F}^{\stdface_1}_{\pi_2^{-1}(\pi_1(i))}(\stdvertex')$ which contradicts the fact that $\lassignment{f}$ is locally differing.
\end{proof}
We are now ready to prove the lemma:
\begin{lemma}\label{acceptance-of-agreeing-l-assignments}
    If $\lassignment{F}$ is an agreeing and $2$-locally-differing $k$-dimensional $l$-assignment then it passes the test with probability $1$.
\end{lemma}
\begin{proof}
    In order to show that the $l$-assignment passes the test with probability $1$ we will first show that every edge is adequately covered and find the nature of the permutation $\pi$ found by Algorithm~\ref{alg:query-cover-edge}.
    Then we will use this in order to show that the test always accepts.

    $\lassignment{F}$ is an agreeing assignment therefore there exists $\assignment{F}_1,\cdots,\assignment{f}_l$ such that for every face $\stdface \in \stdcomplex(k)$ there exists a permutation $\pi_\stdface$ such that for every $i$ it holds that:
    \[
    \assignment{f}_i|_{\stdface} = \lassignment{f}^\stdface_{\pi_\stdface(i)}
    \]
    Therefore at step~\ref{alg:query-edge:find-agreeing-edge} Algorithm~\ref{alg:query-cover-edge} finds $\pi = \pi_{\stdface_2}\pi_{\stdface_1}^{-1}$ due to the fact that $\pi$ is singular and for every $i$ it holds that:
    \[
    \forall \stdvertex \in \stdface_1 \cap \stdface_2: \lassignment{F}^{\stdface_1}_i(\stdvertex) = \assignment{f}_{\pi_{\stdface_1}^{-1}(i)}(\stdvertex) = \lassignment{F}^{\stdface_2}_{\pi_{\stdface_2}(\pi_{\stdface_1}^{-1}(i))}(\stdvertex)
    \]
    This immediately implies that every edge is adequately covered and that the cochain created by Algorithm~\ref{alg:query-cover-edge} is a coboundary.
    Therefore, if Algorithm~\ref{alg:test-for-list-agreement} checks if an edge is adequately covered then it accepts since every edge is adequately covered.
    Otherwise Algorithm~\ref{alg:test-for-list-agreement} checks whether the cochain generated by Algorithm~\ref{alg:query-cover-edge} is a coboundary and therefore it accepts with probability $1$.
\end{proof}
We are now finally ready to prove the main theorem:
\begin{proof}[Proof of Theorem~\ref{main-theorem}]
    In order to calculate the cover over each edge in the representation complex all the local assignments of both vertices are queried.
    The empty triangle test queries a triangle (either empty or proper) and therefore it queries $3$ edges using Algorithm~\ref{alg:query-cover-edge}.
    Each run of Algorithm~\ref{alg:query-cover-edge} requires $2l$ queries however, some vertices are queried twice, therefore we can reduce the query complexity of this step to $3l$.
    Checking whether an edge is adequately covered takes $2l$ queries (using the same argument).
    Therefore algorithm~\ref{alg:test-for-list-agreement} queries $\lassignment{F}$ at most $3l$ times.
    We finish the proof by noting that the fact that algorithm~\ref{alg:test-for-list-agreement} is indeed a test for whether the $l$-assignment is agreeing stems directly from Lemma~\ref{distance-bound} and Lemma~\ref{acceptance-of-agreeing-l-assignments}.
\end{proof}
We therefore conclude that Algorithm~\ref{alg:test-for-list-agreement} is a test for list agreement expansion in the $2$-differing case.
We note that the test's distance improves as the complex is a better agreement expander and coboundary expander as well as when the dimension of the local assignments decrease.

Note that, in fact, we could have proven a more general statement then the main theorem which we will now present.
Before we do that, however, we have to define the test graph of an agreement tester:
\begin{definition}[Test Graph]
    Let $\stdcomplex$ be a simplicial complex and let $T$ be an agreement test on the $k$th dimensional faces of $\stdcomplex$.
    Given a set of local functions $\set{f_{\stdface}:\stdface \rightarrow \set{0,1}}_{\stdface \in \stdcomplex(d)}$ the test picks two faces $\set{\stdface_1, \stdface_2} \sim \distribution{D}$ and checks whether $f_{\stdface_1}|_{\stdface_1 \cap \stdface_2}= f_{\stdface_2}|_{\stdface_1 \cap \stdface_2}$.
    Define the test graph of the test $T$ to be the weighted graph $G=(V,E,w_T)$ such that:
    \begin{itemize}
        \item $V=\stdcomplex(k)$
        \item $E=\set{\set{\stdface_1, \stdface_2)} \suchthat \prex{\set{\genface_1, \genface_2} \sim \distribution{D}}{{\genface_1 = \stdface_1 \text{ and } \genface_2 = \stdface_2}} > 0}$
        \item $w_T(\set{\stdface_1, \stdface_2)}) = \prex{\set{\genface_1, \genface_2} \sim \distribution{D}}{{\genface_1 = \stdface_1 \text{ and } \genface_2 = \stdface_2}}$
    \end{itemize}
\end{definition}
\begin{theorem}[Main Theorem, Generalized]
    Let $\stdcomplex$ be a simplicial complex and let $T$ be an agreement test on the $k$-th dimensional faces of $\stdcomplex$.
    Then if the test graph of $T$ is a $1$-skeleton of a coboundary expander with respect to $S_l$ (denoted $\stdcomplex'$) such that $w_T(\set{\stdface_1, \stdface_2}) = \weightcplx{\stdcomplex'}{\set{\stdface_1, \stdface_2}}$ and there exists an algorithm that queries the test graph then Algorithm~\ref{alg:test-for-list-agreement} is a test for $l$-agreement over $\stdcomplex(k)$.
\end{theorem}
\section{On the $2$-Differing Assumption}\label{sec:on-the-2-differing-assumption}
In this section we discuss the $2$-differing assumption that we made on assignments.
Specifically we will show some lower bounds on the number of queries required in order to test list agreement in complexes of interest.
Our main tool for proving this lower bound is the following Lemma:
\begin{notation}
    Let $\lassignment{f}$ and let $q$ be a possible query for $\lassignment{f}$.
    Denote by $\lassignment{f}[q]$ the result of performing the query $q$ on $\lassignment{f}$.
\end{notation}
\begin{lemma}\label{tests-fail}
Let $\lassignment{R}$ be a non-agreeing $l$-assignment.
Suppose that for any set of $p$ queries $q_1 , \cdots , q_{p}$ there exists an agreeing assignment $\lassignment{a}_{\parens{q_1 , \cdots , q_{p}}}$ such that:
\[
    \forall i \in \set{1,\cdots,p}: \lassignment{r}[q_i] = \lassignment{a}_{\parens{q_1 , \cdots , q_{o}}}[q_i]
\]
then there is no test that tests $l$-agreement and queries the assignment only $p$ times.
\end{lemma}
\begin{proof}
    We will show that the test always accepts when given oracle access to $\lassignment{r}$ and therefore $\lassignment{r}$ should be an agreeing assignment despite the fact that it is not.
    Assume that there is a test $T$ that queries only $p$ location of the $l$-assignment and accepts with probability $1$ iff the $l$-assignment is agreeing.
    Denote by $T^\lassignment{f}(r)$ running the test $T$ with oracle access to $\lassignment{f}$ and random bits $r$.
    Also, for every set of queries $T$ performs on $\lassignment{f}$ define by $\lassignment{a}_{\parens{q_{1},\dots,q_{p}}}$ the agreeing assignment promised by the lemma.
    Note that the even if $T$ is an adaptive tester its $i$-th query only depends on its random bits and the first $(i-1)$ queries.
    Therefore for every $r$ let $\parens{q_{1}^r,\cdots,q_{p}^r}$ be the queries performed in $T^\lassignment{r}(r)$ and note that:
    \[
        \forall r: T^{\lassignment{r}}(r) = T^{\lassignment{a}_{\parens{q_{1}^r,\dots,q_{p}^r}}}(r)=1
    \]
    Which contradicts the fact that $\lassignment{r}$ is not an agreeing $l$-assignment
\end{proof}
We then use Lemma~\ref{tests-fail} to test $2$-colorings or lack of $2$-colourings in even/odd cycles in the complex.
It is therefore convenient to define the following:
\begin{definition}[Coloring candidate]
    A $2$-assignment is a coloring candidate for a simple cycle (or coloring candidate in short) if the following conditions hold:
    \begin{itemize}
        \item Any edge in the cycle is assigned either $\sparens{\stdvertex=0, \genvertex=1}, \sparens{\stdvertex=1, \genvertex=0}$ or $\sparens{\stdvertex=1, \genvertex=1}, \sparens{\stdvertex=0, \genvertex=0}$.
        \item Any edge with one side in the cycle is assigned $\sparens{\stdvertex=1, \genvertex=0}, \sparens{\stdvertex=0, \genvertex=0}$.
        \item Any edge whose both sides are not in the cycle is assigned $\sparens{\stdvertex=0, \genvertex=0}, \sparens{\stdvertex=0, \genvertex=0}$.
    \end{itemize}
\end{definition}
Consider the following coloring candidates:
\begin{definition}
    Let $\stdcomplex$ be a simplicial complex and let $\stdcycle=\parens{\stdcycle_0,\cdots,\stdcycle_n}$ be a simple cycle in $\stdcomplex$ (when considering $\stdcomplex$ as not oriented).
    Given two vertices in the cycle $\stdvertex_1, \genvertex_2$ denote the distance between them on the cycle as $\dist_\stdcycle\parens{\stdvertex_1, \genvertex_2}$ and the distance on the cycle when not counting the edge $\set{\stdcycle_k, \stdcycle_{k+1}}$ as $\dist_{\stdcycle,k}\parens{\stdvertex_1, \genvertex_2}$.
    Define the following two coloring candidates:
    \begin{multline*}
        \lassignment{F}_{\stdcycle}^e\parens{\set{\stdvertex, \genvertex}} = \begin{cases}
                                                                                 \sparens{\stdvertex=0, \genvertex=1}, \sparens{\stdvertex=1, \genvertex=0} & \exists i, j: \stdvertex = \stdcycle_i, \genvertex=\stdcycle_{j}, \dist_\stdcycle\parens{\stdvertex_1, \genvertex_2} \text{ is odd}\\
                                                                                 \sparens{\stdvertex=0, \genvertex=0}, \sparens{\stdvertex=1, \genvertex=1} & \exists i, j: \stdvertex = \stdcycle_i, \genvertex=\stdcycle_{j}, \dist_\stdcycle\parens{\stdvertex_1, \genvertex_2} \text{ is even}\\
                                                                                 \sparens{\stdvertex=1, \genvertex=0}, \sparens{\stdvertex=0, \genvertex=0} & \exists i \forall j: \stdvertex = \stdcycle_i, \genvertex \ne \stdcycle_{j}\\
                                                                                 \sparens{\stdvertex=0, \genvertex=0}, \sparens{\stdvertex=0, \genvertex=0} & \forall i: \stdvertex = \stdcycle_i, \genvertex \ne \stdcycle_{i}\\
        \end{cases}\\
        \lassignment{F}_{\stdcycle, k}^o\parens{\set{\stdvertex, \genvertex}} = \begin{cases}
                                                                                 \sparens{\stdvertex=1, \genvertex=1}, \sparens{\stdvertex=0, \genvertex=0} & \stdvertex = \stdcycle_k, \genvertex=\stdcycle_{k+1}\\
                                                                                 \sparens{\stdvertex=0, \genvertex=1}, \sparens{\stdvertex=1, \genvertex=0} & \exists i, j: \stdvertex = \stdcycle_i, \genvertex=\stdcycle_{j}, \dist_{\stdcycle,k}\parens{\stdvertex_1, \genvertex_2} \text{ is odd}, (i,j) \ne (k,k+1)\\
                                                                                 \sparens{\stdvertex=0, \genvertex=0}, \sparens{\stdvertex=1, \genvertex=1} & \exists i, j: \stdvertex = \stdcycle_i, \genvertex=\stdcycle_{j}, \dist_{\stdcycle,k}\parens{\stdvertex_1, \genvertex_2} \text{ is even}\\
                                                                                 \sparens{\stdvertex=1, \genvertex=0}, \sparens{\stdvertex=0, \genvertex=0} & \exists i \forall j: \stdvertex = \stdcycle_i, \genvertex \ne \stdcycle_{j}\\
                                                                                 \sparens{\stdvertex=0, \genvertex=0}, \sparens{\stdvertex=0, \genvertex=0} & \forall i: \stdvertex = \stdcycle_i, \genvertex \ne \stdcycle_{i}
        \end{cases}
    \end{multline*}
\end{definition}
We think of $\lassignment{F}_{\stdcycle}^e$ as a $2$-assignment that attempts to color the vertices of $\stdcycle$ in two colors.
Likewise we think of $\lassignment{F}_{\stdcycle}^o$ as a $2$-assignment that ``glues together'' $\stdcycle_0$ and $\stdcycle_1$ and then tries to color the rest of the graph in two colors.
\begin{example}
    Consider the following graphs as parts of the $1$-skeleton of some simplicial complex.
    Our first example is of $\lassignment{F}_{\stdcycle}^e$:
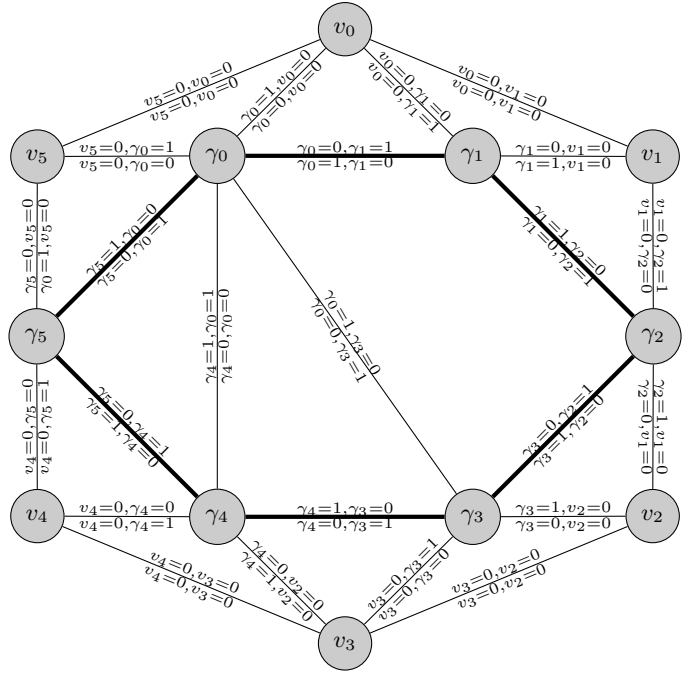
    \begin{center}
        \begin{tikzpicture}[%
        every node/.style={draw,fill=gray!40,circle,minimum size=18pt,font=\footnotesize}, node distance=1cm]
            \node (g0) {$\stdcycle_0$};
            \node[draw=none, fill=none, right=of g0] (t0) {};
            \node (g1) [right=of t0] {$\stdcycle_1$};
            \node[draw=none, fill=none, below right=of g1] (t1) {};
            \node (g2) [below right=of t1] {$\stdcycle_2$};
            \node[draw=none, fill=none, below left=of g2] (t2) {};
            \node (g3) [below left=of t2] {$\stdcycle_3$};
            \node[draw=none, fill=none, left=of g3] (t3) {};
            \node (g4) [left=of t3] {$\stdcycle_4$};
            \node[draw=none, fill=none, above left=of g4] (t4) {};
            \node (g5) [above left=of t4] {$\stdcycle_5$};
            \node[draw=none, fill=none, above right=of g5] (t5) {};
            \node (v5) [above left=of t5] {$\stdvertex_5$};
            \node (v0) [above=of t0] {$\stdvertex_0$};
            \node (v1) [above right=of t1] {$\stdvertex_1$};
            \node (v2) [below right=of t2] {$\stdvertex_2$};
            \node (v3) [below=of t3] {$\stdvertex_3$};
            \node (v4) [below left=of t4] {$\stdvertex_4$};
            \draw[ultra thick] (g0) -- node[draw=none, fill=none]{$\substack{\stdcycle_0=0,\stdcycle_1=1\\ \stdcycle_0=1,\stdcycle_1=0}$} (g1) -- node[draw=none, fill=none, rotate=315]{$\substack{\stdcycle_1=1,\stdcycle_2=0\\\stdcycle_1=0,\stdcycle_2=1}$} (g2) -- node[draw=none, fill=none, rotate=45]{$\substack{\stdcycle_3=0,\stdcycle_2=1\\\stdcycle_3=1,\stdcycle_2=0}$} (g3) -- node[draw=none, fill=none]{$\substack{\stdcycle_4=1,\stdcycle_3=0\\\stdcycle_4=0,\stdcycle_3=1}$} (g4) -- node[draw=none, fill=none, rotate=315]{$\substack{\stdcycle_5=0,\stdcycle_4=1\\\stdcycle_5=1,\stdcycle_4=0}$} (g5) -- node[draw=none, fill=none, rotate=45]{$\substack{\stdcycle_5=1,\stdcycle_0=0\\\stdcycle_5=0,\stdcycle_0=1}$} (g0);
            \draw (g0) -- node[draw=none, fill=none, rotate=45]{$\substack{\stdcycle_0=1,\stdvertex_0=0\\\stdcycle_0=0,\stdvertex_0=0}$} (v0) -- node[draw=none, fill=none, rotate=315]{$\substack{\stdvertex_0=0,\stdcycle_1=0\\\stdvertex_0=0,\stdcycle_1=1}$} (g1);
            \draw (g1) -- node[draw=none, fill=none]{$\substack{\stdcycle_1=0,\stdvertex_1=0\\ \stdcycle_1=1,\stdvertex_1=0}$} (v1) -- node[draw=none, fill=none, rotate=270]{$\substack{\stdvertex_1=0,\stdcycle_2=1\\\stdvertex_1=0,\stdcycle_2=0}$} (g2);
            \draw (g2) -- node[draw=none, fill=none, rotate=270]{$\substack{\stdcycle_2=1,\stdvertex_1=0\\ \stdcycle_2=0,\stdvertex_1=0}$} (v2) -- node[draw=none, fill=none]{$\substack{\stdcycle_3=1,\stdvertex_2=0\\\stdcycle_3=0,\stdvertex_2=0}$} (g3);
            \draw (g3) -- node[draw=none, fill=none, rotate=45]{$\substack{\stdvertex_3=0,\stdcycle_3=1\\ \stdvertex_3=0, \stdcycle_3=0}$} (v3) -- node[draw=none, fill=none, rotate=315]{$\substack{\stdcycle_4=0,\stdvertex_2=0\\\stdcycle_4=1,\stdvertex_2=0}$} (g4);
            \draw (g4) -- node[draw=none, fill=none]{$\substack{\stdvertex_4=0,\stdcycle_4=0\\ \stdvertex_4=0, \stdcycle_4=1}$} (v4) -- node[draw=none, fill=none, rotate=90]{$\substack{\stdvertex_4=0,\stdcycle_5=0\\\stdvertex_4=0, \stdcycle_5=1}$} (g5);
            \draw (g5) -- node[draw=none, fill=none, rotate=90]{$\substack{\stdcycle_5=0,\stdvertex_5=0\\ \stdcycle_0=1, \stdvertex_5=0}$} (v5) -- node[draw=none, fill=none]{$\substack{\stdvertex_5=0,\stdcycle_0=1\\\stdvertex_5=0, \stdcycle_0=0}$} (g0);
            \draw (v0) -- node[draw=none, fill=none, rotate=337.5]{$\substack{\stdvertex_0=0,\stdvertex_1=0\\ \stdvertex_0=0, \stdvertex_1=0}$} (v1);
            \draw (v2) -- node[draw=none, fill=none, rotate=22.5]{$\substack{\stdvertex_3=0,\stdvertex_2=0\\ \stdvertex_3=0, \stdvertex_2=0}$} (v3);
            \draw (v4) -- node[draw=none, fill=none, rotate=337.5]{$\substack{\stdvertex_4=0,\stdvertex_3=0\\ \stdvertex_4=0, \stdvertex_3=0}$} (v3);
            \draw (v5) -- node[draw=none, fill=none, rotate=22.5]{$\substack{\stdvertex_5=0,\stdvertex_0=0\\ \stdvertex_5=0, \stdvertex_0=0}$} (v0);
            \draw (g0) -- node[draw=none, fill=none, rotate=90]{$\substack{\stdcycle_4=1,\stdcycle_0=1\\ \stdcycle_4=0, \stdcycle_0=0}$} (g4);
            \draw (g0) -- node[draw=none, fill=none, rotate=300]{$\substack{\stdcycle_0=1,\stdcycle_3=0\\ \stdcycle_0=0, \stdcycle_3=1}$} (g3);
        \end{tikzpicture}
    \end{center}
    Note that the assignment is $2$-agreeing.
    We move on to demonstrate the odd case, consider the same graph with the $2$-assignment $\lassignment{F}_{\stdcycle, 0}^o$:
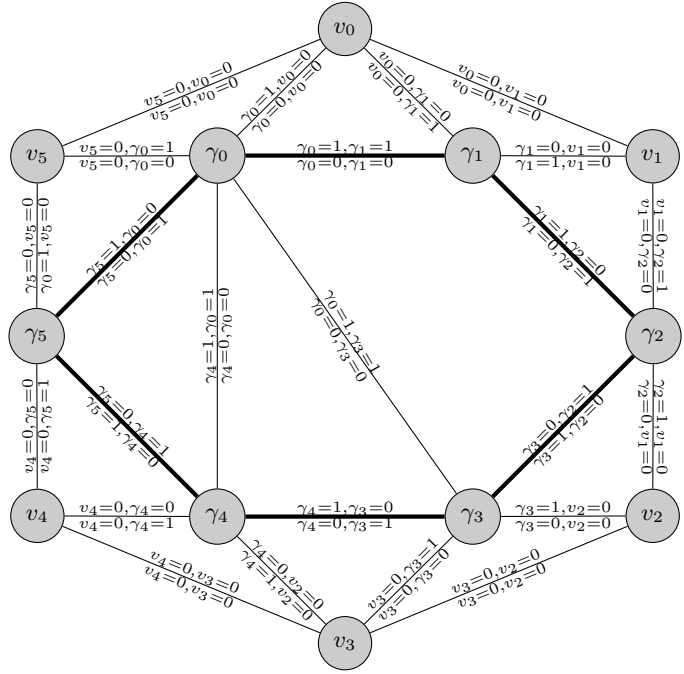
    \begin{center}
        \begin{tikzpicture}[%
        every node/.style={draw,fill=gray!40,circle,minimum size=18pt,font=\footnotesize}, node distance=1cm]
            \node (g0) {$\stdcycle_0$};
            \node[draw=none, fill=none, right=of g0] (t0) {};
            \node (g1) [right=of t0] {$\stdcycle_1$};
            \node[draw=none, fill=none, below right=of g1] (t1) {};
            \node (g2) [below right=of t1] {$\stdcycle_2$};
            \node[draw=none, fill=none, below left=of g2] (t2) {};
            \node (g3) [below left=of t2] {$\stdcycle_3$};
            \node[draw=none, fill=none, left=of g3] (t3) {};
            \node (g4) [left=of t3] {$\stdcycle_4$};
            \node[draw=none, fill=none, above left=of g4] (t4) {};
            \node (g5) [above left=of t4] {$\stdcycle_5$};
            \node[draw=none, fill=none, above right=of g5] (t5) {};
            \node (v5) [above left=of t5] {$\stdvertex_5$};
            \node (v0) [above=of t0] {$\stdvertex_0$};
            \node (v1) [above right=of t1] {$\stdvertex_1$};
            \node (v2) [below right=of t2] {$\stdvertex_2$};
            \node (v3) [below=of t3] {$\stdvertex_3$};
            \node (v4) [below left=of t4] {$\stdvertex_4$};
            \draw[ultra thick] (g0) -- node[draw=none, fill=none]{$\substack{\stdcycle_0=1,\stdcycle_1=1\\ \stdcycle_0=0,\stdcycle_1=0}$} (g1) -- node[draw=none, fill=none, rotate=315]{$\substack{\stdcycle_1=1,\stdcycle_2=0\\\stdcycle_1=0,\stdcycle_2=1}$} (g2) -- node[draw=none, fill=none, rotate=45]{$\substack{\stdcycle_3=0,\stdcycle_2=1\\\stdcycle_3=1,\stdcycle_2=0}$} (g3) -- node[draw=none, fill=none]{$\substack{\stdcycle_4=1,\stdcycle_3=0\\\stdcycle_4=0,\stdcycle_3=1}$} (g4) -- node[draw=none, fill=none, rotate=315]{$\substack{\stdcycle_5=0,\stdcycle_4=1\\\stdcycle_5=1,\stdcycle_4=0}$} (g5) -- node[draw=none, fill=none, rotate=45]{$\substack{\stdcycle_5=1,\stdcycle_0=0\\\stdcycle_5=0,\stdcycle_0=1}$} (g0);
            \draw (g0) -- node[draw=none, fill=none, rotate=45]{$\substack{\stdcycle_0=1,\stdvertex_0=0\\\stdcycle_0=0,\stdvertex_0=0}$} (v0) -- node[draw=none, fill=none, rotate=315]{$\substack{\stdvertex_0=0,\stdcycle_1=0\\\stdvertex_0=0,\stdcycle_1=1}$} (g1);
            \draw (g1) -- node[draw=none, fill=none]{$\substack{\stdcycle_1=0,\stdvertex_1=0\\ \stdcycle_1=1,\stdvertex_1=0}$} (v1) -- node[draw=none, fill=none, rotate=270]{$\substack{\stdvertex_1=0,\stdcycle_2=1\\\stdvertex_1=0,\stdcycle_2=0}$} (g2);
            \draw (g2) -- node[draw=none, fill=none, rotate=270]{$\substack{\stdcycle_2=1,\stdvertex_1=0\\ \stdcycle_2=0,\stdvertex_1=0}$} (v2) -- node[draw=none, fill=none]{$\substack{\stdcycle_3=1,\stdvertex_2=0\\\stdcycle_3=0,\stdvertex_2=0}$} (g3);
            \draw (g3) -- node[draw=none, fill=none, rotate=45]{$\substack{\stdvertex_3=0,\stdcycle_3=1\\ \stdvertex_3=0, \stdcycle_3=0}$} (v3) -- node[draw=none, fill=none, rotate=315]{$\substack{\stdcycle_4=0,\stdvertex_2=0\\\stdcycle_4=1,\stdvertex_2=0}$} (g4);
            \draw (g4) -- node[draw=none, fill=none]{$\substack{\stdvertex_4=0,\stdcycle_4=0\\ \stdvertex_4=0, \stdcycle_4=1}$} (v4) -- node[draw=none, fill=none, rotate=90]{$\substack{\stdvertex_4=0,\stdcycle_5=0\\\stdvertex_4=0, \stdcycle_5=1}$} (g5);
            \draw (g5) -- node[draw=none, fill=none, rotate=90]{$\substack{\stdcycle_5=0,\stdvertex_5=0\\ \stdcycle_0=1, \stdvertex_5=0}$} (v5) -- node[draw=none, fill=none]{$\substack{\stdvertex_5=0,\stdcycle_0=1\\\stdvertex_5=0, \stdcycle_0=0}$} (g0);
            \draw (v0) -- node[draw=none, fill=none, rotate=337.5]{$\substack{\stdvertex_0=0,\stdvertex_1=0\\ \stdvertex_0=0, \stdvertex_1=0}$} (v1);
            \draw (v2) -- node[draw=none, fill=none, rotate=22.5]{$\substack{\stdvertex_3=0,\stdvertex_2=0\\ \stdvertex_3=0, \stdvertex_2=0}$} (v3);
            \draw (v4) -- node[draw=none, fill=none, rotate=337.5]{$\substack{\stdvertex_4=0,\stdvertex_3=0\\ \stdvertex_4=0, \stdvertex_3=0}$} (v3);
            \draw (v5) -- node[draw=none, fill=none, rotate=22.5]{$\substack{\stdvertex_5=0,\stdvertex_0=0\\ \stdvertex_5=0, \stdvertex_0=0}$} (v0);
            \draw (g0) -- node[draw=none, fill=none, rotate=90]{$\substack{\stdcycle_4=1,\stdcycle_0=1\\ \stdcycle_4=0, \stdcycle_0=0}$} (g4);
            \draw (g0) -- node[draw=none, fill=none, rotate=300]{$\substack{\stdcycle_0=1,\stdcycle_3=1\\ \stdcycle_0=0, \stdcycle_3=0}$} (g3);
        \end{tikzpicture}
    \end{center}
    Note that $\lassignment{F}_{\stdcycle, 0}^o$ is \emph{not} $2$-agreeing.
\end{example}
We will move on to show that the agreement or lack thereof in the example is solely dependent on the parity of the length of the cycle.
\begin{lemma}
    Let $\stdcomplex$ be a simplicial complex and let $\stdcycle$ be a simple cycle in $\stdcomplex$.
    Then the following two statements hold:
    \begin{itemize}
        \item $\lassignment{f}^e_{\stdcycle}$ is agreeing iff $\stdcycle$ is of even length.
        \item For every value of $k$: $\lassignment{f}^o_{\stdcycle, k}$ is agreeing iff $\stdcycle$ is of odd length.
    \end{itemize}
\end{lemma}
\begin{proof}
    We will start by showing that the claim holds for $\lassignment{f}^e_\gamma$:
    $\stdcycle$ is even iff it is bi-partite, let $V_1, V_2$ be the two parts of the cycle.
    We construct the following global cochains $\stdcochain_1, \stdcochain_2$:
    \[
        \stdcochain_1(\stdvertex) = \begin{cases}
                                        s(\stdcycle_i) & \stdvertex = \stdcycle_i\\
                                        0 & \text{otherwise}
        \end{cases}
        \qquad \qquad \qquad
        \stdcochain_2(\stdvertex) = \begin{cases}
                                        1-s(\stdcycle_i) & \stdvertex = \stdcycle_i\\
                                        0 & \text{otherwise}
        \end{cases}
    \]
    It is east to see that $\lassignment{f}^e_{\stdcycle}$ is an agreeing $2$-assignment that agrees with $\stdcochain_1$ and $\stdcochain_2$.

    The proof of the statment for the odd case follows very similar arguments and is therefore omitted.
\end{proof}
We consider a subset of all simple cycles defined below:
\begin{definition}[Skipping edge]
    Let $\stdcomplex$ be a simplicial complex and let $\stdcycle$ be a simple cycle in $\stdcomplex$.
    Also let $e, e'$ be edges such that there exists $i, j, k$ such that $e = \set{\stdcycle_i, \stdcycle_j}, e' = \set{\stdcycle_{k}, \stdcycle_{k+1}}$.
    We say that $e$ skips $e'$ if $\dist_\stdcycle\parens{\stdcycle_i, \stdcycle_j} = \dist_\stdcycle\parens{\stdcycle_i, \stdcycle_k} + \dist_\stdcycle\parens{\stdcycle_k, \stdcycle_j} + 1$.
\end{definition}
\begin{definition}[Non-skipping cycle]
    Let $\stdcomplex$ be a simplicial complex and let $\stdcycle$ be a simple cycle in the underlying graph of the complex.
    We say that $\stdcycle$ is $i$-non-skipping if for every $j,k$ such that $e = \set{\stdcycle_j, \stdcycle_k} \in \stdcomplex(1)$ it holds that $e$ skips at most $i$ edges.
\end{definition}
\begin{definition}[Gluing]
    Let $\stdcomplex$ be a simplicial complex, $\stdcycle$ be a simple cycle and $\lassignment{f}_{\stdcycle}$ be a coloring candidate.
    Define the gluing of its $k$-edge as $glue_k(\lassignment{f})$ to be a $2$-assignment such that:
    \begin{itemize}
        \item Any edge that does not skip $\set{\stdcycle_k, \stdcycle_{k+1}}$ is left unchanged.
        \item The edge $\set{\stdcycle_k, \stdcycle_{k+1}}$ is assigned $\sparens{\stdvertex=1, \genvertex=1}, \sparens{\stdvertex=0, \genvertex=0}$.
        \item Any edge that skipped $\set{\stdcycle_k, \stdcycle_{k+1}}$ and was assigned $\sparens{\stdvertex=0, \genvertex=1}, \sparens{\stdvertex=1, \genvertex=0}$ will be assigned $\sparens{\stdvertex=1, \genvertex=1}, \sparens{\stdvertex=0, \genvertex=0}$.
        \item Any edge that skipped $\set{\stdcycle_k, \stdcycle_{k+1}}$ and was assigned $\sparens{\stdvertex=1, \genvertex=1}, \sparens{\stdvertex=0, \genvertex=0}$ will be assigned $\sparens{\stdvertex=0, \genvertex=1}, \sparens{\stdvertex=1, \genvertex=0}$.
    \end{itemize}
\end{definition}
\begin{lemma}
    Let $\stdcomplex$ be a simplicial complex, $\stdcycle$ be a simple cycle.
    Then for every $j \ne k$: $\lassignment{f}^o_{\stdcycle,k}$ is agreeing iff $glue_j\parens{\lassignment{f}^o_{\stdcycle,k}}$ is not agreeing.
\end{lemma}
\begin{proof}
    Consider $\lassignment{f}^o_{\stdcycle,k}$.
    Note that it is agreeing iff $\stdcycle$ is of odd length.
    It will therefore suffice to prove that $glue_j\parens{\lassignment{f}^o_{\stdcycle,k}}$ is only agreeing if the cycle is of even length.
    Note that, by the definition of gluing, $glue_j\parens{\lassignment{f}^o_{\stdcycle,k}}$ is agreeing iff $\lassignment{f}^o_{\stdcycle',k}$ is agreeing where $\stdcycle'$ is $\stdcycle$ with $\stdcycle_j$ and $\stdcycle_{j+1}$ ``glued together'' (i.e.\ considered to be one vertex).
    Therefore $\lassignment{f}^o_{\stdcycle',k}$ is agreeing iff $\stdcycle'$ is of odd length and thus $glue_j\parens{\lassignment{f}^o_{\stdcycle,k}}$ is agreeing iff $\stdcycle$ is of even length.
\end{proof}
We now show that the existence of a non-skipping in the complex yields a lower bound on the number of queries required to test list agreement.
\begin{lemma}\label{lem:cycle-lower-bound}
    Let $\stdcomplex$ be a simplicial complex and let $\stdcycle$ be an $i$-non-skipping simple cycle in $\stdcomplex$.
    Then there is a lower bound of $\frac{\abs{\stdcycle}}{i}$ queries that \emph{must} be performed in order to test list agreement.
\end{lemma}
\begin{proof}
    We will show that it is impossible to distinguish between $\lassignment{f}^e_\stdcycle$ and $\lassignment{f}^o_\stdcycle$ using less than $\frac{\abs{\stdcycle}}{i}$ queries.
    Assume, by way of contradiction, that there exists a test $T$ that distinguishes between $\lassignment{f}^e_\stdcycle$ and $\lassignment{f}^o_\stdcycle$ using less than $\frac{\abs{\stdcycle}}{i}$ queries.

    Note that querying any edge whose vertices are not part of the cycle does not help in distinguishing between $\lassignment{f}^e_\stdcycle$ and $\lassignment{f}^o_\stdcycle$.
    We can therefore assume that $T$ only queries edges that are a part of the cycle (or a skipping edge).
    Let $q_1,\cdots,q_{\frac{\abs{\stdcycle}}{i}-1}$ be the faces queried by $T$.
    Consider the structure of these queries and note that there exists an edge $e=\set{\stdcycle_k, \stdcycle_{k+1}}$ that none of the queries skip as the queries skip at most $\abs{\stdcycle}-i$ edges and the length of the cycle is strictly larger than that.
    The following holds:
    \begin{itemize}
        \item If $\stdcycle$ is of odd length, note that performing $q_1,\cdots,q_{\frac{\abs{\stdcycle}}{i}-1}$ on $\lassignment{f}^e_{\stdcycle}$ and $\lassignment{f}^o_{\stdcycle, k}$ yield the same results.
        \item If $\stdcycle$ is of even length, note that $\lassignment{f}^o_{\stdcycle,0}$ is not an agreeing $2$-assignment.
        Also note that if $e=\set{\stdcycle_0, \stdcycle_1}$ then performing $q_1,\cdots,q_{\frac{\abs{\stdcycle}}{i}-1}$ on $\lassignment{f}^o_{\stdcycle, 0}$ and $\lassignment{f}^e_{\stdcycle}$ yield the same results.
        In addition, if $e \ne \set{\stdcycle_0, \stdcycle_1}$ then performing $q_1,\cdots,q_{\frac{\abs{\stdcycle}}{i}-1}$ on $\lassignment{f}^o_{\stdcycle, 0}$ and $glue_k\parens{\lassignment{f}^o_{\stdcycle, 0}}$ yield the same results.
    \end{itemize}
    We can therefore finish the proof by applying Lemma~\ref{tests-fail}.
\end{proof}
\begin{remark}
    In Lemma~\ref{lem:cycle-lower-bound} we used a different querying model to the model we have been using in the rest of the paper:
    In the model of queries we used in Lemma~\ref{lem:cycle-lower-bound} a query accepts a face and returns \emph{all} of its assignments while in the rest of the paper we used a model where each assignment is queried on its own.
    It is easy to see that a lower bound on the number of queries performed in Lemma~\ref{lem:cycle-lower-bound} holds for the model used at the rest of the paper (As one can think of queries in the former model as queries in the latter model that get some additional ``free'' information).
\end{remark}
The rest of this section is dedicated to showing that complexes of interest have \emph{non-constant} non-skipping cycles and thus without the $2$-differing assumption they have no test for list agreement that performs a constant number of queries.
\subsection{Homology and List Agreement Testing Without the $2$-Differing Assumption}\label{subsec:homology-and-list-agreement-testing-without-the-2-differing-assumption}
One property of interest for high dimensional expanders is their homology.
We show that a lower bound on the homology of the complex yield a lower bound on the number of queries required in order to test list agreement.
It is important to stress that, while the complexes that this result applies to have no non-zero members of the homology, there are many complexes of interest whose minimal non-zero member of the homology is growing with the size of the complex.
We begin by defining the homology of the complex.
\begin{definition}[Boundary operator]
    Given a cochain $\stdcochain \in \cochainset{k}{\stdcomplex;\mathbbm{F}_2}$ define the boundary $\boundaryoperator_k: \cochainset{k}{\stdcomplex;\mathbbm{F}_2} \rightarrow \cochainset{k-1}{\stdcomplex;\mathbbm{F}_2}$ operator:
    \[
        \forall \stdface \in \stdcomplex(k-1): \boundaryoperator_k\stdcochain(\stdface) = \sum_{\substack{\genface \in \stdcomplex(k) \\ \stdface \subseteq \genface}}{\stdcochain(\genface)}
    \]
    When it is clear from context the dimension will be omitted.
\end{definition}
Much like we used the coboundary operator to define coboundaries and cocycles we can use the boundary operator to define boundaries and cycles:
\begin{definition}[Cycles and Boundaries]
    Define the boundaries to be:
    \[
        \boundaryset{k}{\stdcomplex; \mathbbm{F}_2} = \im\parens{\boundaryoperator_{k+1}}
    \]
    And the cycles to be:
    \[
        \cycleset{k}{\stdcomplex; \mathbbm{F}_2} = \ker\parens{\boundaryoperator_k}
    \]
    Note that because we take cochains over $\mathbbm{F}_2$ we think of a set of faces and the cochain that sets the value $1$ to members of the set and $0$ otherwise interchangeably.
\end{definition}
Note that the $1$-dimensional cycles are sets of edges that ``touch'' every vertex an even number of times.
We can now define the homology of the complex:
\begin{definition}[Homology]
    Define the $k$-dimensional homology of the complex to be:
    \[
        \homologyset{k}{\stdcomplex;\mathbbm{F}_2} = \cycleset{k}{\stdcomplex;\mathbbm{F}_2} / \boundaryset{k}{\stdcomplex;\mathbbm{F}_2}
    \]
\end{definition}
Generally, the $1$-dimensional homology can be thought of as the $1$-dimensional ``holes'' in the complex.
We use this notion in order to prove the following:
\begin{lemma}\label{lem:homology-and-non-skipping}
    There is a member of the homology that is a $1$-non-skipping cycle.
\end{lemma}
\begin{proof}
    We will show that any non-empty cycle of minimal size (i.e.\ that sets the value $1$ to the minimal number of edges) in the homology of $\stdcomplex$ is a $1$-non-skipping cycle.
    Let $\stdcycle \in \cycleset{1}{\stdcomplex;\mathbbm{F}_2}$ be a cycle.

    We will first show that the cycle ``touches'' each vertex exactly twice.
    Consider the subgraph induced by the $\stdcycle$ and note that it has an Euler cycle.
    If there is a vertex that the cycle touches more than two times, consider the euler cycle that begins in that vertex.
    Note that it can be divided into cycles - one for every two consecutive visits to that vertex.
    Denote these cycles $\stdcycle'_1,\cdots,\stdcycle'_k$.
    Note that $\stdcycle = \sum_{i=1}^k{\stdcycle'_i}$, therefore there exists $i$ such that $\stdcycle'_i \notin \boundaryset{1}{\stdcomplex; \mathbbm{F}_2}$ which contradicts the minimality of $\stdcycle$.

    We move on to show that it is $1$-non-skipping
    Assume that there are $i, j$ such that $i + 1 < j$ and $\set{\stdcycle_i, \stdcycle_j} \in \stdcomplex(1)$.
    consider the following cycles:
    \[
        \stdcycle_1 = \parens{\stdcycle_0, \cdots, \stdcycle_i, \stdcycle_j, \cdots \stdcycle_n} \qquad \qquad \stdcycle_2 = \parens{\stdcycle_i, \cdots, \stdcycle_j, \stdcycle_i}
    \]
    Note that $\stdcycle = \stdcycle_1 + \stdcycle_2$ and therefore either $\stdcycle_1 \notin \boundaryset{1}{\stdcomplex;\mathbbm{F}_2}$ or $\stdcycle_2 \notin \boundaryset{1}{\stdcomplex;\mathbbm{F}_2}$.
    Both, however, are cycles and therefore either $\stdcycle_1 \in \homologyset{1}{\stdcomplex;\mathbbm{F}_2}$ or $\stdcycle_2 \in \homologyset{1}{\stdcomplex;\mathbbm{F}_2}$ which contradicts the minimality of $\stdcycle$.
\end{proof}
\begin{corollary}
    If there is a lower bound on the size of members of the homology then it is also a lower bound on the number of queries that must be performed in order to test list agreement without the $2$-differing assumption.
\end{corollary}
\begin{proof}
    Combining Lemma~\ref{lem:homology-and-non-skipping} and Lemma~\ref{lem:cycle-lower-bound} proves this corollary.
\end{proof}
We note that many complexes of interest have non-vanishing homology.
Moreover, there are cases where the size of members of the homology are bounded from below by a non-constant bound (see~\cite{kaufman2021new}, for example).
This suggests that in many complexes testing list agreement requires a non-constant number of queries.

Before we close this subsection it is important to note that while the homology include $1$-non-skipping cycles it not the case that every $1$-skipping-cycle is a member of the $1$-dimensional homology.
Consider the following cycle:
\begin{example}[A $1$-non-skipping cycle that is not a member of the homology]
    Consider the following complex:
    
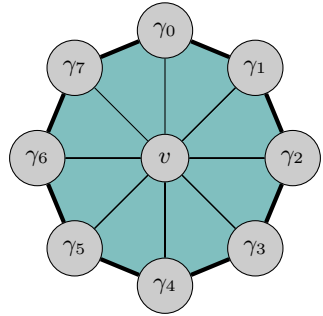
\begin{center}
        \begin{tikzpicture}[%
        every node/.style={draw,fill=gray!40,circle,minimum size=18pt,font=\footnotesize}, node distance=1cm]
            \node (c) {$\stdvertex$};
            \node (g0) [above=of c] {$\stdcycle_0$};
            \node (g1) [above right=of c] {$\stdcycle_1$};
            \node (g2) [right=of c] {$\stdcycle_2$};
            \node (g3) [below right=of c] {$\stdcycle_3$};
            \node (g4) [below=of c] {$\stdcycle_4$};
            \node (g5) [below left=of c] {$\stdcycle_5$};
            \node (g6) [left=of c] {$\stdcycle_6$};
            \node (g7) [above left=of c] {$\stdcycle_7$};
            \begin{pgfonlayer}{bg}
                \draw[fill=teal!50] (g0.center) -- (g1.center) -- (g2.center) -- (g3.center) -- (g4.center) -- (g5.center) -- (g6.center) -- (g7.center) -- cycle;
            \end{pgfonlayer}
            \draw[ultra thick] (g0) -- (g1) -- (g2) -- (g3) -- (g4) -- (g5) -- (g6) -- (g7) -- (g0);
            \draw (g0) -- (c) -- (g1) -- (c) -- (g2) -- (c) -- (g3) -- (c) -- (g4) -- (c) -- (g5) -- (c) -- (g6) -- (c) -- (g7);
        \end{tikzpicture}
    \end{center}
    Note that the bolded cycle in this complex \emph{is} a boundary and yet it is $1$-non-skipping.
\end{example}
\subsection{Spherical Buildings And List Agreement Testing Without the $2$-differing Assumption}\label{subsec:spherical-buildings-and-list-agreement-testing-without-the-2-differing-assumption}
One useful family of coboundary expanders are the spherical buildings.
In this section we will show a lower bound on the number of queries a list agreement must perform when the underlying complexes are the spherical buildings.
We will show said lower bound by finding a $1$-non-skipping cycle whose size increases with the size of the building.
Therefore we conclude that there is no constant query test for list agreement on the spherical buildings.
And, since the spherical buildings have been shown to be both coboundary expanders and local spectral expanders, show a that the $2$-differing assumption is \emph{inherent} for testing list agreement in the domain we consider.
We begin by giving a more detailed description of spherical buildings:
\begin{definition}[Spherical Building]
    Let $p$ be a prime number and let $d$ be a dimension.
    Define the $d$-dimensional spherical building $\sphericalbuilding{p}{d}$ to be a simplicial complex such that:
    \begin{itemize}
        \item Its vertices are the non-trivial subspaces of $\mathbbm{F}_p^{d+2}$ (i.e.\ subsets that are not $\set{0}$ or $\mathbbm{F}_p^{d+2}$).
        \item Its $d$-dimensional spaces are $\stdface=\set{\stdface_1,\dots,\stdface_{d+1}}$ such that:
        \[
            0 < \stdface_1 < \cdots < \stdface_{d+1} < \mathbbm{F}_p^{d+2}
        \]
    \end{itemize}
\end{definition}
We will show the following:
\begin{lemma}\label{lem:1-non-skipping-cycle-in-spherical-building}
    For every prime $p$ there is a $1$-non-skipping cycle of length $2(p-1)$ in $\sphericalbuilding{p}{d}$.
\end{lemma}
We will show that cycle explicitly.
\begin{definition}
    For every $i$ define the following subspaces of $\mathbbm{F}^d_p$ (which correspond to vertices in):
    First consider the following vectors:
    \[
        v_i = \parens{1,i,0,0,0,\cdots,0} \qquad \qquad u_i = \parens{1,0,i,0,0,\cdots,0}
    \]
    And define the following spaces:
    \[
        V_i = \vectorspan{\set{v_i}} \qquad \qquad U_i = \vectorspan{\set{u_i}} \qquad \qquad W_{i,j} = \vectorspan{\set{u_i, v_j}}
    \]
\end{definition}
We will show that the following cycle is $1$-non-skipping:
\[
    \stdcycle = \parens{V_1,W_{1,1},U_{1},W_{1,2},V_2,W_{2,2},U_2,\cdots,U_{p-1}, W_{p-1,1}, V_1}
\]
Note that, by definition, there are no edges between and $V_i$ and $U_j$.
Therefore all we have to show are the following statements:
\begin{itemize}
    \item Out of the vertices that combine the cycle, $W_{i,i}$ is only connected to $V_i$ and $U_i$.
    \item Out of the vertices that combine the cycle, $W_{i,i+1}$ is only connected to $V_{i+1}$ and $U_i$.
    \item Out of the vertices that combine the cycle, $W_{p-1,1}$ is only connected to $V_{1}$ and $U_{p-1}$.
\end{itemize}
\begin{lemma}\label{lem:w-i-i-is-only-directly-connected-to-v-i-and-u-i}
    For every $i \ne j$ it holds that $V_j \nsubseteq W_{i,i}$ and $U_j \nsubseteq W_{i,i}$.
\end{lemma}
\begin{proof}
    We prove the Lemma using the contra-positive argument:
    Assume that $V_j \subseteq W_{i,i}$.
    Therefore there exists $a,b$ such that $v_j = a v_i + b u_i$ and thus:
    \[
        \begin{cases}
            a+b=1\\
            ai = j\\
            bi = 0
        \end{cases}
        \Rightarrow ai+bi=i \Rightarrow j+0=i \Rightarrow i=j
    \]
    Which contradict our assumption.
    Similar claims show that $U_j \nsubseteq W_{i,i}$.
\end{proof}
\begin{lemma}\label{lem:w-i-i+1-is-only-directly-connected-to-v-i+1-and-u-i}
    For every $i \ne j$ it holds that $V_{j+1} \nsubseteq W_{i,i+1}$ and $U_j \nsubseteq W_{i,i}$.
\end{lemma}
\begin{proof}
    We first note that $i \ne 0, p-1$.

    Assume that $V_{j+1} \subseteq W_{i,i+1}$.
    Therefore there exists $a,b$ such that $v_j = a v_{i+1} + b u_i$ and therefore:
    \begin{multline*}
        \begin{cases}
            a+b=1\\
            a(i+1) = j+1\\
            bi = 0
        \end{cases}
        \Rightarrow a(i+1)+b(i+1)=i+1 \Rightarrow j+1+b=i+1 \Rightarrow \\ \Rightarrow b=i-j \Rightarrow (i-j)i = 0
    \end{multline*}
    Note that $i \ne j$ and $i \ne 0$ and therefore $V_{j+1} \nsubseteq W_{i,i+1}$.

    Similarly assume that $U_j \subseteq W_{i,i+1}$.
    Therefore there exists $a,b$ such that $u_j = a v_{i+1} + b u_i$ and therefore:
    \begin{multline*}
        \begin{cases}
            a+b=1\\
            a(i+1) = 0\\
            bi = j
        \end{cases}
        \Rightarrow ai+bi=i \Rightarrow ai+j=i \Rightarrow ai=i-j \Rightarrow \\ \Rightarrow i-j+a=0 \Rightarrow a = j-i \Rightarrow (j-i)(i+1) = 0
    \end{multline*}
    And since $i \ne j$ and $i \ne p-1$ it holds that $U_j \nsubseteq W_{i,i+1}$.
\end{proof}
\begin{lemma}\label{lem:w-p-1-1-is-only-connected-to-u-p-1-and-v-1}
    Let $i \ne 1$ and $j \ne p-1$ therefore it holds that $V_i \nsubseteq W_{p-1,1}$ and $U_j \nsubseteq W_{p-1,1}$.
\end{lemma}
\begin{proof}
    Assume that $V_i \subseteq W_{p-1,i}$.
    Therefore there exists $a,b$ such that $v_i = a v_1 + b u_{p-1}$ and therefore:
    \[
        \begin{cases}
            a+b=1\\
            a = i\\
            b(p-1) = 0
        \end{cases}
        \Rightarrow b=1-i \Rightarrow (1-i)(p-1)=1 \Rightarrow i=1
    \]
    Which contradicts our choice of $i$ and thus $V_i \nsubseteq W_{p-1,1}$

    Likewise assume that $U_j \subseteq W_{p-1,i}$.
    Therefore there exists $a,b$ such that $u_j = a v_1 + b u_{p-1}$ and therefore:
    \[
        \begin{cases}
            a+b=1\\
            a = 0\\
            b(p-1) = j
        \end{cases}
        \Rightarrow b=1 \Rightarrow j=p-1
    \]
    Which contradicts our choice of $j$ and thus $U_j \nsubseteq W_{p-1,1}$
\end{proof}
We are now ready to prove Lemma~\ref{lem:1-non-skipping-cycle-in-spherical-building}.
\begin{proof}[Proof of Lemma~\ref{lem:1-non-skipping-cycle-in-spherical-building}]
    Combining Lemma~\ref{lem:w-i-i-is-only-directly-connected-to-v-i-and-u-i}, Lemma~\ref{lem:w-i-i+1-is-only-directly-connected-to-v-i+1-and-u-i} and Lemma~\ref{lem:w-p-1-1-is-only-connected-to-u-p-1-and-v-1} with the definition of the spherical buildings proves that $\stdcycle$ is a $1$-non-skipping cycle.
\end{proof}
\begin{corollary}\label{cor:lower-bound-on-spherical-buildings}
    There is a lower bound of $2(p-1)$ queries that must be performed in order to test list agreement on $\sphericalbuilding{p}{d}$.
\end{corollary}
\begin{proof}
    Combining Lemma~\ref{lem:cycle-lower-bound} and Lemma~\ref{lem:1-non-skipping-cycle-in-spherical-building} proves this Corollary.
\end{proof}
\section{Testing Direct Sums Using List Agreement Expansion}\label{sec:testing-direct-sums-using-list-agreement-expansion}
In this section we will show how to use a $2$-agreement-expander to provide a test for whether a function is a $k$-direct-sum for any constant $k$.
We will do that by reconstructing an $l$-assignment (with an appropriate choice of $l$) to each $\parens{k+1}$-dimensional face and reducing the problem to list agreement expansion.
We start by defining what a direct sums is:
\begin{definition}[Direct sum]
    Let $\stdcomplex$ be a $d$-dimensional simplicial complex, $i \le d$ and let $F:\stdcomplex(i) \rightarrow \set{0,1}$.
    We say that $F$ is a $(i+1)$-direct-sum if there exists a function $f:\stdcomplex(0) \rightarrow \set{0,1}$ such that:
    \[
        \forall \stdface \in \stdcomplex(i): F(\stdface) = \sum_{\stdvertex \in \stdface}{f(\stdvertex)}
    \]
    We also term $f$ as an origin function of $F$.
\end{definition}
In addition we use the following distance function between functions from the $i$-dimensional faces to $\set{0,1}$:
\begin{definition}[Distance function]
    Let $\stdcomplex$ be a $d$-dimensional simplicial complex and let $F,G: \stdcomplex(i) \rightarrow \set{0,1}$.
    Define:
    \[
        \dist\parens{F,G} = \sum_{\substack{\stdface \in \stdcomplex(i) \\ F(\stdface) \ne G(\stdface)}}{\weight{\stdface}}
    \]
\end{definition}
Let us present the reconstruction (while noting that in the even case the reconstruction yields two functions).
\begin{lemma}
    If $F$ is a $k$-direct-sum ($k$ is even) and $f$ is an origin function of $F$ then $\mathbbm{1}+f$ is also an origin function of $F$.
\end{lemma}
\begin{proof}
    Let $\stdface \in \stdcomplex(k-1)$.
    Consider:
    \begin{equation*}
        F(\stdface)
        = \sum_{\stdvertex \in \stdface}{f(\stdvertex)}
        = \binom{k}{1} 1 + \sum_{\stdvertex \in \stdface}{f(\stdvertex)}
        = \sum_{\stdvertex \in \stdface}{1+f(\stdvertex)}
        = \sum_{\stdvertex \in \stdface}{(\mathbbm{1}+f)(\stdvertex)}
    \end{equation*}
\end{proof}
\begin{lemma}[Reconstructing a $k$-direct-sum]
    \label{reconstructing-k-direct-sum}
    Let $F$ be a $k$-direct-sum and then for every face $\stdface \in \stdcomplex(k+1)$ one can reconstruct two possible origin functions if $k$ is even and a single origin function if $k$ is odd.
    Note that the $\abs{\stdface}=k+2$, this is important as one cannot reconstruct the origin function of a $k$-direct sum using a face whose size is $k+1$.
\end{lemma}
\begin{proof}
    We will show that if $k$ is odd then the following algorithm returns a candidate for an origin function of $F$:\\
    \begin{algorithm}[H]
        \caption{reconstruct origin function for odd values of $k$}\label{alg:reconstruct-origin-function-for-odd-values-of-k}
        \SetAlgoLined
        \DontPrintSemicolon
        \For {every vertex $\stdvertex$}{
        \Pick{a $k$-dimensional face $\stdface$ that includes $\stdvertex$}{
        Set $f(\stdvertex) = \sum_{\face{a} \in \binom{\stdface \setminus \set{\stdvertex}}{k-1}}{F(a \cup \set{\stdvertex})}$
        }
        \Return{$f$}
        }
    \end{algorithm}
    We will now show that this is indeed an origin function of $F$, if $F$ is a $k$-direct-sum then it has an origin function (denoted by $f'$).
    Consider the following:
    \begin{align*}
        &\forall \genface \in \stdcomplex(k): f(\genvertex)
        = \sum_{\face{a} \in \binom{\genface \setminus \set{\genvertex}}{k-1}}{F(a \cup \set{\genvertex})}
        = \sum_{\face{a} \in \binom{\genface \setminus \set{\genvertex}}{k-1}}{\parens{f'(\genvertex) + \sum_{\stdvertex' \in \face{a}}{f'(\stdvertex')}}} =\\
        &\quad \binom{k}{k-1}f'(\genvertex) + \sum_{v' \in q}{\binom{k-1}{k-2}f'(v')} =  k \cdot f'(\genvertex) + \sum_{\stdvertex' \in q}{(k-1) f'(\stdvertex')} = f'(\genvertex)
    \end{align*}
    Therefore $f$ is an origin function of $F$.\\
    In the even case, we will show that the following algorithm reconstruct two options for an origin function of $F$:\\
    \begin{algorithm}[H]
        \caption{reconstruct origin functions for even values of $k$}\label{alg:reconstruct-origin-function-for-even-values-of-k}
        \SetAlgoLined
        \DontPrintSemicolon
        \Pick {$\stdvertex \in \stdcomplex(0)$}{
        Define $f_0, f_1$ to be local origin functions of $F$ and set $f_0(\stdvertex) = 0$ and $f_1(\stdvertex) = 1$.\\
        \For {every other vertex $\stdvertex' \ne \stdvertex$ set}{
        \Pick{a $k$-dimensional face $\stdface$ that includes both $\stdvertex$ and $\stdvertex'$}{
        Set $f_0(\stdvertex') = \sum_{\face{a} \in \binom{\stdface \setminus \set{\stdvertex, \stdvertex'}}{k-2}}{F(a \cup \set{\stdvertex, \stdvertex'})}$\label{alg:reconstruct-origin-function-for-even-values-of-k:find_f0}\\
        Set $f_1(\stdvertex') = 1 + f_0(\stdvertex')$\label{alg:reconstruct-origin-function-for-even-values-of-k:find_f1}
        }
        }
        \Return{$f_0$ and $f_1$}
            }
    \end{algorithm}
    The proof that these are indeed origin functions of $F$ follows the proof of the odd case - $F$ is a direct sum therefore it has an origin function $f'$ (assume WLOG that $f'(\stdvertex) = 0$):
    \begin{align*}
        &\forall \genface \in \stdcomplex(k): f_0(\genvertex)
        = \sum_{\face{a} \in \binom{\genface \setminus \set{\stdvertex, \genvertex}}{k-2}}{F(a \cup \set{\stdvertex, \genvertex})}
        = \sum_{\face{a} \in \binom{\genface \setminus \set{\stdvertex, \genvertex}}{k-2}}{\parens{f'(\stdvertex) + f'(\genvertex) + \sum_{\stdvertex' \in \face{a}}{f'(\stdvertex')}}} =\\
        &\quad \binom{k-1}{k-2}f'(\genvertex) + \sum_{v' \in q}{\binom{k-2}{k-3}f'(v')} =  \parens{k-1} f'(\genvertex) + \sum_{\stdvertex' \in q}{(k-2) \cdot f'(\stdvertex')} = f'(\genvertex)
    \end{align*}
\end{proof}
\begin{corollary}\label{possible-origin-functions}
    If $F$ is a $k$-direct-sum then:
    \begin{itemize}
        \item If $k$ is odd then it has a single origin function.
        \item If $k$ is even then it has two origin functions $f_0$ and $f_1$ such that $f_0 = \mathbbm{1} + f_1$.
    \end{itemize}
\end{corollary}
\begin{proof}
    The proof of this corollary follows exactly the steps shown in Lemma~\ref{reconstructing-k-direct-sum}.
\end{proof}
\begin{note*}
    Regardless of the parity of $k$ the resulting $l$-assignment is locally differing (either trivially since in case that $l=1$ or non trivially since $f_0 = \mathbbm{1} + f_1$).
\end{note*}
Before we move on to describe the test let us first show how to query the origin functions using $k$-dimensional faces:
\begin{corollary}\label{algorithms-are-independent-of-choices}
    If $F$ is a $k$-direct-sum then:
    \begin{itemize}
        \item If $k$ is odd then every choice of $\stdface$ in algorithm~\ref{alg:reconstruct-origin-function-for-odd-values-of-k} results in the same function $f$.
        \item If $k$ is even then any pick of $\stdvertex$ and $\stdface$ in algorithm~\ref{alg:reconstruct-origin-function-for-even-values-of-k} results in either $f_0$ or $f_1$.
    \end{itemize}
\end{corollary}
\begin{proof}
    In Lemma~\ref{reconstructing-k-direct-sum} it is shown that if $F$ is a $k$-direct-sum then when running the algorithm corresponding to the value of $k$ is an origin function of $F$.
    Therefore using Corollary~\ref{possible-origin-functions} these are indeed the desired functions.
\end{proof}
\begin{corollary}\label{querying-origin-functions}
If $F$ is a $k$-direct-sum then for every $\genface \in \stdcomplex(k)$:
\begin{itemize}
    \item If $k$ is odd one can query the values of $f$ on a face $\genface$ using $k+1$ queries.
    \item If $k$ is even one can query the values of $f_0,f_1$ on a face $\genface$ using $k+1$ queries.
\end{itemize}
\end{corollary}
\begin{proof}
    For the odd case run step 3 of Algorithm~\ref{alg:reconstruct-origin-function-for-odd-values-of-k} with $\stdface = \genface$.
    For even values of $k$ pick a vertex $\genvertex$ and set the value of $f_0(\genvertex) = 0$ and $f_1(\genvertex) = 1$ and then run steps~\ref{alg:reconstruct-origin-function-for-even-values-of-k:find_f0}-\ref{alg:reconstruct-origin-function-for-even-values-of-k:find_f1} of algorithm~\ref{alg:reconstruct-origin-function-for-even-values-of-k} with $\stdvertex = \genvertex$ and $\stdface = \genface$.
    Corollary~\ref{algorithms-are-independent-of-choices} guarantees that these local assignments are indeed $f$ or $f_0$ and $f_1$ in the odd or even case respectively.
\end{proof}
Consider the following test for $k$-direct-sum:\\
\begin{algorithm}[H]
    \caption{$k$-direct-sum}\label{alg:k-direct-sum}
    \SetAlgoLined
    \DontPrintSemicolon
    Run the test for $l$-agreement on the $k$-dimensional faces (where $l=1$ if $k$ is odd and $l=2$ if $k$ is even) when querying the local assignment, calculate the origin function(s) using the algorithm described in Corollary~\ref{querying-origin-functions}.
\end{algorithm}
We will show that this tests whether $F$ is a $k$-direct-sum.
\begin{lemma}\label{direct-sums-pass-the-test-with-probability-1}
    If $F$ is a $k$-direct-sum then it passes the test posed in Algorithm~\ref{alg:k-direct-sum} with probability $1$.
\end{lemma}
\begin{proof}
    If $k$ is odd then $f$ is an origin function of $F$ due to Lemma~\ref{reconstructing-k-direct-sum} therefore $\assignment{F} = \set{f|_\stdvertex}_{\stdvertex \in \stdcomplex(k)}$ is the assignment tested which is indeed an agreeing assignment.\\
    If $k$ is even then $f_0$ and $f_1$ are origin functions of $F$ therefore the assignment calculated by Corollary~\ref{querying-origin-functions} (denoted $\lassignment{F}$) satisfied that either $\lassignment{F}_i^\stdface = f_i|_\stdvertex$ or $\lassignment{F}_i^\stdface = f_{1-i}|_\stdface$.
    It is easy to see that $\lassignment{F}$ is an agreeing $2$-assignment and the test always accepts.
\end{proof}
\begin{lemma}\label{k-direct-sum-distance-mesure}
    Let $F:X(k-1) \rightarrow \set{0,1}$ be a cochain and let $\lassignment{F}$ be the $l$-assignment calculated in Corollary~\ref{querying-origin-functions} then:
    \begin{equation*}
        \dist{(F, k\text{-direct-sums})} \le \dist{(\lassignment{F}, \lassignmentset{A})}
    \end{equation*}
\end{lemma}
\begin{proof}
    Let $\lassignment{G}$ be the agreeing $k$-dimensional $l$-assignment closest to $\lassignment{F}$.
    $\lassignment{G}$ is an agreeing $l$-assignment and therefore for every $\stdface \in \stdcomplex(k)$ there exists a permutation $\pi_\stdface$ such that for all $i \in [l]$: $\assignment{G}_i^\stdface = \lassignment{G}^\stdface_{\pi_\stdface(i)}$ is an agreeing assignment.
    Let $\assignment{F}_i$ be the assignment $\assignment{F}_i^\stdface = \lassignment{F}_{\pi_\stdface(i)}^\stdface$.
    Also let $\assignment{G}$ be the assignment $\assignment{G}_{\tilde{i}}$ such that $\tilde{i} = \argmin_{i}{\set{\dist{(\assignment{G}_i, \assignment{F}_i)}}}$.
    Consider $G(\stdface) = \sum_{\stdvertex \in \stdface}{\assignment{G}(\stdvertex)}$.
    It is easy to see that $G$ is a $k$-direct-sum therefore all we have left is to find the distance between $F$ and $G$.
    \begin{align*}
        \dist{(F,G)} &= \norm{\set{\stdface \in \stdcomplex(k-1) \suchthat F(\stdface) \ne G(\stdface)}}
        \le \norm{\containment^{k}\parens{\set{\stdface \in \stdcomplex(k-1) \suchthat F(\stdface) \ne G(\stdface)}}}\\
        &= \norm{\set{\genface \in \stdcomplex(k) \suchthat \exists \stdface \in \binom{\genface}{k}: F(\stdface) \ne G(\stdface)}}
        = \norm{\set{\genface \in \stdcomplex(k) \suchthat \exists \stdface \in \binom{\genface}{k}: \assignment{F}_{\tilde{i}}^\stdface \ne \assignment{G}^\stdface}}\\
        &= \dist{(\assignment{F}_{\tilde{i}}, \assignment{G})}
        = \frac{1}{l}\sum_{i=1}^l{\dist{(\assignment{F}_{\tilde{i}}, \assignment{G})}}
        \le \frac{1}{l}\sum_{i=1}^l{\dist{(\assignment{F}_{i}, \assignment{G}_i)}}
        = \dist{(\lassignment{F}, \lassignmentset{A})}
    \end{align*}
\end{proof}
\begin{theorem}\label{thm:direct-sum}
Any simplicial complex that is a $\gamma$-list agreement expansion supports a $\parens{3\parens{k+1}, \gamma}$-test for $k$-direct-sums.
\end{theorem}
\begin{proof}
Note that combing Corollary~\ref{querying-origin-functions}, Lemma~\ref{direct-sums-pass-the-test-with-probability-1} and Lemma~\ref{k-direct-sum-distance-mesure} yields a $\parens{3\parens{k+1}, \gamma}$-test for $k$-direct-sums.
\end{proof}
    \bibliographystyle{alpha}
    \bibliography{bibliography}

    \begin{appendices}
\section{On the Testability of Coboundaries in the Representation Complex}\label{sec:on-the-testability-of-co-boundaries-in-the-representation-complex}

In this section we will show how test the coboundaries from the cochains in the representation complex.
It is important to note that the representation complexes are expanding but their cohomology is not trivial.
Before presenting the test we have to define a $\parens{k-1}$-empty-triangle:
\begin{definition}[$\parens{k-1}$-empty-triangle]
    Let $\stdcomplex$ be a simplicial complex and let $\repcplx{\stdcomplex}$ be its $k$-dimension's representation complex.
    Define a $\parens{k-1}$-empty-triangle in $\repcplx{\stdcomplex}$ to be a set of
    vertices $\set{u,v,w}$ such that all there is an edge connecting every pair of them and  $\set{u \cap v, v \cap w, w \cap u} \in \repcplx[k-1]{\stdcomplex}(2)$ (with some orientation).
    Denote the set of $\parens{k-1}$-empty-triangles by $\repcplx{\stdcomplex}(\triangle)$.
    When the dimension is clear from context we will omit it.
\end{definition}
Consider the following test:\\
\begin{algorithm}[H]\caption{Empty Triangle Test}\label{alg:empty-triangle-test}
\SetAlgoLined
\DontPrintSemicolon
\Pick{with probability 0.5.}{
    Pick a triangle $(u,v,w)$ with respect to the norm of $\repcplx{\stdcomplex}$.\\
    Pick a $\parens{k-1}$-empty-triangle $(u,v,w)$ with respect to the norm of $\repcplx[k-1]{\stdcomplex}$.
}
\Return whether $\cochain{f}(u,v)\cochain{f}(v,w)\cochain{f}(w,u) = 1$\footnote{Note that $\cochain{f}$ is only defined on edges from $\repcplx{\stdcomplex}$ and not $\repcplx[k-1]{\stdcomplex}$.}.
\end{algorithm}
We will show that the empty triangle test is a property test for whether the cochain $f$ is a coboundary.
As previously stated, it is known that the representation complex is expanding (i.e.\ given an oracle access to a cochain $\cochain{f}$ one can test the distance of $\cochain{f}$ from the cocycles).
Therefore we would like to differentiate between coboundaries and non-trivial cohomology components.
In order to do so, it would be useful to examine the structure of cocycles in the representation complex.
Specifically we will use the fact that the original complex's links are coboundary expanders in order to claim that a cocycle in the representation complex is a coboundary around every core.
\begin{lemma}
    Let $\cochain{f} \in \cocycleset{1}{\repcplx{\stdcomplex};G}$ then for every $c \in \stdcomplex(k-1)$ it holds that $\cochain{f}|_{\repcplxcore{c}{\stdcomplex}} \in \coboundaryset{1}{\repcplx{\stdcomplex}}$.
\end{lemma}
\begin{proof}
    Since $\cochain{f} \in \cocycleset{1}{\repcplx{\stdcomplex};G}$ it holds that for every $(u,v,w) \in \repcplx{\stdcomplex}(2)$ it holds that \\
    $\cochain{f}(u,v)\cochain{f}(v,w)\cochain{f}(w,u) = 1$.
    Therefore for every $(u,v,w) \in \repcplxcore{c}{\stdcomplex}(2)$ the same holds.
    Consequently $\cochain{f}|_{\repcplxcore{c}{\stdcomplex}}$ is a cocycle in $\repcplxcore{c}{\stdcomplex} \cong \stdcomplex_c$.
    The lemma holds since every cocycle in $\stdcomplex_c$ is a coboundary (due to the fact that $\stdcomplex_c$ is a coboundary expander).
\end{proof}
We use this lemma in order to conclude that the only way in which cocycles differ from each other is by how the coboundaries of different cores are ``attached'' to one another.
In the next section we will define that notion more precisely and investigate the properties of different of such an attachment.

\subsection{The Attachment Map}\label{subsec:the-attachment-map}
In order to discuss the ways coboundaries around cores are attached to one another we define the following attachment map:
\begin{definition}[Attachment maps]
    Let $\cochain{f} \in \cocycleset{1}{\repcplx{\stdcomplex};S_l}$ be a cocycle in $\stdcomplex$ and let \\
    $\set{h^c}_{c \in \stdcomplex(k-1)}$ be functions such that:
    \begin{itemize}
        \item $h^c: c \rightarrow S_l$
        \item $\forall c \in \stdcomplex(k-1): \cochain{f}|_{\repcplxcore{c}{\stdcomplex}} = d_0 h^c$.
    \end{itemize}
    Define the attachment map of $\cochain{f}$ according to $\set{h^c}$ denoted by $\cochain{\check{f}}_{\set{h^c}} \in \cochainset{1}{\repcplx[k-1]{\stdcomplex}}$ to be:
    \[
        \cochain{\check{f}}_{\set{h^c}}(\check{u},\check{v}) = \parens{h^{\check{u}}\parens{\check{u} \cup \check{v}}}^{-1}h^{\check{v}}\parens{\check{u} \cup \check{v}}
    \]
\end{definition}
We will begin by showing a connection between the distance of an attachment map of $\cochain{f}$ from being a cocycle and the distance of $\cochain{f}$ from being a coboundary.
Then, in subsection~\ref{subsec:bounding-the-distance-of-a-co-chain-from-the-co-boundaries}, we will show how to use this in order to show that the empty triangle test is indeed a test for the coboundaries in the representation complex.
But before we do that let us first show that the distance between the attachment map and the coboundaries is invariant of the choice of local coboundaries.
In order to do that we will begin by characterizing the different choices of functions $h$ such that $\cochain{f}|_{\repcplxcore{c}{\stdcomplex}} = d_0 h$:
\begin{lemma}\label{co-boundary-charachterization}
    Let $\stdcomplex$ be a simplicial complex such that the graph $(\stdcomplex(0), \stdcomplex(1))$ is connected and let $h_1, h_2 \in \cochainset{0}{\stdcomplex;S_l}$ such that $\cochain{f}=d_0 h_1 = d_0 h_2$ then there exists $\sigma \in S_l$ such that $h_1(u) = h_2(u)\sigma$.
\end{lemma}
\begin{proof}
    Let $u_0\in \stdcomplex(0)$ be a vertex.
    Also let $\sigma=h_2^{-1}(u_0)h_1(u_0)$ and note that $h_1(u_0) = h_2(u_0)\sigma$.
    For every vertex $u \in \stdcomplex(0)$ pick a path $\gamma=(u,u_{n-1},\dots,u_{0})$ between $u$ and $u_0$ and consider:
    \[
        h_1(u)h_1^{-1}(u_0) = \cochain{f}(u, u_{n-1})\cochain{f}(u_{n-1}, u_{n-2})\cdots \cochain{f}(u_1, u_0) = h_2(u)h_2^{-1}(u_0)
    \]
    And therefore:
    \[
        h_1(u) = h_2(u)h_2^{-1}(u_0)h_1^{-1}(u_0) = h_2(u) \sigma
    \]
\end{proof}
We are now ready to prove the invariance of the distance to the choice of functions $h$:
\begin{lemma}
    Let $\set{h^c_1}_{c \in \stdcomplex(k-1)}$ and $\set{h^c_2}_{c \in \stdcomplex(k-1)}$ be two choices of functions such that $\cochain{f}(u,v) = h^{u \cap v}_1(u)\parens{h^{u \cap v}_1(u)}^{-1} = h^{u \cap v}_2(u)\parens{h^{u \cap v}_2(u)}^{-1}$.
    Then:
    \[
        \dist\parens{\cochain{\check{f}}_{h_1}, \coboundaryset{1}{\repcplx[k-1]{\stdcomplex};S_l}} = \dist\parens{\cochain{\check{f}}_{h_2}, \coboundaryset{1}{\repcplx[k-1]{\stdcomplex};S_l}}
    \]
\end{lemma}
\begin{proof}
    Let $\epsilon = \dist\parens{\cochain{\check{f}}_{h_1}, \coboundaryset{1}{\repcplx[k-1]{\stdcomplex};S_l}}$ therefore there exists $\varphi_{h_1}, \psi_{h_1}$ such that $\varphi_{h_1}$ is a coboundary, $\norm{\psi_{h_1}} = \epsilon$, and $\varphi_{h_1} = \cochain{\check{f}}_{h_1} \psi_{h_1}^{-1}$.
    Let $\phi_{h_1}$ be a cochain such that $\varphi_{h_1} = d_0 \phi_{h_1}$.
    Note that around every core the complex is connected\footnote{Due to the fact that the links of $\stdcomplex$ are coboundary expanders their $0$-cohomology is trivial and therefore the complex is connected.} therefore due to Lemma~\ref{co-boundary-charachterization} it holds that:
    \[
        \forall c \in \stdcomplex(k-1) \exists \stdface_c: h_2^c(u) = h_1^c(u)\stdface_c
    \]
    Therefore:
    \[
        \cochain{\check{f}}_{h_2}(\check{u},\check{v}) = \sigma_{\check{u}}^{-1}\parens{h_1^{\check{u}}(\check{u}\cup \check{v})}^{-1}h_1^{\check{v}}(\check{u}\cup \check{v})\sigma_{\check{v}}
    \]
    Now, consider $\phi_{h_2}(\check{u}) = \sigma_{\check{u}}^{-1} \phi_{h_1}(\check{u})$ and:
    \[
        \varphi_{h_2}(\check{u},\check{v}) = d_0 \phi_{h_2}(\check{u},\check{v}) = \sigma_{\check{u}}^{-1} \phi_{h_1}(\check{u})\phi_{h_1}^{-1}(\check{v})\sigma_{\check{v}}
    \]
    We will conclude the proof by sowing that $\dist(\cochain{f}_{h_1}, \varphi_{h_1}) = \dist(\cochain{f}_{h_2}, \varphi_{h_2})$:
    \begin{align*}
        \varphi_{h_2}(\check{u},\check{v}) = \cochain{\check{f}}_{h_2}(\check{u},\check{v}) & \Leftrightarrow \sigma_{\check{u}}^{-1} \phi_{h_1}(\check{u})\phi_{h_1}^{-1}(\check{v})\sigma_{\check{v}} = \sigma_{\check{u}}^{-1}\parens{h_1^{\check{u}}(\check{u}\cup \check{v})}^{-1}h_1^{\check{v}}(\check{u}\cup \check{v})\sigma_{\check{v}} \Leftrightarrow \\
                                               & \Leftrightarrow \phi_{h_1}(\check{u})\phi_{h_1}^{-1}(\check{v})=\parens{h_1^{\check{u}}(\check{u}\cup \check{v})}^{-1}h_1^{\check{v}}(\check{u}\cup \check{v}) \Leftrightarrow \varphi_{h_1}(\check{u},\check{v}) = \cochain{\check{f}}_{h_1}(\check{u},\check{v})
    \end{align*}
\end{proof}
Due to the invariance of the choice of functions $h$ whenever we refer to the attachment map of a cocycle $\cochain{f}$ we would consider any cocycle and will omit the choice of $h$ from the notation.
We will now show a connection between the distance of attachments maps from the coboundaries and the distance of a cocycle from the coboundaries:
\begin{lemma}\label{first-connection-to-lower-representation}
    Let $\cochain{f}\in \cocycleset{1}{\repcplx{\stdcomplex};S_l}$ be a cocycle in $\repcplx{\stdcomplex}$ such that $\cochain{\check{f}}$ is $\epsilon$-far from being a coboundary, i.e.~there exists $\varphi$ and $\psi$ such that $\norm{\psi} \le \epsilon$, $\varphi=d_0 \widetilde{g}$ is a coboundary and $\varphi=\cochain{\check{f}}\psi^{-1}$.
    In addition for every vertex $u$ let $c_u$ be a core of $u$ (which we will refer to as the canonical core of $c$).
    Consider $g(u)=h^{c_u}(u)\widetilde{g}(c_u)$ and $\cochain{\widetilde{f}} = d_0 g$ then:
    \begin{equation*}
        \dist\parens{\cochain{f}, \cochain{\widetilde{f}}} = \norm{\set{(u,v) \in \repcplx{\stdcomplex}(1) \suchthat \psi(c_u, u \cap v) = \parens{\varphi(c_v, u \cap v)}^{-1} \psi(c_v, u \cap v)\varphi(c_v, u \cap v)}}
    \end{equation*}
\end{lemma}
\begin{proof}
    Note that:
    \begin{equation*}
        \cochain{\widetilde{f}} = h^{c_u}(u)\widetilde{g}(c_u)\parens{\widetilde{g}(c_v)}^{-1}\parens{h^{c_v}(v)}^{-1}
    \end{equation*}
    Now, consider when $\cochain{f}\ne \cochain{\widetilde{f}}$:
    \begin{align*}
        \cochain{f} \ne \cochain{\widetilde{f}} \Leftrightarrow& h^{u \cap v}(u)\parens{h^{u \cap v}(v)}^{-1} \ne h^{c_u}(u)\widetilde{g}(c_u)\parens{\widetilde{g}(c_v)}^{-1}\parens{h^{c_v}(v)}^{-1} \Leftrightarrow \\
                                  &\parens{h^{c_u}(u)}^{-1}h^{u \cap v}(u)\parens{h^{u \cap v}(v)}^{-1}h^{c_v}(v) \ne \widetilde{g}(c_u)\parens{\widetilde{g}(c_v)}^{-1} \Leftrightarrow \\
                                  & \cochain{\check{f}}(c_u, u \cap v) \cochain{\check{f}}(u \cap v, c_v) \ne \widetilde{g}(c_u)\parens{\widetilde{g}(c_v)}^{-1} \Leftrightarrow \\
                                  & \varphi(c_u, u \cap v)\psi(c_u, u \cap v)\varphi(u \cap v, c_v)\psi(u \cap v, c_v) \ne \widetilde{g}(c_u)\parens{\widetilde{g}(c_v)}^{-1} \Leftrightarrow \\
                                  & \widetilde{g}(c_u)\parens{\widetilde{g}(u \cap v)}^{-1} \psi(c_u, u \cap v) \varphi(u \cap v, c_v) \psi(u \cap v, c_v) \ne \widetilde{g}(c_u)\parens{\widetilde{g}(c_v)}^{-1} \Leftrightarrow \\
                                  & \varphi(c_v, u \cap v) \psi(c_u, u \cap v) \varphi(u \cap v, c_v) \psi(u \cap v, c_v) \ne 1 \Leftrightarrow \\
                                  & \psi(c_u, u \cap v) = \parens{\varphi(c_v, u \cap v)}^{-1} \psi(c_v, u \cap v)\varphi(c_v, u \cap v)
    \end{align*}
\end{proof}
We are going to bound $\dist(\cochain{f}, \cochain{\widetilde{f}})$ using the stronger condition presented in the following corollary
\begin{corollary}\label{first-connection-between-psi-f-and-tilde-f}
    Let $\cochain{F}, \cochain{\widetilde{F}}, \psi$ as in Lemma~\ref{first-connection-to-lower-representation}.
    If $\cochain{f}(u,v) \ne \cochain{\widetilde{f}}(u,v)$ then either $\psi(c_u, u \cap v) \ne 1$ or $\psi(c_v, u \cap v) \ne 1$ (or both).
\end{corollary}
\begin{proof}
    If $\psi(c_u, u \cap v) = \psi(c_v, u \cap v) = 1$ then:
    \begin{equation*}
        \psi(c_u, u \cap v) = \parens{\varphi(c_v, u \cap v)}^{-1} \psi(c_v, u \cap v)\varphi(c_v, u \cap v) \Leftrightarrow \parens{\varphi(c_v, u \cap v)}^{-1} \varphi(c_v, u \cap v) = 1
    \end{equation*}
    Since by definition $\parens{\varphi(c_v, u \cap v)}^{-1} \varphi(c_v, u \cap v) = 1$ it holds that $\cochain{f}(u,v) = \cochain{\widetilde{f}}(u,v)$.
\end{proof}
We will now move on to show that the $\dist\parens{\cochain{f}, \cochain{\widetilde{f}}}$ can be bound from above by $\norm{\psi}$.
Before we are able to do that, however, we have to first understand the structure of the empty triangles (which we explore in subsection~\ref{subsec:on-k-1-empty-triangles}) as well as further explore the connection between a cocycle and its attachment maps (which we explore in subsection~\ref{subsec:on-the-connection-between-a-co-cycle-and-its-attachment-map}).
\subsection{On the $\parens{k-1}$-empty-triangles}\label{subsec:on-k-1-empty-triangles}
In this section we will discuss the structure of the $\parens{k-1}$-empty-triangles in the $k$-dimensional representation complex.
Let us begin by describing the $\parens{k-1}$-empty-triangles:
\begin{lemma}\label{finding-the-empty-triangles}
    Let $(u,v) \in \repcplx{\stdcomplex}(1)$ and for any choice of $A \in \binom{u \cap v}{k-1}$ let $w=\parens{u \triangle v} \cup A$ then $\set{u,v,w}$ is a $\parens{k-1}$-empty-triangle.
\end{lemma}
\begin{proof}
    We will first prove that all the edges of the $\parens{k-1}$-empty-triangle exist.
    We will show WLOG that $\abs{u \cap w} = k$:
    \begin{equation*}
        \abs{u \cap w} = \abs{u \cap \parens{\parens{u \triangle v} \cup A}} = \abs{\parens{u \cap \parens{u \triangle v}} \cup \parens{u \cap A}} = \abs{u \cap \parens{u \triangle v}}+ \abs{u \cap A} = 1 + k - 1 = k
    \end{equation*}
    All we have left to prove is that $\set{u \cap v, v \cap w, w \cap u} \in \repcplx[k-1]{\stdcomplex}(2)$.
    Note that the intersection between any pair of vertices out of $\set{u \cap v, v \cap w, w \cap u}$ is the same and equals to $u \cap v \cap w$.
    Therefore all we have to prove is that $\abs{\set{u \cap v, v \cap w, w \cap u}}= k-1 $ and that $\parens{u \cap v} \cup \parens{v \cap w} \cup \parens{w \cap u} \in \stdcomplex(k+1)$ in order to prove the lemma.
    Consider the following:
    \begin{equation*}
        \abs{\parens{u \cap v} \cap \parens{v \cap w}} = \abs{u \cap v \cap w} = \abs{A} = k-1
    \end{equation*}
    Finally note that:
    \begin{align*}
        \parens{u \cap v} \cup \parens{v \cap w} \cup \parens{w \cap u} &= \parens{u \cap v} \cup \parens{\parens{v \setminus u} \cup A} \cup \parens{\parens{u \setminus v} \cup A} =& \\
        &=\parens{u \cap v} \cup \parens{v \setminus u} \cup \parens{u \setminus v} = u \cup v \in \stdcomplex(k+1)
    \end{align*}
    And therefore $w$ forms an empty triangle with $u,v$.
\end{proof}
\begin{lemma}\label{structure-of-an-empty-triangle}
    Let $w \in \repcplx{\stdcomplex}(0)$.
    If there exist $u,v \in \repcplx{\stdcomplex}(0)$ such that $\set{w,u,v}$ is an empty triangle in $\repcplx{\stdcomplex}$ then there exists $A \in \binom{u \cap v}{k-1}$ such that $w=\parens{u \triangle v} \cup A$
\end{lemma}
\begin{proof}
    First note that $\set{w,u,v}$ is an empty triangle and therefore:
    \begin{equation*}
        k-1 = \abs{\parens{u \cap v} \cap \parens{v \cap w} \cap \parens{w \cap u}} = \abs{u \cap v \cap w}
    \end{equation*}
    Let $A=u \cap v \cap w$.
    We will now prove that $w = \parens{u \triangle v} \cup A$:
    let $a \in \parens{u \triangle v} \cup A$.
    If $a \in A$ then $a \in w$.
    Otherwise $a \in u \triangle v$, assume WLOG that $a \in u$.
    If $a \notin w$ then $w = \parens{u \cup v} \setminus \set{a} = v$ which contradicts the fact that $w \ne v$.
    Therefore $w \subseteq \parens{u \triangle v} \cup A$.
    We will finish the proof by showing that both sets are of the same cardinality:
    \begin{equation*}
        \abs{\parens{u \triangle v} \cup A} = \abs{\parens{u \triangle v}} + \abs{A} = 2+k-1 = k+1 = \abs{w}
    \end{equation*}
\end{proof}
Now that we understand the structure of the empty triangles we can consider how many empty triangles are supported by a single edge in $\repcplx{\stdcomplex}$:
\begin{lemma}\label{supported-empty-triangles}
    Evey edge $(u,v) \in \repcplx{\stdcomplex}(1)$ supports exactly $k$ empty triangles.
\end{lemma}
\begin{proof}
This is immediate from Lemma~\ref{finding-the-empty-triangles} and Lemma~\ref{structure-of-an-empty-triangle}: Due to Lemma~\ref{finding-the-empty-triangles} there exists at least $k$ empty triangles supported by $(u,v)$ and due to Lemma~\ref{structure-of-an-empty-triangle} there cannot be more then $k$ vertices that form an empty triangle with $(u,v)$ (note that this is because there are exactly $k$ different choices of $A$).
\end{proof}
We are now ready to explore the connections between the empty triangles of $\repcplx{\stdcomplex}$ and the triangles of $\repcplx[k-1]{\stdcomplex}$:
\begin{lemma}
    Let $(u,v) \in \repcplx{\stdcomplex}(1)$ and let $w \in \repcplx{\stdcomplex}(0)$ such that $\set{u,v,w}$ is an empty triangle then:
    \begin{equation*}
        \frac{\weightcplx{\repcplx{\stdcomplex}}{(u,v)}}{\weightcplx{\repcplx[k-1]{\stdcomplex}}{(u \cap v, v \cap w, w \cap u)}} = \frac{k}{3}
    \end{equation*}
\end{lemma}
\begin{proof}
    Note that:
    \begin{equation*}
        \weightcplx{\repcplx{\stdcomplex}}{(u,v)} = \frac{1}{\binom{k+2}{k}}\weightcplx{\stdcomplex}{u \cup v}
    \end{equation*}
    In addition note that:
    \begin{equation*}
        \weightcplx{\repcplx[k-1]{\stdcomplex}}{(u \cap v, v \cap w, w \cap u)} = \frac{1}{\binom{k+2}{k+1}}\weightcplx{\stdcomplex}{u \cup v}
    \end{equation*}
    Both due to Lemma~\ref{weight-relation-between-a-complex-and-its-representation} therefore:
    \begin{equation*}
        \frac{\weightcplx{\repcplx{\stdcomplex}}{(u,v)}}{\weightcplx{\repcplx[k-1]{\stdcomplex}}{(u \cap v, v \cap w, w \cap u)}} = \frac{\frac{1}{\binom{k+2}{k}}\weightcplx{\stdcomplex}{u \cup v}}{\frac{1}{\binom{k+2}{k+1}}\weightcplx{\stdcomplex}{u \cup v}} = \frac{\binom{k+2}{k+1}}{\binom{k+2}{k}} = \frac{k}{3}
    \end{equation*}
\end{proof}
\begin{corollary}\label{triangles-affected-by-change}
    The weight of all the empty triangles supported by a single edge $(u,v)\in \repcplx{\stdcomplex}$ satisfies:
    \[
        \frac{\weightcplx{\repcplx{\stdcomplex}}{(u,v)}}{\sum_{\substack{w\in\repcplx{\stdcomplex}(0) \\ \set{u,v,w} \in \repcplx{\stdcomplex}(\triangle)}}{\weightcplx{\repcplx[k+1]{\stdcomplex}}{(u \cap v, v \cap w, w \cap u)}}} = \frac{1}{3}
    \]
\end{corollary}
\begin{proof}
    Note that, due to Lemma~\ref{supported-empty-triangles} there are exactly $k$ possible options for $w$, denote them $w_1,\dots, w_k$.
    Now consider:
    \begin{align*}
        &\frac{\weightcplx{\repcplx{\stdcomplex}}{(u,v)}}{\sum_{\substack{w\in\repcplx{\stdcomplex}(0) \\ \set{u,v,w} \in \repcplx{\stdcomplex}(\triangle)}}{\weightcplx{\repcplx[k+1]{\stdcomplex}}{(u \cap v, v \cap w, w \cap u)}}} =
        \frac{\weightcplx{\repcplx{\stdcomplex}}{(u,v)}}{\sum_{i=1}^k{\weightcplx{\repcplx[k+1]{\stdcomplex}}{(u \cap v, v \cap w_i, w_i \cap u)}}} = \\
        &\quad=\frac{\weightcplx{\repcplx{\stdcomplex}}{(u,v)}}{k \cdot \weightcplx{\repcplx[k-1]{\stdcomplex}}{(u \cap v, v \cap w, w \cap u)}} = \frac{k}{3k} = \frac{1}{3}
    \end{align*}
\end{proof}
Now that we have further explored the connection between the edges of $\repcplx{\stdcomplex}$ and the triangles of $\repcplx[k-1]{\stdcomplex}$ we are ready to explore the connection between a cochain and its attachment maps.
\subsection{On the Connection Between a Cocycle and Its Attachment Map}\label{subsec:on-the-connection-between-a-co-cycle-and-its-attachment-map}
We are now ready to further explore the connection between a cocycle and its attachment maps.
Before we do that, however, it would be useful to define a version of the $1$-faces that are not oriented:
\begin{definition}
    Let $\stdcomplex$ be a simplicial complex.
    Define the non-oriented set of $1$-faces to be:
    \[
        \stdcomplex(1)^+ = \set{(u,v) \suchthat (u,v)\in \stdcomplex(1) \text{ or } (v,u)\in \stdcomplex(1)}
    \]
    In addition, we would consider the weight of both $(u,v)$ and $(v,u)$ to be the same.
\end{definition}
Let us now describe another connection between the $1$-faces of $\repcplx[k-1]{\stdcomplex}$ and $\repcplx{\stdcomplex}$:
\begin{lemma}\label{edges-in-k-and-k-1}
    Let $u \in \repcplx{\stdcomplex}(0)$ and let $c_u, \check{v}$ be two cores that are contained in $u$.
    Then:
    \[
        \weightcplx{\repcplx[k-1]{\stdcomplex}}{(c_u, \check{v})} = \frac{1}{k} \sum_{\substack{(u,u')\in \repcplx{\stdcomplex}(1)^+ \\ u \cap u'=\check{v}}}{\weightcplx{\repcplx{\stdcomplex}}{(u,u')}}
    \]
\end{lemma}
\begin{proof}
    Consider the weight of $(c_u,\check{v})$:
    \begin{equation}\label{eq:weight-of-represention}
        \weightcplx{\repcplx[k-1]{\stdcomplex}}{(c_u, \check{v})} = \frac{1}{\binom{k+1}{k-1}} \weightcplx{\stdcomplex}{c_u \cup \check{v}} = \frac{\weightcplx{\stdcomplex}{u}}{\binom{k+1}{k-1}}
    \end{equation}
    Note that:
    \[
        \set{(u,u') \in \repcplx{\stdcomplex}(1) \suchthat u \cap u' = \check{v}} = \set{\repex{k-1}{w}{\check{v}} \suchthat w \in \containment^{k+1}(\set{u})}
    \]
    Therefore, for every $\check{v} \in \stdcomplex(k-1)$ it holds that:
    \begin{align*}
        \sum_{\substack{(u,u')\in \repcplx{\stdcomplex}(1)^+ \\ u \cap u'=\check{v}}}{\weightcplx{\repcplx{\stdcomplex}}{(u,u')}} & = \norm{\set{\repex{k-1}{w}{\check{v}} \suchthat w \in \containment^{k+1}(\set{u})}}_{\repcplx{\stdcomplex}} = \frac{\norm{\containment^{k+1}(\set{u})}_\stdcomplex}{\binom{k+2}{k}} = \\
        & = \frac{\binom{k+2}{k+1}\norm{\set{u}}_\stdcomplex}{\binom{k+2}{k}} = \frac{2}{k+1}\norm{\set{u}}_\stdcomplex = \frac{2}{k+1}\weightcplx{\stdcomplex}{u}
    \end{align*}
    Combining this with~\ref{eq:weight-of-represention} completes the proof.
\end{proof}
Consider the proof of Lemma~\ref{first-connection-to-lower-representation}.
In that proof a special core was selected for every face (which we termed the canonical core).
Using these cores we showed a connection between the distance of a cocycle and the coboundaries and some non-trivial condition on its attachment maps.
We will now refine that connection by showing that, with a careful selection of the canonical cores, one can bound the distance of a cocycle and the coboundaries using the distance of the attachment maps from the coboundaries.
Let us begin by showing the following:
\begin{lemma}\label{choice-of-cores}
    Let $\psi$ be a cochain in $\repcplx[k-1]{\stdcomplex}$ then there exists a choice of $\set{c_u}_{u \in \repcplx{\stdcomplex}(0)}$ such that:
    \[
        \sum_{u \in \stdcomplex(k)}{\sum_{\substack{(c_u,\check{v}) \in \repcplx[k-1]{\stdcomplex}(1)^+ \\ \psi(c_u, \check{v}) \ne 1}}{\weightcplx{\repcplx[k-1]{\stdcomplex}}{(c_u, \check{v})}}} \le \frac{2}{k+1}\norm{\psi}_{\repcplx[k-1]{\stdcomplex}}
    \]
\end{lemma}
\begin{proof}
    Let $\epsilon_u$ be the error seen by edges that represent $u$:
    \[
        \epsilon_u = \norm{\set{(\check{v},\check{w}) \in \repcplx[k-1]{\stdcomplex}(1) \suchthat \check{v} \cup \check{w} = u, \psi(\check{v},\check{w}) \ne 1}}
    \]
    We note that:
    \[
        2 \epsilon_u = \sum_{\check{v} \in \binom{u}{k}}{\sum_{\substack{(\check{v},\check{w}) \in \repcplx[k-1]{\stdcomplex}(1)^+ \\ \psi(\check{v}, \check{w}) \ne 1}}{\weightcplx{\repcplx[k+1]{\stdcomplex}}{(\check{v},\check{w})}}}
    \]
    Note that every face on which $\psi$ differs from $1$ is summed twice (once for each orientation).
    Pick, for every face $u \in \repcplx{\stdcomplex}$, a core $c_u$ such that:
    \[
        \forall \check{v} \in \binom{u}{k}: \sum_{\substack{(c_u,\check{w}) \in \repcplx[k-1]{\stdcomplex}(1)^+ \\ \psi(c_u, \check{w}) \ne 1}}{\weightcplx{\repcplx[k+1]{\stdcomplex}}{(\check{v}, \check{w})}} \le \sum_{\substack{(\check{v},\check{w}) \in \repcplx[k-1]{\stdcomplex}(1)^+ \\ \psi(\check{v}, \check{w}) \ne 1}}{\weightcplx{\repcplx[k+1]{\stdcomplex}}{(\check{v}, \check{w})}}
    \]
    Using this new definition we can see that:
    \begin{align*}
        2 \epsilon_u &= \sum_{\check{v} \in \binom{u}{k}}{\sum_{\substack{(\check{v},\check{w}) \in \repcplx[k-1]{\stdcomplex}(1)^+ \\ \psi(\check{v}, \check{w}) \ne 1}}{\weightcplx{\repcplx[k+1]{\stdcomplex}}{(\check{v},\check{w})}}} \ge \\
                     & \ge \sum_{\check{v} \in \binom{u}{k}}{\sum_{\substack{(c_u,\check{w}) \in \repcplx[k-1]{\stdcomplex}(1)^+ \\ \psi(c_u, \check{w}) \ne 1}}{\weightcplx{\repcplx[k+1]{\stdcomplex}}{(c_u,\check{w})}}} = (k+1)\sum_{\substack{(c_u,\check{w}) \in \repcplx[k-1]{\stdcomplex}(1)^+ \\ \psi(c_u, \check{w}) \ne 1}}{\weightcplx{\repcplx[k+1]{\stdcomplex}}{(\check{v}, \check{w})}}
    \end{align*}
    We finish the proof by noting that:
    \[
        \norm{\psi}_{\repcplx[k-1]{\stdcomplex}} = \sum_{u \in \stdcomplex(k)}{\epsilon_u} \ge \frac{k+1}{2}\sum_{\substack{(c_u,\check{w}) \in \repcplx[k-1]{\stdcomplex}(1)^+ \\ \psi(c_u, \check{w}) \ne 1}}{\weightcplx{\repcplx[k+1]{\stdcomplex}}{(\check{v}, \check{w})}}
    \]
\end{proof}
We are now ready to show a connection between the distance of a cocycle from the coboundaries and the distance of its attachment maps from the coboundaries.
Specifically we show that:
\begin{lemma}
    Let $\cochain{f} \in \cocycleset{1}{\repcplx{\stdcomplex};S_l}$ be a cocycle and $\cochain{\widetilde{f}}$ be the coboundary constructed in Lemma~\ref{first-connection-to-lower-representation}.
    In addition, let $\check{\varphi} \in \coboundaryset{1}{\repcplx[k-1]{\stdcomplex}; S_l}$ and $\check{\psi} \in \cochainset{1}{{\repcplx[k-1]{\stdcomplex}; S_l}}$ such that $\check{\varphi}$ is the closest coboundary to $\cochain{\check{f}}$ and $\check{\varphi} = \cochain{\check{f}} \check{\psi}^{-1}$.
    The following holds:
    \[
        \dist_{\repcplx{\stdcomplex}}(\cochain{f},\cochain{\widetilde{f}}) \le 2 \frac{k}{k-1} \norm{\check{\psi}}
    \]
\end{lemma}
\begin{proof}
    Let $\set{c_u}_{u \in \repcplx{\stdcomplex}(0)}$ be the canonical cores chosen as in Lemma~\ref{choice-of-cores}.
    Consider the following:
    \begin{align*}
        \dist_{\repcplx{\stdcomplex}}(\cochain{f},\cochain{\widetilde{f}}) & = \norm{\set{(u,u')\in \repcplx{\stdcomplex} \suchthat \cochain{f}(u,u') \ne \cochain{\widetilde{f}}(u,u')}}_{\repcplx{\stdcomplex}} \le \\
        & \le \norm{\set{(u,u')\in \repcplx{\stdcomplex} \suchthat \check{\psi}(c_u,u \cap u') \ne 1 \text{ or } \check{\psi}(c_{u'},u \cap u')}}_{\repcplx{\stdcomplex}} \le \\
        & \le \sum_{u \in \repcplx{\stdcomplex}(0)}{\sum_{\substack{\check{v} \in \binom{u}{k} \\ \psi(c_u, \check{v}) \ne 1}}{\sum_{\substack{(u,u') \in \repcplx{\stdcomplex}(1)^{+} \\ u \cap u' = \check{v}}}{\weightcplx{\repcplx{\stdcomplex}}{(u,u')}}}} \le \\
        & \le \sum_{u \in \repcplx{\stdcomplex}(0)}{\sum_{\substack{\check{v} \in \binom{u}{k} \\ \psi(c_u, \check{v}) \ne 1}}{k \weightcplx{\repcplx[k-1]{\stdcomplex}}{(c_u, \check{v})}}} = \\
        & = k \sum_{u \in \repcplx{\stdcomplex}(0)}{\sum_{\substack{(c_u, \check{v}) \in \repcplx[k-1]{\stdcomplex}(1)^{+} \\ \psi(c_u, \check{v}) \ne 1}}{\weightcplx{\repcplx[k-1]{\stdcomplex}}{(c_u, \check{v})}}} \le \\
        & \le k \frac{2}{k+1} \norm{\check{\psi}}_{\repcplx[k-1]{\stdcomplex}} = 2 \frac{k}{k+1} \norm{\check{\psi}}_{\repcplx[k-1]{\stdcomplex}}
    \end{align*}
    The first inequality is due to Corollary~\ref{first-connection-between-psi-f-and-tilde-f}.
    The second inequality is due to the fact that edges $(u,u')$ for which $\check{\psi}(c_u,u \cap u') \ne 1$ and $\check{\psi}(c_{u'},u \cap u')$ are summed twice.
    The third inequality is due to Lemma~\ref{edges-in-k-and-k-1} and the equality that follows it is due to the fact that any two cores of a single face $u$ in the $k$-dimensional representation complex intersect on exactly $k-1$ vertices and their union is exactly $u$.
    The last inequality is due to Lemma~\ref{choice-of-cores}.
\end{proof}
\subsection{Bounding the Distance of a Cochain from the Coboundaries}\label{subsec:bounding-the-distance-of-a-co-chain-from-the-co-boundaries}
We are now finally ready to prove that the empty triangle test is indeed a property test for the coboundaries of the representation complex.
We will do that by analyzing the following process:
\begin{enumerate}
    \item Use the coboundary expansion around every core in order to bound the distance of $\cochain{f}$ from a cocycle $\stdcocycle$.
    \item Pick a set of functions $h^c$ and and find the distance between $\cochain{\check{f}}_{\set{h^c}}$ and the coboundaries of $\repcplx[k-1]{\stdcomplex}$.
    \item Use the aforementioned distance in order to bound the distance of $\stdcocycle$ from the coboundaries. \label{fix-co-cycle-to-coboundary}
\end{enumerate}
We will start by analyzing step~\ref{fix-co-cycle-to-coboundary}.
In order to do so, we will define the following:
\begin{definition}
    Let $\trianglesnorm{f} = \norm{\set{(u,v,w) \in \repcplx{\stdcomplex}(2) \suchthat \cochain{f}(u,v)\cochain{f}(v,w)\cochain{f}(w,u) \ne 1}}$ be the norm of the triangles that, if chosen, cause algorithm~\ref{alg:empty-triangle-test} to reject.\\
    Also let $\emptytrianglesnorm{f} = \sum_{\substack{\set{u,v,w} \in \repcplx{\stdcomplex}(\triangle) \\ \cochain{f}(u,v)\cochain{f}(v,w)\cochain{f}(w,u) \ne 1}}{\weightcplx{\repcplx[k-1]{\stdcomplex}}{(u,v,w)}}$ be the norm of the empty triangles that are rejected by algorithm~\ref{alg:empty-triangle-test}.
\end{definition}
\begin{note*}
    The probability that algorithm~\ref{alg:empty-triangle-test} rejects satisfies:
    \[
        \pr{\text{Algorithm~\ref{alg:empty-triangle-test} rejects}} = 0.5 (\trianglesnorm{f} + \emptytrianglesnorm{f})
    \]
\end{note*}
We will start by considering the norm of the empty triangles of an attachment map.
Specifically we will show that any attachment map satisfies every empty triangle:
\begin{lemma}\label{empty-triangles-of-attachment-map}
    Let $\cochain{\check{\stdcocycle}} \in \cochainset{1}{\repcplx[k-1]{\stdcomplex}; S_l}$ be an attachment map of some cocycle \\
    $\stdcocycle \in \cocycleset{1}{\repcplx{\stdcomplex}; S_l}$ then $\emptytrianglesnorm[k-1]{\cochain{\check{\stdcocycle}}} = 0$.
\end{lemma}
\begin{proof}
    Let $\set{\check{u},\check{v},\check{w}} \in \repcplx[k-1]{\stdcomplex}(\triangle)$ and let $\set{h^c}_{c \in \stdcomplex(k-1)}$ be the choice of local functions such that $\stdcocycle(u,v) = h^{u \cap v}(u)\parens{h^{u \cap v}(v)}^{-1}$ and $\cochain{\check{\stdcocycle}}(\check{u},\check{v}) = \parens{h^{\check{u}}(\check{u} \cup \check{v})}^{-1}h^{\check{v}}(\check{u} \cap \check{v})$.
    Consider:
    \begin{align*}
        \cochain{\check{\stdcocycle}}&(\check{u},\check{v})\cochain{\check{\stdcocycle}}(\check{v},\check{w})\cochain{\check{\stdcocycle}}(\check{w},\check{u}) = \\
        & = \parens{h^{\check{u}}(\check{u} \cup \check{v}))}^{-1} h^{\check{v}}(\check{u} \cup \check{v})\parens{h^{\check{v}}\parens{\check{v} \cup \check{w}}}^{-1}  h^{\check{w}}\parens{\check{v} \cup \check{w}}\parens{h^{\check{w}}\parens{\check{w} \cup \check{u}}}^{-1} h^{\check{u}}\parens{\check{w} \cup \check{u}}
    \end{align*}
    Note that because $\set{\check{u},\check{v},\check{w}}$ is an empty triangle it holds that $\check{u} \cup \check{v} = \check{v} \cup \check{w} =\check{w} \cup \check{u}$ which completes the proof.
\end{proof}

Before we move on to show stronger connections between $\trianglesnorm{\cochain{f}}$, $\emptytrianglesnorm{\cochain{f}}$, $\trianglesnorm[k-1]{\cochain{\check{f}}}$ and $\emptytrianglesnorm[k-1]{\cochain{\check{f}}}$ we would first want to make the connection between the empty triangles that fail the test in the $k$-th dimension and the triangles that fail the test in the $\parens{k-1}$-th dimension.
\begin{lemma}\label{empty-to-full-lemma}
    It holds that:
    \begin{align*}
        \alignset{(\check{u},\check{v},\check{w}) \in \repcplx[k-1]{\stdcomplex}(2)}{ \substack{\cochain{f}(\check{u} \cup \check{v},\check{v} \cup \check{w}) \ne \cochain{g}(\check{u} \cup \check{v},\check{v} \cup \check{w}) \text{ or} \\ \cochain{f}(\check{v} \cup \check{w},\check{w} \cup \check{u}) \ne \cochain{g}(\check{v} \cup \check{w},\check{w} \cup u) \text{ or} \\ \cochain{f}(\check{w} \cup \check{u},\check{u} \cup \check{v}) \ne \cochain{g}(\check{w} \cup \check{u},\check{u} \cup \check{v})}} \subseteq \\
        &\bigcup_{\substack{(u,v) \in \repcplx{\stdcomplex}(1) \\ \cochain{f}(u,v)\ne \cochain{g}(u,v)}}{\bigcup_{\substack{w \in \repcplx{\stdcomplex}(0) \\ (u,v,w) \in \repcplx{\stdcomplex}(\triangle)}}{\set{(u \cap v, v \cap w, w \cap u)}}}
    \end{align*}
\end{lemma}
\begin{proof}
    Let $(\check{u},\check{v},\check{w}) \in \set{(\check{u},\check{v},\check{w}) \in \repcplx[k-1]{\stdcomplex}(2) \suchthat \substack{\cochain{f}(\check{u} \cup \check{v},\check{v} \cup \check{w}) \ne \cochain{g}(\check{u} \cup \check{v},\check{v} \cup \check{w}) \text{ or} \\ \cochain{f}(\check{v} \cup \check{w},\check{w} \cup \check{u}) \ne \cochain{g}(\check{v} \cup \check{w},\check{w} \cup u) \text{ or} \\ \cochain{f}(\check{w} \cup \check{u},\check{u} \cup \check{v}) \ne \cochain{g}(\check{w} \cup \check{u},\check{u} \cup \check{v})}}$.
    Assume WLOG that $\cochain{f}(\check{u} \cup \check{v},\check{v} \cup \check{w}) \ne \cochain{g}(\check{u} \cup \check{v},\check{v} \cup \check{w})$ and let $u = \check{u} \cup \check{v}$, $v = \check{v} \cup \check{w}$ and $w = \check{w} \cup \check{u}$ and note that $\cochain{f}(u,v) \ne \cochain{g}(u,v)$.
    All we have left to prove is that $(u,v,w) \in \repcplx{\stdcomplex}(\triangle)$ and that $(u,v), (v,w), (w,u) \in \repcplx{\stdcomplex}(1)$ but $u \cap v = \check{v}$, $v \cap w = \check{w}$ and $w \cap u = \check{u}$ and therefore:
    \begin{itemize}
        \item $\set{v \cap u, u \cap w, w \cap v} \in \repcplx[k-1]{\stdcomplex}(2)$ and $(u,v,w) \in \repcplx{\stdcomplex}(\triangle)$
        \item It holds that $\abs{u} = \abs{v} = \abs{w} = k$, $\abs{u \cap v} = \abs{v \cap w} = \abs{w \cap u} = k-1$ and $u \cup v = v \cup w = w \cup u = \check{u} \cup \check{v} \cup \check{w} \in \stdcomplex(k+1)$ therefore, by definition, it holds that $(u,v), (v,w), (w,u) \in \repcplx{\stdcomplex}(1)$.
    \end{itemize}
\end{proof}
We will now show connections between $\trianglesnorm{\cochain{f}}$, $\emptytrianglesnorm{\cochain{f}}$, $\trianglesnorm[k-1]{\cochain{\check{f}}}$, and $\emptytrianglesnorm[k-1]{\cochain{\check{f}}}$:
\begin{lemma}\label{cost-of-fixing-to-co-cycle}
    Let $\cochain{f}, \cochain{g} \in \cochainset{1}{\repcplx{\stdcomplex}; S_l}$ be two cochains such that $\dist\parens{\cochain{f}, \cochain{g}} \le \epsilon$ then $\emptytrianglesnorm{g} \le \emptytrianglesnorm{f} + 3 \epsilon$
\end{lemma}
\begin{proof}
    Consider $\emptytrianglesnorm{g}$ and note that:
    \begin{align*}
        \emptytrianglesnorm{g} &= \norm{\set{(\check{u},\check{v},\check{w}) \in \repcplx[k-1]{\stdcomplex}(2) \suchthat {\scriptstyle \cochain{g}(\check{u} \cup \check{v},\check{v} \cup \check{w})\cochain{g}(\check{v} \cup \check{w} ,\check{w} \cup \check{u})\cochain{g}(\check{w} \cup \check{u},\check{u} \cup \check{v}) \ne 1}}}_{\repcplx[k-1]{\stdcomplex}} \le \\
        & \le \multilinenorm{\set{(\check{u},\check{v},\check{w}) \in \repcplx[k-1]{\stdcomplex}(2) \suchthat \substack{\cochain{f}(\check{u} \cup \check{v},\check{v} \cup \check{w})\cochain{f}(\check{v} \cup \check{w},\check{w} \cup \check{u})\cochain{f}(\check{w} \cup \check{u},\check{u} \cup \check{v}) \ne 1 \text{ and}\\ \cochain{f}(\check{u} \cup \check{v},\check{v} \cup \check{w}) = \cochain{g}(\check{u} \cup \check{v}, \check{v} \cup \check{w}) \text{ and} \\ \cochain{f}(\check{v} \cup \check{w},\check{w} \cup \check{u}) = \cochain{g}(\check{v} \cup \check{w},\check{w} \cup \check{u}) \text{ and} \\ \cochain{f}(\check{w} \cup \check{u},\check{u} \cup \check{v}) = \cochain{g}(\check{w} \cup \check{u},\check{u} \cup \check{v})}} \cup}{\cup \set{(\check{u},\check{v},\check{w}) \in \repcplx[k-1]{\stdcomplex}(2) \suchthat \substack{\cochain{f}(\check{u} \cup \check{v},\check{v} \cup \check{w}) \ne \cochain{g}(\check{u} \cup \check{v},\check{v} \cup \check{w}) \text{ or} \\ \cochain{f}(\check{v} \cup \check{w},\check{w} \cup \check{u}) \ne \cochain{g}(\check{v} \cup \check{w},\check{w} \cup \check{u}) \text{ or} \\ \cochain{f}(\check{w} \cup \check{u},\check{u} \cup \check{v}) \ne \cochain{g}(\check{w} \cup \check{u},\check{u} \cup \check{v})}}} \le \\
        & \le \emptytrianglesnorm{f} + \sum_{\substack{(u,v) \in \repcplx{\stdcomplex}(1) \\ \cochain{f}(u,v) \ne \cochain{g}(u,v)}}{\sum_{\substack{w \in \repcplx{\stdcomplex} \\ \set{u,v,w} \in \repcplx{\stdcomplex}(\triangle)}}{\weightcplx{\repcplx[k-1]{\stdcomplex}}{u \cap v, v \cap w, w \cap u}}} = \\
        & = \emptytrianglesnorm{f} + 3\sum_{\substack{(u,v) \in \repcplx{\stdcomplex}(1) \\ \cochain{f}(u,v) \ne \cochain{g}(u,v)}}{\weightcplx{\repcplx{\stdcomplex}}{u,v}} \le \emptytrianglesnorm{f} + 3 \epsilon
    \end{align*}
    The second inequality is due to Lemma~\ref{empty-to-full-lemma} and the last inequality is due to Corollary~\ref{triangles-affected-by-change}.
\end{proof}
\begin{lemma}\label{full-triangles-of-attachment-map}
    Let $\stdcocycle$ be a cocycle of dimension $k$ and let $\check{\stdcocycle}$ be an attachment map of $\stdcocycle$.
    Then it holds that:
    \[
        \emptytrianglesnorm{\stdcocycle} = \trianglesnorm[k-1]{\check{\stdcocycle}}
    \]
\end{lemma}
\begin{proof}
    The proof follows directly from the definition of the attachment map.
    Let $\set{h_c}_{c \in \stdcomplex(k-1)}$ be the functions from the definition of the attachment map:
    \begin{align*}
        \trianglesnorm[k-1]{\check{\stdcocycle}} & = \norm{\set{(u,v,w) \in \repcplx[k-1]{\stdcomplex}(2) \suchthat \check{\stdcocycle}(w \cap u, u \cap v)\check{\stdcocycle}(u \cap v, v \cap w)\check{\stdcocycle}(v \cap w, w \cap u) \ne 1}} = \\
        & = \norm{\set{(u,v,w) \in \repcplx[k-1]{\stdcomplex}(2) \suchthat {\scriptstyle \parens{h^{w \cap u}(u)}^{-1}h^{u \cap v}(u)\parens{h^{u \cap v}(v)}^{-1}h^{v \cap w}(v)\parens{h^{v \cap w}(w)}^{-1}h^{w \cap u}(w) \ne 1}}} = \\
        & = \norm{\set{(u,v,w) \in \repcplx[k-1]{\stdcomplex}(2) \suchthat {\scriptstyle h^{u \cap v}(u)\parens{h^{u \cap v}(v)}^{-1}h^{v \cap w}(v)\parens{h^{v \cap w}(w)}^{-1}h^{w \cap u}(w)\parens{h^{w \cap u}(u)}^{-1} \ne 1}}} = \\
        & = \norm{\set{(u,v,w) \in \repcplx[k-1]{\stdcomplex}(2) \suchthat \stdcocycle(u,v)\stdcocycle(v,w)\stdcocycle(w,u) \ne 1}} = \\
        & = \sum_{\substack{\set{u,v,w} \in \repcplx{\stdcomplex}(\triangle) \\ \stdcocycle(u,v)\stdcocycle(v,w)\stdcocycle(w,u) \ne 1}}{\weightcplx{\repcplx[k-1]{\stdcomplex}}{(u,v,w)}} = \emptytrianglesnorm{f}
    \end{align*}
\end{proof}
All we have left to do is show how to combine Lemma~\ref{cost-of-fixing-to-co-cycle} and Lemma~\ref{choice-of-cores} in order to prove that algorithm~\ref{alg:empty-triangle-test} is indeed a test for the coboundaries in $\repcplx{\stdcomplex}$:
\begin{lemma}\label{e_k-bounds-the-distance}
    For every $k$ there exists a linear function $e_k$ such that for every cochain $\cochain{f}$:
    \[
        \dist\parens{\cochain{f}, \coboundaryset{1}{\repcplx{\stdcomplex}}} \le e_k\parens{\trianglesnorm{f}, \emptytrianglesnorm{f}}
    \]
\end{lemma}
\begin{proof}
    We will prove this Lemma by induction on the dimension of the representation.
    Note that $\repcplx[0]{\stdcomplex} = \stdcomplex$.
    Let $\cochain{f} \in \cochainset{1}{\repcplx[0]{\stdcomplex}; S_l} = \cochainset{1}{\stdcomplex; S_l}$.
    Due to the assumption that $\stdcomplex$ is a $\gamma$-coboundary expander there exists a coboundary $\cochain{\widetilde{F}}$ such that $\dist\parens{\cochain{f}, \cochain{\widetilde{f}}} \le \frac{1}{\gamma} \trianglesnorm[0]{\cochain{f}}$.

    Assuming that the claim holds for the representation complex of dimension $k-1$, let $\cochain{f} \in \cochainset{1}{\repcplx{\stdcomplex}; S_l}$ be a cochain.

    We start by showing that $\cochain{f}$ is close to a cocycle:
    Due to Lemma~\ref{representation-complex-expands-in-first-dimension} the representation complex is expanding (with the same expansion parameter as the original complex) and therefore there exists a cocycle $\stdcocycle$ such that $\dist\parens{\cochain{f}, \stdcocycle} \le \frac{1}{\gamma} \trianglesnorm{f}$.
    In addition, due to Lemma~\ref{triangles-affected-by-change} it holds that $\emptytrianglesnorm{\stdcocycle} \le \emptytrianglesnorm{f} + \frac{3}{\gamma} \trianglesnorm{f}$.

    Then we show that this cocycle is close to a coboundary:
    Consider an attachment map $\stdcocycle'$ of $\stdcocycle$.
    Due to Lemma~\ref{empty-triangles-of-attachment-map} it holds that $\emptytrianglesnorm[k-1]{\check{\stdcocycle}} = 0$ and due to Lemma~\ref{full-triangles-of-attachment-map} $\trianglesnorm[k-1]{\check{\stdcocycle}} \le \emptytrianglesnorm{\stdcocycle} \le \emptytrianglesnorm{f} + \frac{3}{\gamma} \trianglesnorm{f}$.
    Due to our assumption it holds that $\dist\parens{\stdcocycle, \coboundaryset{1}{\repcplx[k-1]{\stdcomplex;S_l}}} \le e(\trianglesnorm[k-1]{\check{\stdcocycle}}, \emptytrianglesnorm[k-1]{\check{\stdcocycle}})$.
    We finish this part of the proof by using Lemma~\ref{choice-of-cores} in order to prove that there exists a coboundary $\widetilde{\stdcocycle}$ such that:
    \[
        \dist\parens{\stdcocycle, \widetilde{\stdcocycle}} \le 2 \frac{k}{k-1} e(\trianglesnorm[k-1]{\check{\stdcocycle}}, \emptytrianglesnorm[k-1]{\check{\stdcocycle}}) = 2 \frac{k}{k+1} \cdot e\parens{\emptytrianglesnorm{f} + \frac{3}{\gamma} \trianglesnorm{f}, 0}
    \]

    All we have left is to use the triangle inequality in order to bound the distance of $\cochain{f}$ from $\widetilde{\stdcocycle}$:
    \[
        \dist\parens{\cochain{f}, \widetilde{\stdcocycle}} \le \dist\parens{\cochain{f}, \stdcocycle} + \dist\parens{\stdcocycle, \widetilde{\stdcocycle}} \le \frac{1}{\gamma} \trianglesnorm{f} + 2 \frac{k}{k+1} \cdot e\parens{\emptytrianglesnorm{f} + \frac{3}{\gamma} \trianglesnorm{f}, 0}
    \]
    Note that if $e_{k-1}$ is linear then so is $e_k$
\end{proof}
We move on to finding $e_k$ explicitly
\begin{lemma}\label{explicit-e_k}
    For every $k$ it holds that:
    \[
        e_k\parens{\varepsilon_\blacktriangle, \varepsilon_\triangle} = \frac{1}{\gamma}\parens{\sum_{i=1}^{k}{\frac{k+2-i}{k+1}\parens{\frac{6}{\gamma}}^{i-1}}}\varepsilon_\blacktriangle + \frac{2}{\gamma}\parens{\sum_{i=1}^{k-1}{\frac{k+1-i}{k+1}\parens{\frac{6}{\gamma}}^{i-1}}} \epsilon_{\triangle}
    \]
\end{lemma}
\begin{proof}
    In order to present $e_k$ explicitly we consider the following properties:
    \begin{enumerate}
        \item $\forall k: e_k(\varepsilon_\blacktriangle, \varepsilon_\triangle) = \frac{\varepsilon_\blacktriangle}{\gamma} + 2 \frac{k}{k+1} \cdot e_{k-1}\parens{\varepsilon_\triangle + \frac{3 \varepsilon_\blacktriangle}{\gamma}, 0}$ \label{e_k-rec-step}
        \item $e_0(\varepsilon_\blacktriangle, \varepsilon_\triangle) = \frac{\varepsilon_\blacktriangle}{\gamma}$ \label{e_k-termintion-condition}
        \item $\forall k: e_k(\varepsilon_\blacktriangle, \varepsilon_\triangle) = \frac{\varepsilon_\blacktriangle}{\gamma} + e_{k}\parens{0, \varepsilon_\triangle + \frac{3 \varepsilon_\blacktriangle}{\gamma}}$ \label{e_k-ignore-right-side}
    \end{enumerate}
    Note that the first two properties were proven in Lemma~\ref{e_k-bounds-the-distance} and property~\ref{e_k-ignore-right-side} is, in essence, a restatement of Lemma~\ref{cost-of-fixing-to-co-cycle}\footnote{The inequality Lemma~\ref{cost-of-fixing-to-co-cycle} bounds the distance from above and since we are looking for a bound from above we can consider the case where $e_k$ equals its upper bound.}.
    We will begin with analyzing how $e_k$ behaves under the assumption that $\epsilon_\triangle=0$.
    We will then use property~\ref{e_k-ignore-right-side} to find $e_k$ explicitly.

    Let us begin by showing that:
    \begin{equation}\label{eq:e_k-explicit-left-side}
        e_k(\varepsilon_\blacktriangle, 0) = \frac{1}{\gamma}\parens{\sum_{i=1}^{k}{\frac{k+2-i}{k+1}\parens{\frac{6}{\gamma}}^{i-1}}}\varepsilon_\blacktriangle
    \end{equation}
    We will prove this by induction on $k$.
    Note that when $k=0$ it is easy to see that:
    \[
        e_0(\varepsilon_\blacktriangle, 0) =  \frac{\varepsilon_\blacktriangle}{\gamma}
    \]
    Assume that equation~\ref{eq:e_k-explicit-left-side} holds for $k-1$ and consider the following:
    \begin{align*}
        e_k(\varepsilon_\blacktriangle, 0) & = \frac{\varepsilon_\blacktriangle}{\gamma} + 2 \frac{k}{k+1} \cdot e_{k-1}\parens{\frac{3 \varepsilon_\blacktriangle}{\gamma}, 0} = \frac{\varepsilon_\blacktriangle}{\gamma} + 2 \frac{k}{k+1} \cdot \frac{1}{\gamma}\parens{\sum_{i=1}^{k-1}{\frac{k+1-i}{k}\parens{\frac{6}{\gamma}}^{i-1}}} \frac{3 \varepsilon_\blacktriangle}{\gamma} = \\
        & = \frac{\varepsilon_\blacktriangle}{\gamma} + \frac{6}{\gamma}\parens{\sum_{i=1}^{k-1}{\frac{k+1-i}{k+1}\parens{\frac{6}{\gamma}}^{i-1}}} \frac{\varepsilon_\blacktriangle}{\gamma} = \frac{\varepsilon_\blacktriangle}{\gamma} + \parens{\sum_{i=1}^{k-1}{\frac{k+2-(i+1)}{k+1}\parens{\frac{6}{\gamma}}^{i}}} \frac{\varepsilon_\blacktriangle}{\gamma} = \\
        & = \frac{\varepsilon_\blacktriangle}{\gamma} + \parens{\sum_{i=2}^{k}{\frac{k+2-i}{k+1}\parens{\frac{6}{\gamma}}^{i-1}}} \frac{\varepsilon_\blacktriangle}{\gamma} =\parens{\sum_{i=2}^{k}{\frac{k+2-i}{k+1}\parens{\frac{6}{\gamma}}^{i-1}}+1}\frac{\varepsilon_\blacktriangle}{\gamma} = \\
        & = \parens{\sum_{i=1}^{k}{\frac{k+2-i}{k+1}\parens{\frac{6}{\gamma}}^{i-1}}}\frac{\varepsilon_\blacktriangle}{\gamma}
    \end{align*}
    Now all we have left is to use property~\ref{e_k-rec-step} and property~\ref{e_k-ignore-right-side} to find $e_k$ explicitly:
    \begin{align*}
        e_k(\varepsilon_\blacktriangle, \varepsilon_\triangle) & = \frac{\varepsilon_\blacktriangle}{\gamma} + 2 \frac{k}{k+1} \cdot e_{k-1}\parens{\varepsilon_\triangle + \frac{3 \varepsilon_\blacktriangle}{\gamma}, 0} = \\
        & = \frac{\varepsilon_\blacktriangle}{\gamma} + 2 \frac{k}{k+1} \frac{1}{\gamma}\parens{\sum_{i=1}^{k-1}{\frac{k+1-i}{k}\parens{\frac{6}{\gamma}}^{i-1}}}\parens{\varepsilon_\triangle + \frac{3 \varepsilon_\blacktriangle}{\gamma}} =\\
        & = \frac{\varepsilon_\blacktriangle}{\gamma} + \frac{2}{\gamma}\parens{\sum_{i=1}^{k-1}{\frac{k+1-i}{k+1}\parens{\frac{6}{\gamma}}^{i-1}}}\parens{\varepsilon_\triangle + \frac{3 \varepsilon_\blacktriangle}{\gamma}} = \\
        & = \frac{1}{\gamma}\parens{\sum_{i=1}^{k}{\frac{k+2-i}{k+1}\parens{\frac{6}{\gamma}}^{i-1}}}\varepsilon_\blacktriangle + \frac{2}{\gamma}\parens{\sum_{i=1}^{k-1}{\frac{k+1-i}{k+1}\parens{\frac{6}{\gamma}}^{i-1}}} \epsilon_{\triangle}
    \end{align*}
\end{proof}
We can now finally prove that the empty triangle test is indeed a test for the property that a cochain $\cochain{f}$ is a coboundary:
\begin{lemma}\label{empty-triangle-test-tests-co-boundaries}
    Algorithm~\ref{alg:empty-triangle-test} is a test for whether a cochain $\cochain{f}$ is a coboundary that performs exactly three queries and, if it rejects $\cochain{f}$ with probability of $\varepsilon(f)$ then:
\[
    \dist\parens{\cochain{f}, \coboundaryset{1}{\repcplx{\stdcomplex}}} \le 2\parens{\frac{1}{\gamma}\parens{\sum_{i=1}^{k}{\frac{k+2-i}{k+1}\parens{\frac{6}{\gamma}}^{i-1}}} + \frac{2}{\gamma}\parens{\sum_{i=1}^{k-1}{\frac{k+1-i}{k+1}\parens{\frac{6}{\gamma}}^{i-1}}} } \varepsilon(\cochain{f})
\]
\end{lemma}
\begin{proof}
    Denote the probability that Algorithm~\ref{alg:empty-triangle-test} rejects $\cochain{F}$ as $\varepsilon(\cochain{F})$ and consider the rejection probability of Algorithm~\ref{alg:empty-triangle-test} and note that:
    \[
        \varepsilon(\cochain{F}) = 0.5\trianglesnorm{f} + 0.5\emptytrianglesnorm{f}
    \]
    Therefore $\trianglesnorm{f} \le 2\varepsilon(\cochain{F})$ and $\emptytrianglesnorm{f} \le 2\varepsilon(\cochain{F})$.
    Using Lemma~\ref{explicit-e_k} and Lemma~\ref{e_k-bounds-the-distance} we conclude that:
    \[
        \dist\parens{\cochain{f}, \coboundaryset{1}{\repcplx{\stdcomplex}}} \le 2\parens{\frac{1}{\gamma}\parens{\sum_{i=1}^{k}{\frac{k+2-i}{k+1}\parens{\frac{6}{\gamma}}^{i-1}}} + \frac{2}{\gamma}\parens{\sum_{i=1}^{k-1}{\frac{k+1-i}{k+1}\parens{\frac{6}{\gamma}}^{i-1}}} } \varepsilon(\cochain{f})
    \]
\end{proof}
Note that Lemma~\ref{co-boundaries-in-representation-complex-are-testable} and Lemma~\ref{empty-triangle-test-tests-co-boundaries} are essentially the same.
\section{Lower Bound on the Number of Queries}\label{sec:lower-bound-on-the-number-of-queries}

In this section we will show that every test that tests list agreement has to perform at least $l-1$ queries.
We will do so by finding a non-agreeing $l$-assignment $\lassignment{r}$ such that for any $l-1$ queries there exists an agreeing $l$-assignment $\lassignment{a}$ such that the values returned by the queries are the same when querying $\lassignment{a}$ and $\lassignment{r}$.
Note that this is enough due to Lemma~\ref{tests-fail}.
\begin{definition}
    Let $\cochain{f}_0, \dots, \cochain{f}_{l-1}$ be cochains and let $\genface \in \stdcomplex(k)$ be a face.
    Also let $\cochain{f}_{l}$ be a cochain such that:
    \[
        \exists \genface' \in \binom{\genface}{k} \forall i \in [l] \exists \stdvertex \in \genface': \cochain{f}_i(\stdvertex) \ne \cochain{f}_l(\stdvertex)
    \]
    Define the $l$-assignment $\lassignment{r}$ to be the $l$-assignment that if $i \ne l$ or $\stdface \ne \genface$ then $\lassignment{r}^\stdface_i = \cochain{f}_i|_\stdface$ and $\lassignment{r}^\genface_l = \cochain{f}_l|_\genface$
\end{definition}
All that is left is to show that $\lassignment{r}$ satisfies the requirements of Lemma~\ref{tests-fail}.
\begin{lemma}
    $\lassignment{r}$ is not an agreeing $l$-assignment.
\end{lemma}
\begin{proof}
    Let $\tilde{\genface} \in \stdcomplex(k)$ such that $\tilde{\genface} \cap \genface = \genface'$.
    Note that for every $i$ it holds that:
    \begin{equation}
        \cochain{f}_l|_{\genface \cap \tilde{\genface}} \ne \cochain{f}_i|_{\genface \cap \tilde{\genface}}\label{eq:local-differance}
    \end{equation}
    If $\lassignment{r}$ is an agreeing $l$-assignment then there are cochains $\cochain{f}'_0,\dots, \cochain{f}'_{l-1}$ and $\pi_{\tilde{\genface}}, \pi_\genface$ such that:
    \[
        \forall i: \cochain{f}'_i|_{\genface} = \lassignment{r}^{\genface}_{\pi_\genface(i)} \text{ and } \cochain{f}'_i|_{\tilde{\genface}} = \lassignment{r}^{\tilde{\genface}}_{\pi_{\tilde{\genface}}(i)}
    \]
    Therefore:
    \[
        \cochain{f}_l|_{\genface \cap \tilde{\genface}} = \lassignment{r}^{\genface}_{l-1}|_{\genface \cap \tilde{\genface}} = \cochain{f'}_{\pi_\genface^{-1}(l-1)}|_{\genface \cap \tilde{\genface}} = \lassignment{r}^{\tilde{\genface}}_{\pi_{\tilde{\genface}(\pi_\genface^{-1}(l-1))}}|_{\genface \cap \tilde{\genface}} = \cochain{f}_{\pi_{\tilde{\genface}}(\pi_\genface^{-1}(l-1))}|_{\genface \cap \tilde{\genface}}
    \]
    Which contradicts~\ref{eq:local-differance}.
\end{proof}
\begin{lemma}
    For every set of queries $q_{1},\dots,q_{l-1}$ to $\lassignment{r}$ there exists an agreeing $l$-assignment $\lassignment{a}_{\parens{q_{1},\dots,q_{l-1}}}$ such that when querying $\lassignment{a}_{\parens{q_{1},\dots,q_{l-1}}}$ with $q_{1},\dots,q_{l-1}$ the results are the same as the results on the original assignment.
\end{lemma}
\begin{proof}
    Note that the queries are of the form $\parens{\stdface_i,a_i}$ where $\stdface_i \in \stdcomplex(k)$ and $a_i \in \sparens{l}$.
    Let $a' \in \sparens{l}$ such that $a' \ne a_i$ for any $i$.
    If $a' = l-1$ then ${\lassignment{a}_{\parens{q_{1},\dots,q_{l-1}}}}_i^\stdface = \cochain{f}_i|_\stdface$ (which is an agreeing assignment).
    Otherwise define $\lassignment{a}_{\parens{q_{1},\dots,q_{l-1}}}$ in the following way:
    \begin{itemize}
        \item If $i \ne a'$ set ${\lassignment{a}_{\parens{q_{1},\dots,q_{l-1}}}}_{i}^\genface = \lassignment{f}^\genface_{i}$
        \item Set ${\lassignment{a}_{\parens{q_{1},\dots,q_{l-1}}}}_{a'}^\genface = \lassignment{f}^\genface_{l-1}$
        \item If $\stdface \ne \genface$ set ${\lassignment{a}_{\parens{q_{1},\dots,q_{l-1}}}}_{a'}^\stdface = {\cochain{f}_l}|_\stdface$
    \end{itemize}
    And note that $\lassignment{a}_{\parens{q_{1},\dots,q_{l-1}}}$ agrees with $\cochain{f}_1,\dots, \cochain{f}_{a'-1}, \cochain{f}_{a'+1},\dots,\cochain{f}_{l-1}, \cochain{f}_l$ which completes the proof.
    Finally note that for every query $\parens{\stdface_i,a_i}$ it holds that $\lassignment{a}_{\parens{q_{1},\dots,q_{l-1}}}[\parens{\stdface_i,a_i}] = \lassignment{r}[\parens{\stdface_i,a_i}]$ since $\lassignment{r}$ and $\lassignment{a}_{\parens{q_{1},\dots,q_{l-1}}}$ differ only when $a_i = a'$ which is never queried.
\end{proof}
\section{Combinatorial Bounds}\label{sec:combinatorial-bounds}

\begin{lemma}
    \label{combinatorial-bounds-1}
    $\binom{d-k+1}{i}\binom{d+1}{k} = \binom{k+i+1}{k}\binom{d+1}{k+i+1}$
\end{lemma}
\begin{proof}
    \begin{align*}
        \binom{d-k+1}{i}\binom{d+1}{k}
        & = \frac{(d-k+1)!}{(i+1)!(d-k-i)!}\frac{(d+1)!}{k!(d-k+1)!} \\
        & = \frac{(k+i+1)!}{(i+1)!(d-k-i)!}\frac{(d+1)!}{k!(k+i+1)!} \\
        & = \frac{(k+i+1)!}{(i+1)!k!}\frac{(d+1)!}{(k+i+1)!(d-k-i)!} \\
        & = \binom{k+i+1}{k}\binom{d+1}{k+i+1}
    \end{align*}
\end{proof}
    \end{appendices}
\end{document}